\documentclass[final,leqno]{siamltex1213}
\usepackage{graphicx,graphics,epsfig,epstopdf}
\usepackage{amsfonts,amssymb,mathrsfs,amsmath}
\usepackage{hyperref}
\usepackage{ifpdf}
\def\l{\langle}
\def\r{\rangle}

\newcommand{\mbb}{\mathbb}
\newcommand{\mrm}{\mathrm}
\newcommand{\msf}{\mathsf}

\title{Continuous-mode multi-photon filtering\thanks{This research is supported in part by National Natural Science Foundation of China (NSFC) grants (Nos. 61374057, 61134008 and 61227902) and a Hong Kong RGC grant (No. 531213).}}

\author{H.~T. Song\footnotemark[2]
\and G.~F. Zhang\footnotemark[3]
\and Z.~R. Xi\footnotemark[4]}

\begin{document}
\maketitle

\renewcommand{\thefootnote}{\fnsymbol{footnote}}
\footnotetext[2]{Qian Xuesen Laboratory of Space Technology, China Academy of Space Technology, Beijing 100094, China (\email{songhongting@qxslab.cn}).}
\footnotetext[3]{Department of Applied Mathematics, Hong Kong Polytechnic University, Hong Kong, China (\email{Guofeng.Zhang@polyu.edu.hk}).}
\footnotetext[4]{Academy of Mathematics and Systems Science, Chinese Academy of Sciences, Beijing 100190, China (\email{zrxi@iss.ac.cn}).}


\begin{abstract}
The purpose of this paper is to derive filters for an arbitrary open quantum system driven by a light wavepacket prepared in a continuous-mode multi-photon state.  A continuous-mode multi-photon state is a state of a travelling light wavepacket that contains a definite number of photons and is characterised by a temporal (or equivalently spectral) profile.  After the interaction with the system, the outgoing light can be monitored by means of homodyne detection or photodetection. Filters for both measurement schemes are derived in this paper. Unlike the vacuum or the coherent state case, the annihilation operator of the light field acting on a multi-photon state changes the state by annihilating a photon, and this makes the traditional filtering techniques inapplicable. To circumvent this difficulty, we adopt a non-Markovian embedding technique proposed in \cite{gough:2013} for the study of the single-photon filtering problem. However, the multi-photon nature of the problem addressed in this paper makes the study much more mathematically involved. Moreover, as demonstrated by an example  --- a two-level system driven by a continuous-mode two-photon state, multi-photon filters can reveal interesting strong nonlinear optical phenomena absent in both the single-photon state case and the continuous-mode Fock state case. \end{abstract}



\pagestyle{myheadings}
\thispagestyle{plain}
\markboth{Continuous-mode multi-photon filtering}{Continuous-mode multi-photon filtering}

\section{Introduction}\label{intro}

When light impinges on a quantum system, e.g., an atom or a quantum-mechanical oscillator, partial system information may be carried away by the outgoing light. The outgoing light can be directed to another quantum system, thus serving as a (directional) link to facilitate cascade connection \cite{carmichael:1993}-\cite{ZJ12}. Alternatively, the outgoing light may be continuously monitored to produce photocurrent, on which the state of the quantum system can be conditioned. The stochastic evolution of the conditional system state is commonly called a {\it quantum trajectory}. A quantum filter can be designed to estimate quantum trajectories \cite{carmichael:1993}, \cite{HA13}-\cite{APR14}; this is the basis of quantum measurement-based feedback control \cite{vHSM05}-\cite{james:2007}, \cite{wiseman:1993}-\cite{wiseman:2010}.

In quantum optics, the familiar formalism of quantum filtering considers incident lights in Gaussian states, including the vacuum state, coherent states, thermal states, and squeezed states \cite{wiseman:2010}-\cite{EWP15}.  This is natural as Gaussian states are commonly used in quantum optics laboratories and have been well studied. With the advent of modern experimental technology, nowadays, non-Gaussian states such as single-photon states, multi-photon states and Schrodinger cat states, can be reliably generated and manipulated. Therefore, very recently, there is a growing interest in deriving quantum filters for such non-Gaussian states. For example, filters have been derived for quantum systems driven by light fields prepared in single-photon states or cat states \cite{gough:2013,carvalho:2012,gough:2012b}.

Single-photon states and multi-photon states are very useful resources in quantum computing, quantum communication, and quantum cryptography \cite{loudon:2000}-\cite{ANL+14} and references therein. Roughly speaking, a continuous-mode $n$-photon state is a state of a travelling light wavepacket containing exactly $n$ photons that share a common pulse shape superposed on a continuum of spectral modes.  When $n=0$, the wavepacket is in the vacuum state, whose filtering equation is well-known in the quantum optics community. When $n=1$, the wavepacket is in a single-photon state, whose filtering equations have recently been derived in \cite{gough:2013,gough:2012b}. When $n>1$, for convenience we call the state a {\it multi-photon state}.

In this paper, we study the problem of quantum filtering for arbitrary quantum systems driven by continuous-mode multi-photon states. Due to the multi-photon nature of the problem, it turns out that the derivation of multi-photon filters is very mathematically involved. For example, if the input is a wavepacket containing $n$ photons, we need a hierarchy of  $\frac{2^n (2^n+1)}{2}$ differential equations to determine the $n$-photon filter.  When $n=0$, namely the vacuum state case, a single differential equation is sufficient. When $n=1$, namely the single-photon state case, we need 3 coupled differential equations. When $n=2$, a system of 10 equations is required.  Similarly, a hierarchy of 36 differential equations is required for the case of  $3$-photon state, and so on.  Therefore, to present the main ideas clearly, we investigate the 2-photon case in detail before proceeding to the general $n$-photon case.

For the $2$-photon case, master equations are given in Theorem \ref{thm:master_Heisenberg} and Corollary \ref{cor:master_Heisenberg}, while quantum filters for homodyne detection are given in Theorem \ref{thm:2-photon_filter} and Corollary \ref{cor:filter_up}.   These results contain those in \cite{gough:2013} and \cite{gough:2012b} for the single-photon case as special cases. Numerical studies conducted in Examples 2 and 3 show that two-photon excitation of a two-level atom has highly nonlinear optical phenomena absent in both the single-photon state case and the continuous-mode Fock state case.  For the general $n$-photon case,  quantum filters are given in Theorem \ref{homodyne_n_fiter_general} for the homodyne detection case.  For photodetection, the quantum filter is given in Theorem \ref{thm:n_filter_PD}, which reduces to the single-photon filter for photodetection when $n=1$ which is studied in \cite{gough:2013,gough:2012b}. Finally, the multi-photon master equations are given in Theorem \ref{master_n}, which in the Fock state case is actually the master equation (20) in \cite{baragiola:2012} for continuous-mode Fock states. Therefore, the results presented in this paper are indeed very general. Due to the multi-photon nature, the mathematical description of general multi-photon filtering equations  are very messy; in fact, the lexicographical ordering \cite{Skiena:1990} plays an essential role.


The rest of the paper is organized in the following way. Section \ref{sec:preliminaries} introduces open quantum systems and poses the filtering problem. Section \ref{sec:two_photon} focuses on the two-photon case. Here, we first define two-photon states in Subsection \ref{subsec:state}, then present the master equations in Subsection \ref{subsec:2 photon mater eqn}. In order to derive the two-photon filtering equation, we define an extended system in Subsection \ref{subsec:extended_system}, and derive the filtering equation for this extended system in Subsection \ref{subsec:filter:extended}, based on which in Subsection \ref{subsec:filter:2photon} we derive the quantum filter for the original system driven by a two-photon state. After the study of the two-photon filtering in Section \ref{sec:two_photon}, we proceed to the general multi-photon case in Section \ref{sec:multi_photon_filtering}, where we present the general filtering equations for both the homodyne detection case and the photon-counting case.  Section \ref{sec:conclusion} concludes the paper.

{\it Notation.}  $|0\rangle$ is the vacuum state of the free field.  $|\eta\r$ is the initial state of the quantum system of interest. $\mbb{R}^+$ is the set of non-negative real numbers, $L_2(\mbb{R}^+,\mbb{C})$ is the space of Lebesgue measurable and square integrable functions from $\mbb{R}^+$ to $\mbb{C}$. For $\xi_1$ and $\xi_2 \in L_2(\mbb{R}^+, \mbb{C})$, their inner product is $\l \xi_1 | \xi_2 \rangle \triangleq \int_0^\infty \xi_1^\ast(t) \xi_2(t){{\rm d}} t$.  The norm of a  function $\xi\in L_2(\mbb{R}^+,\mbb{C})$ is $\|\xi\| \triangleq \sqrt{\l \xi|\xi\r}$.  $\delta_{jk}$ is the Kronecker delta, namely, $\delta_{jk} = 1$ if $j=k$ or 0 otherwise. $\otimes$ stands for tensor product. $X^\ast$ denotes the complex conjugate of $X$ if $X$ is a complex number or the adjoint operator of $X$ if $X$ is an operator. The commutator of two operators $A$ and $B$  is defined to be $[A,B]\triangleq AB-BA$.

\section{Preliminaries}\label{sec:preliminaries}

In this section we briefly introduce quantum systems and pose the multi-photon filtering problem.

\subsection{Quantum systems} \label{subsec:systems_states}

This subsection gives a very brief introduction to quantum systems, more details can be found in, e.g., \cite{gardiner:2000}, \cite{ZJ12},  \cite{james:2007},   \cite{wiseman:2010}, \cite{hudson:1984}-\cite{WM08}.

The model we study is an arbitrary quantum system $G$ driven by a single-channel
light field which can be effectively described by the so-called $(S,L,H)$ language \cite{GJ09,ZJ12}. Here, $S$ is a unitary scattering operator, $L$ is a coupling operator that describes how the system is coupled to the input field, and the self-adjoint operator $H$ is the initial system Hamiltonian. $S$, $L$, and $H$ are system operators on a separable Hilbert space $\mathsf{H}_{S}$ where the system states reside. The single-channel light field has an annihilation operator $b(t)$ and a creation operator $b^\ast(t)$, which are operators on a Fock space $\mathsf{H}_{F}$ (an infinite-dimensional Hilbert space). $B(t)\triangleq\int_0^t b(r){{\rm d}} r$ and $B^\ast(t)\triangleq\int_0^tb^\ast(r){{\rm d}} r$ are integrated annihilation and creation field operators respectively. The gauge process, often called counting process,
$\Lambda(t)\triangleq\int_0^t b^*(\tau)b(\tau){{\rm d}}\tau$
is also an integrated operator on the Fock space $\mathsf{H}_{F}$ for the input field. In this paper, the input field is \emph{canonical}, that is, the non-zero Ito products are
\[
{{\rm d}} B(t){{\rm d}} B^{\ast }(t)={{\rm d}} t, {{\rm d}} B(t){{\rm d}}\Lambda (t)={{\rm d}} B(t), {{\rm d}}\Lambda (t){{\rm d}}\Lambda(t)={{\rm d}}\Lambda (t),  {{\rm d}} \Lambda (t){{\rm d}} B^{\ast}(t)={{\rm d}} B^{\ast}(t).
\]

The temporal evolution of the composite system composed of the system and the field can be described by a unitary operator $U(t)$ on the tensor product Hilbert space $\mathsf{H}_{S}\otimes \mathsf{H}_{F}$, and is given by the the following Hudson-Parthasarathy (HP) quantum stochastic differential equation (QSDE)
\begin{equation}
{{\rm d}} U(t)=\left\{ (S-I){{\rm d}}\Lambda (t)+L{{\rm d}} B^{\ast }(t)-L^{\ast }S{{\rm d}} B(t)-(\frac{1}{2}
L^{\ast }L+{\rm i} H){{\rm d}} t\right\} U(t), ~~ t > 0\nonumber
\end{equation}
with the initial condition $U(0)=I$ (the identity operator) and ${\rm i}=\sqrt{-1}$.

In Heisenberg picture, the system operator $X$ at time $t$ is given by $X(t)\equiv j_t(X)\triangleq U^*(t)(X\otimes I)U(t)$, which is an operator on  $\mathsf{H}_{S}\otimes \mathsf{H}_{F}$, and whose temporal evolution is governed by the following Heisenberg equation of motion
\begin{equation}\label{eq:QDSE_X}
{{\rm d}} j_t(X)\!=\!j_t(\mathcal{L}_{00}(X)){{\rm d}} t+j_t(\mathcal{L}_{01}(X)){{\rm d}} B(t)+j_t(\mathcal{L}_{10}(X)){{\rm d}} B^*(t) +j_t(\mathcal{L}_{11}(X)){{\rm d}} \Lambda(t) ,
\end{equation}
where the Evans-Parthasarathy superoperators are
\begin{eqnarray}
\mathcal{L}_{00}(X)
&\triangleq&
\frac{1}{2}L^*[X,L]+\frac{1}{2}[L^*,X]L-{\rm i}[X,H],
\label{eq:L_00}
\\
\mathcal{L}_{01}(X)
&\triangleq&
[L^*,X]S,
\label{eq:L_01}
\\
\mathcal{L}_{10}(X)
&\triangleq&
S^*[X,L]=(\mathcal{L}_{01}(X^*))^*,
\label{eq:L_10}
\\
\mathcal{L}_{11}(X)
&\triangleq&
S^*XS-X.
\label{eq:L_11}
\end{eqnarray} 
After interaction, the quantum field becomes $B_{\rm out}(t) \triangleq U^{\ast }(t)(I\otimes B(t))U(t)$, an operator on  $\mathsf{H}_{S}\otimes \mathsf{H}_{F}$, whose dynamics are given by the following QSDE
\begin{equation}
{{\rm d}} B_{\rm out}(t)=j_{t}(L){{\rm d}} t+j_{t}(S){{\rm d}} B(t).\nonumber
\end{equation}
The output field can be monitored. Homodyne detection and photodetection are the two most commonly used measurement methods in quantum optics. In homodyne detection, the noise quadrature
\begin{eqnarray*}
Y(t) \triangleq U^{\ast }(t)(I\otimes (B(t)+ B^{\ast }(t)))U(t)=B_{\rm out}(t)+B_{\rm out}^{\ast}(t)
\end{eqnarray*}
may be measured, while in photodetection (photon counting),
\begin{eqnarray*}
Y^{\Lambda }(t)
\triangleq
U^{\ast }(t)(I\otimes \Lambda(t))U(t)
\end{eqnarray*}
is measured. By Ito rules, the observation processes $Y(t)$ and $Y^{\Lambda }(t)$ satisfy
\begin{eqnarray*}
{{\rm d}} Y(t)
=
j_{t}(S^{\ast }){{\rm d}} B^{\ast }(t)+j_{t}(S){{\rm d}} B(t)+j_{t}(L+L^{\ast }){{\rm d}} t,
\label{eq:Y_W}
\end{eqnarray*}
and
\begin{eqnarray*}
{{\rm d}} Y^{\Lambda }(t)
=
{{\rm d}} \Lambda (t)+j_{t}(S^{\ast }L){{\rm d}} B^{\ast }(t)+j_{t}(L^{\ast}S){{\rm d}} B(t)+j_{t}(L^{\ast }L){{\rm d}} t,
 \label{eq:Y_Lambda}
\end{eqnarray*}
respectively. Moreover, $Y(t)$ and $Y^{\Lambda }(t)$  obey the so-called {\it self-nondemolition} property,  i.e.,
\begin{equation}
[Y(t),Y(s)]=[Y^{\Lambda}(t),Y^{\Lambda }(s)]=0, ~~~ 0\leq s \leq t.\nonumber
\end{equation}
 We denote by $\mathscr{Y}(t)$ and $\mathscr{Y}^{\Lambda }(t)$  the commutative von Neumann algebras generated by  $\{Y(s);~ 0\leq s\leq t\}$ and $\{ Y^{\Lambda }(s);~ 0\leq s\leq t \}$, respectively.


\subsection{Quantum Filtering}\label{subsec:filtering problem}


Simply speaking, the quantum filtering problem studied in this paper is about finding a least mean-square estimate of system observables $j_t(X)$ based on the past measurement outcome information up to time $t$ for a quantum system driven by a continuous-mode multi-photon state. In the homodyne detection case,  it is about the computation of the quantum conditional expectation
\begin{equation} \label{eq:def}
\pi_t^{n;n}(X) \triangleq \mathbb{E}_{n;n}[j_t(X)|\mathscr{Y}(t)].
\end{equation}
Here, the subscript ``$n;n$'' in the expectation notation $\mathbb{E}$ is used to indicate that the input field is in an $n$-photon state. The exact form of  $n$-photon states and the notation $\mathbb{E}_{n;n}$ will be made clear in due course.  As introduced in Subsection \ref{subsec:systems_states}, $\mathscr{Y}(t)$ is the commutative von Neumann algebra generated by the observation processes $Y(s),~0\leq s\leq t$. The quantum conditional expectation in equation (\ref{eq:def}) is well-defined due to the fact that $j_t(X)$ satisfies the {\it non-demolition} condition $[j_t(X),Y(s)]=0$
for all $s\leq t$.   The quantum conditional expectation for the photodetection case can be defined in a similar manner, specifically,
\begin{equation}
\hat{\pi}_t^{n;n}(X) \triangleq \mathbb{E}_{n;n}[j_t(X)|\mathscr{Y}^\Lambda(t)].\label{eq:poisson}
\end{equation}

Due to the complexity of the multi-photon filtering problem, to better present the main ideas, we first focus on the two-photon case and conduct a detailed study of the two-photon filtering problem in Section \ref{sec:two_photon}. After that we proceed to the general $n$-photon case in Section \ref{sec:multi_photon_filtering}.

{\it Example }1. In this example we demonstrate the above-mentioned $(S,L,H)$ language by means of a toy model:  a  two-level system in a one-way waveguide. Here, ``one-way'' means that photons can only propagate along one direction in the waveguide \cite{FKS10}-\cite{SF05}. The state space of the two-level system $G$ is $\msf{H}_{S}=\mbb{C}^2$ whose basis vectors are the ground state $|g\r =[0~1]^T$ and the excited state $|e\r=[1~0]^T$. System operators are 2-by-2 matrices of complex numbers, for example $\sigma_z = |e\r \l e| - |g\r \l g|$, $\sigma_+ = |e\r \l g|$, and $\sigma_- =  |g\r \l e|$. In the interaction picture the total Hamiltonian of the composite system is ($\hbar=1$)
\begin{equation} \label{eq:nov_H_total}
H_{\rm total}=\frac{ \omega _{c}}{2}\sigma _{z}+\int_{-\infty }^{\infty }
\omega b^{\ast }(\omega )b(\omega )\; {{\rm d}}\omega +{\rm i} \sqrt{\frac{\kappa}{2\pi}}\int_{-\infty }^{\infty
}\left( \sigma _{+}b(\omega )-\sigma _{-}b^{\ast }(\omega )\right) {{\rm d}}\omega,\nonumber
\end{equation}
in which the first term and second term on the right-hand side are the free Hamiltonians of the system and the field respectively, while the third one is the interaction Hamiltonian.
The detuning $\omega_c = \Omega - \omega_0$ where $\Omega$  is the atomic transition frequency between the ground state $|g\r $ and the excited state $|e\r$ and $\omega_0$ is the carrier frequency of the input light field,  $b(\omega)$ and its adjoint  $b^\ast(\omega)$ are the annihilation operator and the creation operator of the input field, respectively, and $\kappa >0$ is related to the coupling constant between the two-level system $G$ and the field. For this model, in Heisenberg picture we have the following expressions at $t
\geq t_0$,
\begin{eqnarray}
&&\frac{{{\rm d}}}{{{\rm d}} t}b(\omega ,t) = -{\rm i}[b(\omega ,t),H_{\rm total}(t)] = -{\rm i}\omega b(\omega ,t)- \sqrt{\frac{\kappa}{2\pi}}\sigma _{-}(t), \label{sys2_a}
\\
&&\dot{\sigma}_{-}(t) =-{\rm i}[\sigma_-(t),H_{\rm total}(t)] = -{\rm i}\omega _{c}\sigma _{-}(t)- \sqrt{\frac{\kappa}{2\pi}}\sigma
_{z}(t)\int_{-\infty }^{\infty }b(\omega ,t){{\rm d}}\omega.  \label{sys2_b}
\end{eqnarray}
Here, $t_0$ is the initial time, namely, the time when the system and the field start to interact. Integrating equation (\ref{sys2_a}) from $t_0$ to $t$ yields
\begin{equation}
b(\omega ,t)={\rm e}^{-{\rm i}\omega (t-t_{0})}b(\omega
,t_{0})- \sqrt{\frac{\kappa}{2\pi}}\int_{t_{0}}^{t}{\rm e}^{-{\rm i}\omega (t-r)}\sigma _{-}(r){\rm d}r.
\label{eq:b_omega_t}
\end{equation}
Putting (\ref{eq:b_omega_t}) back into (\ref{sys2_b}) we have
\begin{eqnarray}
\dot{\sigma}_{-}(t) = -\left( \frac{\kappa }{2}+{\rm i}\omega _{c}\right) \sigma _{-}(t)+\sqrt{\kappa
}\sigma _{z}(t)b(t),\nonumber
\end{eqnarray}
where
\begin{equation}
b(t) \triangleq -\frac{1}{\sqrt{2\pi }}\int_{-\infty }^{\infty }{\rm e}^{-{\rm i}\omega
(t-t_{0})}b(\omega ,t_{0}){{\rm d}}\omega\nonumber    \label{eq:b_3}
\end{equation}
is the annihilation operator introduced in Subsection \ref{subsec:systems_states}.
On the other hand, let $t_{1}$ be the terminal time. Integrating equation (\ref{sys2_a}) from $t$ to $t_1$ we have
\begin{equation}
b(\omega ,t)={\rm e}^{-{\rm i}\omega (t-t_{1})}b(\omega
,t_{1})-\sqrt{\frac{\kappa}{2\pi}}\int_{t_{1}}^{t}{\rm e}^{-{\rm i}\omega (t-r)}\sigma _{-}(r){{\rm d}} r.\nonumber
\label{eq_b_omega_t_2}
\end{equation}
Define the output operator by
\begin{equation}
b_{\rm out}(t) \triangleq -\frac{1}{\sqrt{2\pi }}\int_{-\infty }^{\infty }{\rm e}^{-{\rm i}\omega
(t-t_{1})}b(\omega ,t_{1}){{\rm d}}\omega . \nonumber \label{eq:b_out_2}
\end{equation}
It can be easily shown that the input operator $b(t)$ and the output operator $b_{\rm out}(t)$ are related by
\[
 b_{\rm out}(t) = \sqrt{\kappa}\sigma _{-}(t)+b(t).
\]
Finally, the following rotations
\begin{equation}
\sigma _{-}(t) \to {\rm e}^{{\rm i}\omega _{c}t}\sigma _{-}(t), ~~
b(t) \to  {\rm e}^{{\rm i}\omega _{c}t}b(t), ~~
b_{\rm out}(t) \to {\rm e}^{{\rm i}\omega _{c}t}b_{\rm out}(t)\nonumber
\end{equation} 
yield the final model of interest
\begin{eqnarray}
\dot{\sigma}_-(t)
&=&
-\frac{\kappa }{2}\sigma_{-}(t)+\sqrt{\kappa }\sigma _{z}(t)b(t),
\label{eq:nov08_1a}
\\
b_{\rm out}(t) &=&\sqrt{\kappa }\sigma_{-}(t)+b(t).
\label{eq:nov08_1b}
\end{eqnarray}
For the model (\ref{eq:nov08_1a})--(\ref{eq:nov08_1b}) we have $S=I$, $L=\sqrt{\kappa }\sigma_{-}$, and $H=0$.

{\it Remark }1. It is worthwhile to notice that the model (\ref{eq:nov08_1a})-(\ref{eq:nov08_1b})  can also be used to describe a two-level atom in free space, as previously studied in \cite{gough:2013}, \cite{gough:2012b}, \cite{baragiola:2012}, \cite{GEP+98}-\cite{wang:2011}.


\section{Two-photon filtering}\label{sec:two_photon}

In this section we present a detailed study of the quantum filtering problem for an arbitrary quantum system driven by a two-photon state. Two-photon states are defined in Subsection \ref{subsec:state}, the mater equations are presented in Subsection \ref{subsec:2 photon mater eqn}, and the filtering equations for the homodyne detection case are derived in Subsections \ref{subsec:extended_system}--\ref{subsec:filter:2photon}.

\subsection{Two-photon states}\label{subsec:state}

Given a function $\xi\in L_2(\mbb{R}^+,\mbb{C})$, define an operator
\begin{equation}
B(\xi) \triangleq \int_0^\infty \xi^\ast(t)b(t){\rm d} t,\nonumber
\end{equation}
whose adjoint operator is
\begin{equation} \label{eq:B^ast_xi}
B^\ast(\xi) \triangleq \int_0^\infty \xi(t) b^\ast(t){\rm d} t.
\end{equation}

Given two functions  $\xi_1,\xi_2\in L_2(\mbb{R}^+,\mbb{C})$ satisfying $\|\xi_1\|=\|\xi_2\|=1$,  define single-photon states $|\Phi_{10}\r$ and $|\Phi_{01}\r$ to be
\begin{equation}
|\Phi_{10}\r \triangleq B^\ast(\xi_1) |0\r\qquad\mbox{and}\qquad|\Phi_{01}\r \triangleq B^\ast(\xi_2) |0\r,\label{eq:Phi_1}
\end{equation}
respectively. Then we define a two-photon state
\begin{eqnarray}
|\Phi_{11}\rangle\triangleq\frac{1}{\sqrt{N_2}}B^*(\xi_1)B^*(\xi_2)|0\rangle,
\label{eq:2_photon_state}
\end{eqnarray}
where
$N_2 = 1+|\l \xi_1|\xi_2 \r|^2$ is a normalization coefficient. If $\xi_1\equiv \xi_2$, then $|\Phi_{11}\r$ is a continuous-mode two-photon Fock state. Finally, for notational convention, denote $|\Phi_{00}\rangle \triangleq |0\r$.

For these states we have
\begin{eqnarray}\label{eq:aug17_temp1}
&&\qquad{\rm d} B(t)|\Phi_{00}\r =0, ~~~~{\rm d} B(t)|\Phi_{10}\rangle=\xi_1(t)|\Phi_{00}\rangle {\rm d} t,~~~~{\rm d} B(t)|\Phi_{01}\rangle=\xi_2(t)|\Phi_{00}\rangle {\rm d} t, \\
&&\qquad{\rm d} B(t)|\Phi_{11}\rangle=\frac{\xi_1(t)}{\sqrt{N_2}}|\Phi_{01}\rangle {\rm d} t+\frac{\xi_2(t)}{\sqrt{N_2}}|\Phi_{10}\rangle {\rm d} t.\nonumber
\end{eqnarray}

\subsection{Master equations}\label{subsec:2 photon mater eqn}
In this subsection we present the master equations for a quantum system $G$ driven by the two-photon state $|\Phi_{11}\r$ defined in equation (\ref{eq:2_photon_state}).

For a given system operator $X$ on $\mathsf{H}_{S}$, define expectations
\begin{eqnarray}\label{om}
\qquad \qquad \omega_t^{jk;mn}(X)
\triangleq
\mbb{E}_{jk;mn}[j_t(X)]
\equiv
\langle\eta\Phi_{jk}|j_t(X)|\eta\Phi_{mn}\rangle,~~~~~\forall ~  j,k,m,n=0,1,
\end{eqnarray}
where $|\eta\rangle$ is the initial state of the system. It can be easily verified that 
\begin{eqnarray}\label{eq:master_symmetry}
\omega_t^{mn;jk}(X)=(\omega_t^{jk;mn}(X^*))^*,~~~~~\forall ~ j,k,m,n=0,1.
\end{eqnarray}

In view of equations (\ref{eq:QDSE_X}) and (\ref{eq:aug17_temp1}),  if we differentiate $\omega_t^{11;11}(X)$, we will get such expressions as $\omega_t^{11;01}(X)$, $\omega_t^{11;10}(X)$, $\omega_t^{10;11}(X)$, and $\omega_t^{01;11}(X)$. Following this logic,  in order to derive the master equation for $\omega_t^{11;11}(X)$, we have to find derivatives of $\omega_t^{jk;mn}(X)$, for all $j,k,m,n=0,1$.

\begin{theorem} \label{thm:master_Heisenberg}
The master equation in Heisenberg picture for the quantum system $G$ driven by the two-photon input field state $|\Phi_{11}\rangle$ is given by the system of differential equations
\begin{eqnarray*}
\dot{\omega}_t^{00;00}(X)&=&\omega_t^{00;00}(\mathcal{L}_{00}(X)),
\label{eq:mastereq_first}
\\
\dot{\omega}_t^{00;10}(X)&=&\omega_t^{00;10}(\mathcal{L}_{00}(X))+\xi_1(t)\omega_t^{00;00}(\mathcal{L}_{01}(X)),
\label{eq:mastereq_2nd}
\\
\dot{\omega}_t^{00;01}(X)&=&\omega_t^{00;01}(\mathcal{L}_{00}(X))+\xi_2(t)\omega_t^{00;00}(\mathcal{L}_{01}(X)),
\label{eq:mastereq_third}
\\
\dot{\omega}_t^{10;10}(X)&=&\omega_t^{10;10}(\mathcal{L}_{00}(X))+\xi_1(t)\omega_t^{10;00}(\mathcal{L}_{01}(X))+\xi_1^*(t)\omega_t^{00;10}(\mathcal{L}_{10}(X))\nonumber\\
&&+|\xi_1(t)|^2\omega_t^{00;00}(\mathcal{L}_{11}(X)),
\label{eq:mastereq_5th}
\\
\dot{\omega}_t^{10;01}(X)&=&\omega_t^{10;01}(\mathcal{L}_{00}(X))+\xi_2(t)\omega_t^{10;00}(\mathcal{L}_{01}(X))+\xi_1^*(t)\omega_t^{00;01}(\mathcal{L}_{10}(X))\nonumber\\
&&+\xi_1^*(t)\xi_2(t)\omega_t^{00;00}(\mathcal{L}_{11}(X)),
\\
\dot{\omega}_t^{01;01}(X)&=&\omega_t^{01;01}(\mathcal{L}_{00}(X))+\xi_2(t)\omega_t^{01;00}(\mathcal{L}_{01}(X))+\xi_2^*(t)\omega_t^{00;01}(\mathcal{L}_{10}(X))\nonumber\\
&&+|\xi_2(t)|^2\omega_t^{00;00}(\mathcal{L}_{11}(X)),
\\
\dot{\omega}_t^{00;11}(X)&=&\omega_t^{00;11}(\mathcal{L}_{00}(X))+\frac{1}{\sqrt{N_2}}\xi_1(t)\omega_t^{00;01}(\mathcal{L}_{01}(X))+\frac{1}{\sqrt{N_2}}\xi_2(t)\omega_t^{00;10}(\mathcal{L}_{01}(X)),
\label{eq:mastereq_4th}
\\
\dot{\omega}_t^{10;11}(X)&=&\omega_t^{10;11}(\mathcal{L}_{00}(X))+\xi_1^*(t)\omega_t^{00;11}(\mathcal{L}_{10}(X))\nonumber\\
&&+\frac{1}{\sqrt{N_2}}\xi_1(t)\omega_t^{10;01}(\mathcal{L}_{01}(X))+\frac{1}{\sqrt{N_2}}\xi_2(t)\omega_t^{10;10}(\mathcal{L}_{01}(X))\nonumber\\
&&+\frac{1}{\sqrt{N_2}}|\xi_1(t)|^2\omega_t^{00;01}(\mathcal{L}_{11}(X))+\frac{1}{\sqrt{N_2}}\xi_1^*(t)\xi_2(t)\omega_t^{00;10}(\mathcal{L}_{11}(X)),
\\
\dot{\omega}_t^{01;11}(X)&=&\omega_t^{01;11}(\mathcal{L}_{00}(X))+\xi_2^*(t)\omega_t^{00;11}(\mathcal{L}_{10}(X))\nonumber\\
&&+\frac{1}{\sqrt{N_2}}\xi_1(t)\omega_t^{01;01}(\mathcal{L}_{01}(X))+\frac{1}{\sqrt{N_2}}\xi_2(t)\omega_t^{01;10}(\mathcal{L}_{01}(X))\nonumber\\
&&+\frac{1}{\sqrt{N_2}}|\xi_2(t)|^2\omega_t^{00;10}(\mathcal{L}_{11}(X))+\frac{1}{\sqrt{N_2}}\xi_1(t)\xi_2^*(t)\omega_t^{00;01}(\mathcal{L}_{11}(X)),
\\
\dot{\omega}_t^{11;11}(X)&=&\omega_t^{11;11}(\mathcal{L}_{00}(X))+\frac{1}{\sqrt{N_2}}\xi_1(t)\omega_t^{11;01}(\mathcal{L}_{01}(X))+\frac{1}{\sqrt{N_2}}\xi_2(t)\omega_t^{11;10}(\mathcal{L}_{01}(X))
\nonumber\\
&&+\frac{1}{\sqrt{N_2}}\xi_1^*(t)\omega_t^{01;11}(\mathcal{L}_{10}(X))+\frac{1}{\sqrt{N_2}}\xi_2^*(t)\omega_t^{10;11}(\mathcal{L}_{10}(X))
\nonumber \\
&&+\frac{1}{N_2}|\xi_1(t)|^2\omega_t^{01;01}(\mathcal{L}_{11}(X))+\frac{1}{N_2}\xi_1^*(t)\xi_2(t)\omega_t^{01;10}(\mathcal{L}_{11}(X))
 \nonumber\\
&&+\frac{1}{N_2}\xi_1(t)\xi_2^*(t)\omega_t^{10;01}(\mathcal{L}_{11}(X))+\frac{1}{N_2}|\xi_2(t)|^2\omega_t^{10;10}(\mathcal{L}_{11}(X)),
\label{eq:mastereq_1111}
\end{eqnarray*}
with the initial conditions $\omega_0^{jk;mn}(X)=\langle\eta|X|\eta\rangle\langle\Phi_{jk}|\Phi_{mn}\r$ for all $j,k,m,n=0,1$.
Moreover, the differential equations for $\omega^{10;00}_t(X)$, $\omega^{01;00}_t(X)$, $\omega^{01;10}_t(X)$, $\omega^{11;00}_t(X)$, $\omega^{11;10}_t(X)$, and $\omega^{11;01}_t(X)$ can be obtained from the above differential equations by means of the property (\ref{eq:master_symmetry}).
\end{theorem}

{\it Remark }2.  The system of equations in Theorem \ref{thm:master_Heisenberg} can be established directly by means of equations (\ref{eq:QDSE_X})
and (\ref{eq:aug17_temp1}). Alternatively, they can be obtained from the system of filtering equations
by averaging over the environment  (to be discussed in Subsection \ref{subsec:filter:2photon}).  Finally, as will be pointed out in Remark 5,
Theorem \ref{thm:master_Heisenberg} is an immediate consequence of Theorem \ref{thm:master_extended} in Subsection \ref{subsec:extended_system}. Thus, the  proof of Theorem \ref{thm:master_Heisenberg} is omitted.


{\it Remark }3. The first equation in Theorem \ref{thm:master_Heisenberg}
is nothing else but the master equation when the input state is the vacuum state $|0\r$. Moreover,
the first, second and fourth equations in the theorem
are the system of master equations when the input state is the single-photon state $|\Phi_{10}\r$, as derived in \cite{gough:2013,gough:2012b}.

Next, we present the master equations in Schrodinger picture. Define operators $\varrho_t^{jk;mn}$ on $\mathsf{H}_{S}$ via
\begin{eqnarray} \label{eq:varrho_omega}
\mrm{Tr}[(\varrho^{jk;mn}_t)^*X]=\omega_t^{jk;mn}(X), ~~~\forall ~ j,k,m,n=0,1.
\end{eqnarray}
Clearly, $\varrho_t^{jk;mn}$ are reduced system density operators.  Given a system operator $\varrho$ on $\mathsf{H}_{S}$, define superoperators
\begin{eqnarray*}
&&\mathcal{D}_{00}(\varrho) \triangleq \frac{1}{2}[L\varrho,L^*]+\frac{1}{2}[L,\varrho L^*]-{\rm i}[H,\varrho],\\
&&\mathcal{D}_{01}(\varrho) \triangleq [S\varrho,L^*], ~
\mathcal{D}_{10}(\varrho) \triangleq [L,\varrho S^*], ~
\mathcal{D}_{11}(\varrho) \triangleq S\varrho S^*-\varrho.
\end{eqnarray*}

The master equation in Schrodinger picture is given in the following corollary, which is a direct consequence of Theorem \ref{thm:master_Heisenberg} and equation (\ref{eq:varrho_omega}).

\begin{corollary} \label{cor:master_Heisenberg}
The master equation in Schrodinger picture for the quantum system $G$ driven by the two-photon input field state $|\Phi_{11}\rangle$ is given by the system of differential equations
\begin{eqnarray*}{\label{master_2_example}}
\dot{\varrho}^{00;00}_t&=&\mathcal{D}_{00}(\varrho_t^{00;00}),
 \\
\dot{\varrho}^{00;10}_t&=&\mathcal{D}_{00}(\varrho_t^{00;10})+\xi_1^*(t)\mathcal{D}_{10}(\varrho_t^{00;00}),
 \\
 \dot{\varrho}^{00;01}_t&=&\mathcal{D}_{00}(\varrho_t^{00;01})+\xi_2^*(t)\mathcal{D}_{10}(\varrho_t^{00;00}),
 \\
\dot{\varrho}^{10;10}_t&=&\mathcal{D}_{00}(\varrho_t^{10;10})+\xi_1^*(t)\mathcal{D}_{10}(\varrho_t^{10;00})+\xi_1(t)\mathcal{D}_{01}(\varrho_t^{00;10})+|\xi_1(t)|^2\mathcal{D}_{11}(\varrho_t^{00;00}),
 \\
\dot{\varrho}_t^{10;01}&=&\mathcal{D}_{00}(\varrho_t^{10;01})+\xi_2^*(t)\mathcal{D}_{10}(\varrho_t^{10;00})+\xi_1(t)\mathcal{D}_{01}(\varrho_t^{00;01})+\xi_1(t)\xi_2^*(t)\mathcal{D}_{11}(\varrho_t^{00;00}),
 \\
\dot{\varrho}_t^{01;01}&=&\mathcal{D}_{00}(\varrho_t^{01;01})+\xi_2^*(t)\mathcal{D}_{10}(\varrho_t^{01;00})+\xi_2(t)\mathcal{D}_{01}(\varrho_t^{00;01})+|\xi_2(t)|^2\mathcal{D}_{11}(\varrho_t^{00;00}),
  \\
\dot{\varrho}_t^{00;11}&=&\mathcal{D}_{00}(\varrho_t^{00;11})+\frac{1}{\sqrt{N_2}}\xi_1^*(t)\mathcal{D}_{10}(\varrho_t^{00;01})+\frac{1}{\sqrt{N_2}}\xi_2^*(t)\mathcal{D}_{10}(\varrho_t^{00;10}),
\\
\dot{\varrho}_t^{10;11}&=&\mathcal{D}_{00}(\varrho_t^{10;11})+\frac{1}{\sqrt{N_2}}\xi_1^*(t)\mathcal{D}_{10}(\varrho_t^{10;01})+\frac{1}{\sqrt{N_2}}\xi_2^*(t)\mathcal{D}_{10}(\varrho_t^{10;10})+\xi_1(t)\mathcal{D}_{01}(\varrho_t^{00;11})
 \nonumber
 \\
&&+\frac{1}{\sqrt{N_2}}|\xi_1(t)|^2\mathcal{D}_{11}(\varrho_t^{00;01})+\frac{1}{\sqrt{N_2}}\xi_1(t)\xi_2^*(t)\mathcal{D}_{11}(\varrho_t^{00;10}),
\\
\dot{\varrho}_t^{01;11}&=&\mathcal{D}_{00}(\varrho_t^{01;11})+\frac{1}{\sqrt{N_2}}\xi_1^*(t)\mathcal{D}_{10}(\varrho_t^{01;01})+\frac{1}{\sqrt{N_2}}\xi_2^*(t)\mathcal{D}_{10}(\varrho_t^{01;10})+\xi_2(t)\mathcal{D}_{01}(\varrho_t^{00;11})
 \nonumber
\\
&&+\frac{1}{\sqrt{N_2}}|\xi_2(t)|^2\mathcal{D}_{11}(\varrho_t^{00;10})+\frac{1}{\sqrt{N_2}}\xi_1^*(t)\xi_2(t)\mathcal{D}_{11}(\varrho_t^{00;01}),
\end{eqnarray*}
\begin{eqnarray*}
\dot{\varrho}_t^{11;11}&=&\mathcal{D}_{00}(\varrho_t^{11;11})+\frac{1}{\sqrt{N_2}}\xi_1^*(t)\mathcal{D}_{10}(\varrho_t^{11;01})+\frac{1}{\sqrt{N_2}}\xi_2^*(t)\mathcal{D}_{10}(\varrho_t^{11;10})
\nonumber
\\
&&+\frac{1}{\sqrt{N_2}}\xi_1(t)\mathcal{D}_{01}(\varrho_t^{01;11})+\frac{1}{\sqrt{N_2}}\xi_2(t)\mathcal{D}_{01}(\varrho_t^{10;11})+\frac{1}{N_2}|\xi_1(t)|^2\mathcal{D}_{11}(\varrho_t^{01;01})
\nonumber\\
&&+\frac{1}{N_2}\xi_1(t)\xi_2^*(t)\mathcal{D}_{11}(\varrho_t^{01;10})+\frac{1}{N_2}\xi_1^*(t)\xi_2\mathcal{D}_{11}(\varrho_t^{10;01})+\frac{1}{N_2}|\xi_2(t)|^2\mathcal{D}_{11}(\varrho_t^{10;10}),
\label{eq:varpho_22}
\end{eqnarray*}
and
\begin{eqnarray}
&&\varrho^{10;00}_t = (\varrho^{00;10}_t)^\ast, ~~~ \varrho^{01;00}_t = (\varrho^{00;01}_t)^\ast, ~~~\varrho^{01;10}_t = (\varrho^{10;01}_t)^\ast,
\nonumber
\\
&&\varrho^{11;00}_t = (\varrho^{00;11}_t)^\ast, ~~~ \varrho^{11;10}_t = (\varrho^{10;11}_t)^\ast, ~~~ \varrho^{11;01}_t = (\varrho^{01;11}_t)^\ast,\nonumber
\end{eqnarray}
with the initial conditions
\begin{eqnarray}
\varrho^{jk;mn}_0
&=&
\langle\Phi_{mn}|\Phi_{jk}\rangle|\eta\rangle\langle\eta|, ~~~~\forall ~ j,k,m,n=0,1.
\label{eq_varrho_0}
\end{eqnarray}
\end{corollary}


{\it Remark }4. Restricted to the 2-photon Fock state case, i.e., $\xi_1 \equiv \xi_2$, the above master equations reduce to the master equation (41) in \cite{baragiola:2012}, while the initial conditions (\ref{eq_varrho_0}) reduce to (42)--(43) in \cite{baragiola:2012}. It is clear that the general 2-photon case is much more complicated than the 2-photon Fock state case.

\begin{figure}
\begin{center}
\includegraphics[width=3.5in]{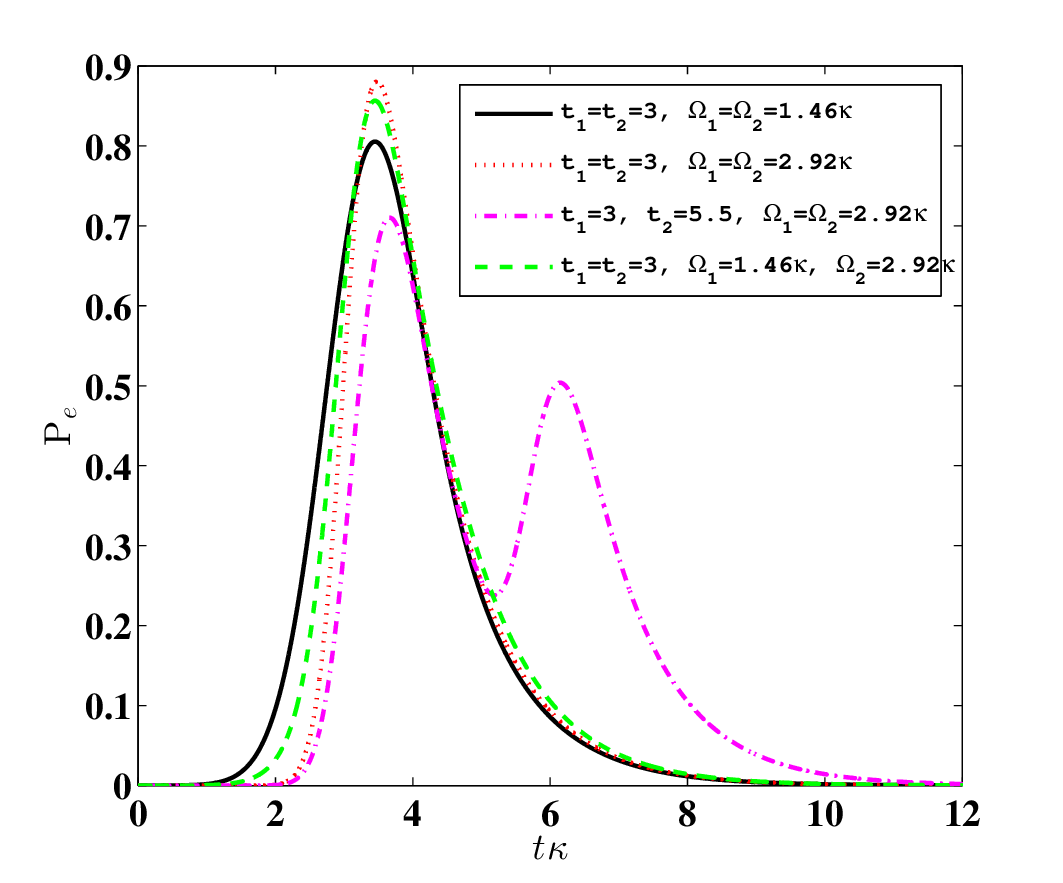}
\caption{Excitation probability for the model in Example 2.}
\label{fig_ex1}
\end{center}
\end{figure}

{\it Example }2.  We illustrate the two-photon master equations derived above by means of  the model (\ref{eq:nov08_1a})--(\ref{eq:nov08_1b}) studied in Example 1.  Let the two-level system $G$ be driven by a wavepacket prepared in the two-photon state $|\Phi_{11}\r$, as defined in equation (\ref{eq:2_photon_state}). We show that, in this two-photon case, interesting phenomena can be observed which are absent from both the single-photon case and the two-photon Fock state case, as previously studied in \cite{gough:2013}, \cite{gough:2012b}, \cite{baragiola:2012}, \cite{GEP+98}-\cite{RSF10}. Here, we assume $\kappa=1$ and the two-level system is initialized in the ground state $|g\r$.  For the two-photon input state $|\Phi_{11}\r$, we use Gaussian pulse shapes. Specifically, we choose
\begin{eqnarray}
\xi_i(t)=\Bigl(\frac{\Omega_i^2}{2\pi}\Bigr)^{1/4}\exp\Bigl(-\frac{\Omega_i^2}{4}(t-t_i)^2\Bigr),~~~i=1,2.
\label{eq:Gaussian_pulse}
\end{eqnarray}
For the single-photon state $|\Phi_{10}\r$ or $|\Phi_{01}\r$  defined in equation (\ref{eq:Phi_1}), $t_i$ can be interpreted as the peak arrival time of the photon, and $\Omega_i$ is the frequency  bandwidth. More discussions on the physical aspect of the model can be found in, e.g., \cite{baragiola:2012}, \cite{stobinska:2009}-\cite{RSF10}. Let $\varrho_t^{11;11}$ be the solution to the master equations in Corollary \ref{cor:master_Heisenberg}. Then the unconditional excitation probability (the probability of finding the two-level system in the excited state $|e\r$) is $\mathbb{P}_e(t) \triangleq \mrm{Tr}[\varrho_t^{11;11}|e\rangle\langle e|]$.


The excitation of a two-level system by a light field in a single-photon state has been studied extensively, e.g., \cite{baragiola:2012,stobinska:2009,wang:2011,RSF10}. If the two-level system $G$ is driven by $|\Phi_{10}\r$ with the Gaussian pulse shape in equation (\ref{eq:Gaussian_pulse}), it is found in \cite{stobinska:2009} that the largest value, denoted $\mathbb{P}_e^{\mrm{max}}$, of the excitation probability $\mathbb{P}_e(t)$ is around 0.8, which is achieved when $t_1=3$ and $\Omega_1=1.46\kappa$ (the optimal bandwidth). This value has been confirmed in \cite{gough:2012b,baragiola:2012,wang:2011}. 

The situation of the excitation of a two-level atom by a 2-photon state is much more complicated than the single-photon case. A two-level atom is a nonlinear system, it can at most absorb or emit one photon at a given time; and the absorption of one photon by the two-level atom may have drastic effect of the atom's response to the second coming photon. This nonlinear photon-photon interaction meditated by a two-level system gives rise to this interesting phenomenon. In what follows we study this by means of the master equations derived above. We have the following observations

\begin{itemize}
\item If the two peak arrival times $t_1$ and $t_2$ are very far away from each other, then the two photon interacts the two-level system one-by-one, thus the meditated photon-photon interaction does not happen. This is similar to the single-photon case and the maximal excitation probability $\mathbb{P}_e(t)=0.805$.

\item When $t_1=t_2=3$ and $\Omega_1=\Omega_2=1.46\kappa$, it can be seen from the black solid line in figure \ref{fig_ex1} that $\mathbb{P}_e^{\mrm{max}}=0.805$, which is consistent  with  figure 2(a) in \cite{baragiola:2012}.  In this case, it appears that one photon is absorbed by the two-level system while the other one just goes through without interaction. Therefore, this case is similar to the one-photon case. 

\item when $t_1=t_2=3$, and $\Omega_1= \Omega_2=2.92\kappa$, it can be seen from the red dotted line in figure \ref{fig_ex1} that  $\mathbb{P}_e^{\mrm{max}}=0.8796$.  This cannot occur in the single-photon state case. It is also interesting to notice that in this case the optimal bandwidth (equivalently, the ratio of the bandwidth and the decay rate) is exactly twice of that for the single-photon case. To the best knowledge of the author, this has never been reported in the literature.

\item When $t_1=t_2=3$,  $\Omega_1=1.46\kappa$, and $\Omega_2=2.92\kappa$, as can be seen from the green dashed line in figure \ref{fig_ex1}, there is one peak whose value  is approximately $0.8556$, which is still bigger than $0.805$ for the optimal single-photon case. That is, the nonlinear photon-photon interaction meditated by the two-level system still exists.

\item When $\Omega_1= \Omega_2=2.92\kappa$, and $t_1=3, t_2=5.5$,  as shown by the magenta dash dotted line in figure \ref{fig_ex1},  the value of the first peak is around $0.7102$ which is even than $0.805$ for the optimal single-photon case (the black solid line in figure \ref{fig_ex1}), while the value of the second peak is only around 0.5. Moreover $\mathbb{P}_e(t)$ does not drop to zero after the first  peak. This means that the excited two-level system is being affected by the other photon in the field in its decay process.  The authors are not aware of existing literature that reports such interesting nonlinear atom-photon interaction phenomena to the mathematical rigour presented here.

\end{itemize}

We will return to this example later and study its two-photon filters. We will show that there are many quantum trajectories whose largest excitation probability can be very close to 1 in all the above senarios.


\subsection{Master equations for an extended system} \label{subsec:extended_system}
In this subsection we define an ancilla, then derive the master equations for the extended system: system plus field plus ancilla.

Let
\begin{equation} \label{eq:e_jk}
|e_{11}\rangle = |e\rangle \otimes |e\rangle,  ~ |e_{10}\rangle = |e\rangle \otimes |g\rangle, ~ |e_{01}\rangle = |g\rangle \otimes |e\rangle, ~ |e_{00}\rangle = |g\rangle \otimes |g\rangle
\end{equation}
be an orthonormal basis for $\mbb{C}^4$.
Define a state $ |\Sigma\rangle \in \mbb{C}^4\otimes\mathsf{H}_{S}\otimes\mathsf{H}_{F}$ to be
\begin{equation}  \label{eq:sigma_2_photon}
|\Sigma\rangle \triangleq \alpha_{11}|e_{11}\eta\Phi_{11}\rangle+\alpha_{10}|e_{10}\eta\Phi_{10}\rangle+\alpha_{01}|e_{01}\eta\Phi_{01}\rangle+\alpha_{00}|e_{00}\eta\Phi_{00}\rangle,
\end{equation}
where $\alpha_{11},~\alpha_{10},~\alpha_{01}$ and $\alpha_{00}$ are nonzero complex numbers satisfying the normalization condition $\sum_{j,k=0}^1|\alpha_{jk}|^2=1$.

Now we have an extended system defined on the tensor product space $\mbb{C}^4\otimes\mathsf{H}_{S}\otimes\mathsf{H}_{F}$.  We assume that operators defined on $\mbb{C}^4$ do not evolve temporally. More specifically, for an arbitrary $4\times 4$ complex matrix $A$ on $\mbb{C}^4$ and $X$ on $\mathsf{H}_{S}$, the temporal evolution of $A\otimes X\otimes I$ is governed by $(I\otimes U^\ast(t)) (A\otimes X\otimes I)(I\otimes U(t)) = A\otimes j_t(X)$. The adoption of the auxiliary space $\mbb{C}^4$ allows us to define conditional expectations and derive their filtering equations on the extended space $\mbb{C}^4\otimes\mathsf{H}_{S}\otimes\mathsf{H}_{F}$ with respect to the superposition state $|\Sigma\rangle$. Such conditional expectations in terms of the orthonormal vectors $|e_{jk}\r$ in equation (\ref{eq:e_jk}) help us find the conditional expectations for the original quantum system $G$ driven by the 2-photon state $|\Phi_{11}\r$, cf. (\ref{exp_con_2}) in the sequel. Therefore, careful manipulation on the quantum filters for the extended system enables us to derive 2-photon filters for the original system $G$. This is the so-called non-Markovian embedding method, which has already been used in \cite{gough:2013} for the problem of single-photon filtering.

For an arbitrary $4\times 4$ complex matrix $A$,  define superoperators
\begin{eqnarray}
&&\mathcal{K}_{00}(A)
\triangleq
A,
\label{eq_K_00}
\\
&&\mathcal{K}_{01}(A)
\triangleq
\frac{\alpha_{11}}{\alpha_{01}\sqrt{N_2}}\xi_1(t)A|e_{11}\rangle\langle e_{01}|+\frac{\alpha_{11}}{\alpha_{10}\sqrt{N_2}}\xi_2(t)A|e_{11}\rangle\langle e_{10}|
\label{eq_K_01}
\\
&&\hspace{4.5em}+\frac{\alpha_{10}}{\alpha_{00}}\xi_1(t)A|e_{10}\rangle\langle e_{00}|+\frac{\alpha_{01}}{\alpha_{00}}\xi_2(t)A|e_{01}\rangle\langle e_{00}|,\nonumber
\\
&&\mathcal{K}_{10}(A)
\triangleq
\mathcal{K}_{01}(A^\ast)^*,
\label{eq_K_10}
\\
&&\mathcal{K}_{11}(A)
\triangleq
\mathcal{K}_{10}(\mathcal{K}_{01}(A)).
\label{eq_K_11}
\end{eqnarray}
For these superoperators the following relations hold:
\begin{eqnarray}
&&\mathbb{E}_{\Sigma}[A{\rm d} t]
=
 \mathbb{E}_\Sigma[\mathcal{K}_{00}(A)]{\rm d} t, ~
\mathbb{E}_{\Sigma}[A\otimes {\rm d} B(t)]
=
\mathbb{E}_\Sigma[\mathcal{K}_{01}(A)]{\rm d} t,
\label{eq_K_00_A}
\\
&&\mathbb{E}_{\Sigma}[A\otimes {\rm d} B^*(t)]
=
\mathbb{E}_\Sigma[\mathcal{K}_{10}(A)]{\rm d} t, ~
\mathbb{E}_{\Sigma}[A\otimes {\rm d} \Lambda(t)]
=
 \mathbb{E}_\Sigma[\mathcal{K}_{11}(A)]{\rm d} t.
\label{eq_K_11_A}
\end{eqnarray}

The expectation of $A\otimes j_t(X)$ with respect to the superposition state $|\Sigma\rangle$ is defined by $\tilde{\omega}_t(A\otimes X)\triangleq\mathbb{E}_\Sigma[A\otimes j_t(X)] \equiv \l \Sigma |  A\otimes j_t(X) |\Sigma \r$.
This expectation is normalized, that is, $\tilde{\omega}_t(I\otimes I)=1$.

\begin{theorem}\label{thm:master_extended}
The temporal evolution of the expectation $\tilde{\omega}_t(A\otimes X)$ is governed by the following master equation
\begin{eqnarray}\label{zhu1}
\dot{\tilde{\omega}}_t(A\otimes X)=\tilde{\omega}_t(\mathcal{G}(A\otimes X)),
\end{eqnarray}
where the superoperator $\mathcal{G}(A\otimes X)$ is defined as
\begin{eqnarray} \label{eq:G}
\mathcal{G}(A\otimes X) \triangleq \sum_{j,k=0}^1\mathcal{K}_{jk}(A)\otimes \mathcal{L}_{jk}(X).
\end{eqnarray}
\end{theorem}

{\em Proof}.
By equations (\ref{eq:QDSE_X}) and (\ref{eq_K_00_A})--(\ref{eq_K_11_A}), we have
\begin{eqnarray}
{\rm d}\tilde{\omega}_t(A\otimes X)
&=&
\mbb{E}_\Sigma\bigl[A\otimes j_t(\mathcal{L}_{00}(X)){\rm d} t\bigr]+\mathbb{E}_\Sigma\bigl[A\otimes j_t(\mathcal{L}_{01}(X)){\rm d} B(t)\bigr]
\nonumber
\\
&&+\mathbb{E}_\Sigma\bigl[A\otimes j_t(\mathcal{L}_{10}(X)){\rm d} B^*(t)\bigr]+\mathbb{E}_\Sigma\bigl[A\otimes j_t(\mathcal{L}_{11}(X)){\rm d} \Lambda(t)\bigr]
\nonumber
\\
&=&\tilde{\omega}_t(\mathcal{G}(A\otimes X)){\rm d} t.\nonumber\qquad\endproof
\end{eqnarray}

{\it Remark }5.
It can be easily verified that the expectations $\omega_t^{jk;mn}(X)$, defined in equation (\ref{om}),  are scaled components of $\tilde{\omega}_t(A\otimes X)$, that is,
\begin{eqnarray*}\label{re}
I\otimes \omega_t^{jk;mn}(X)=\frac{\tilde{\omega}_t(|e_{jk}\rangle\langle e_{mn}|\otimes X)}{\alpha^*_{jk}\alpha_{mn}}, ~~~~\forall ~ j,k,m,n=0,1.
\end{eqnarray*}
As a result, the system of master equations for $\omega_t^{jk;mn}(X)$ in Theorem \ref{thm:master_Heisenberg} can be derived from equation (\ref{zhu1}) by setting $A=|e_{jk}\rangle\langle e_{mn}|$ with $j,k,m,n=0,1$.

{\it Remark }6.
 Note that $\omega_t^{jk;mn}(X)$ can be alternatively re-written as
\begin{eqnarray*}
I\otimes\omega_t^{jk;mn}(X)\tilde{\omega}_t(|e_{11}\rangle\langle e_{11}|\otimes I)=\frac{|\alpha_{11}|^2}{\alpha^*_{jk}\alpha_{mn}}\tilde{\omega}_t(|e_{jk}\rangle\langle e_{mn}|\otimes X), ~~\forall ~ j,k,m,n=0,1.
\end{eqnarray*}
Interestingly,  a similar relation also holds in the filtering problem, cf. (\ref{exp_con_2}).

\subsection{Quantum filter for the extended system} \label{subsec:filter:extended}

In this subsection, we consider the homodyne detection and present the quantum filter for the extended system as introduced in the previous subsection.

Define the quantum conditional expectation by
\begin{eqnarray}\label{def_con}
\tilde{\pi}_t(A\otimes X)\triangleq\mathbb{E}_\Sigma[A\otimes j_t(X)|I\otimes\mathscr{Y}(t)].
\end{eqnarray}

\begin{theorem}\label{thm:filter_extended}
In the case of homodyne detection, the filtering equation for the conditional expectation $\tilde{\pi}_t(A\otimes X)$ for the extended system is
\begin{eqnarray}\label{conditional}
{\rm d}\tilde{\pi}_t(A\otimes X)=\tilde{\pi}_t(\mathcal{G}(A\otimes X)){\rm d} t+\tilde{\mathcal{H}}_t(A\otimes X){\rm d}\tilde{W}(t),
\end{eqnarray}
in which the superoperator $\mathcal{G}(A\otimes X)$ is that defined in equation (\ref{eq:G}) and
\begin{eqnarray*}
\tilde{\mathcal{H}}_t(A\otimes X) =\tilde{M}_t(A\otimes X) -\tilde{\pi}_t(A\otimes X)\tilde{M}_t(I\otimes I),
\label{eq:tilde_H}
\end{eqnarray*}
with
\begin{equation}\label{eq:tilde_M}
\tilde{M}_t(A\otimes X)\triangleq\tilde{\pi}_t(\mathcal{K}_{00}(A)\otimes (XL+L^*X)+\mathcal{K}_{01}(A)\otimes XS+\mathcal{K}_{10}(A)\otimes S^*X).
\end{equation}
Moreover, the stochastic process $\tilde{W}(t)$, which satisfies the following Ito equation
\begin{equation}
{\rm d}\tilde{W}(t)=I\otimes {\rm d} Y(t)-\tilde{M}_t(I\otimes I){\rm d} t,\nonumber
\end{equation}
 is  a Wiener process with respect to the state $|\Sigma\rangle$.
\end{theorem}

{\em Proof}. We use the characteristic function method by postulating the filter to be of the form
\begin{eqnarray} \label{eq:aug10_extended}
{\rm d}\tilde{\pi}_t(A\otimes X)=\tilde{\mathcal{F}}_t(A\otimes X){\rm d} t+\tilde{\mathcal{H}}_t(A\otimes X)I\otimes {\rm d} Y(t),
\end{eqnarray}
where $\tilde{\mathcal{F}}_t(A\otimes X)$ and $\tilde{\mathcal{H}}_t(A\otimes X)$ are to be determined.
For an arbitrary function $f\in L_2(\mbb{R}^+,\mbb{C})$,  define a random process $c_f(t)$ via
$c_f(t) \triangleq {\rm e}^{\int_0^tf(s){\rm d} Y(s)-\frac{1}{2}\int_0^t f^2(s){\rm d}s}$.
Clearly $c_f(0)=1$. Moreover, $c_f(t)$ satisfies ${\rm d}c_f(t)=f(t)c_f(t){\rm d} Y(t)$. Thus $I\otimes c_f(t)$ is adapted to $I\otimes\mathscr{Y}(t)$. By a property of conditional expectations, we have
\begin{eqnarray}\label{eq:aug27_1}
\mathbb{E}_\Sigma[(A\otimes j_t(X))(I\otimes c_f(t))]=\mathbb{E}_\Sigma[\tilde{\pi}_t(A\otimes X)(I\otimes c_f(t))].
\end{eqnarray}
Differentiating both sides of equation (\ref{eq:aug27_1}) and by means of properties of conditional expectations, we find
\begin{eqnarray*}
{\rm d}\mathbb{E}_{\Sigma}[A\otimes j_t(X)c_f(t)]&=&\mathbb{E}_{\Sigma}\biggl[\bigl(I\otimes c_f(t)\bigr)\tilde{\pi}_t(\mathcal{G}(A\otimes X))+\bigl(I\otimes f(t)c_f(t)\bigr)\tilde{M}_t(A\otimes X)\biggr]{\rm d} t,
\end{eqnarray*}
and
\begin{eqnarray}
{\rm d}\mathbb{E}_{\Sigma}[\tilde{\pi}_t(A\otimes X)(I\otimes c_f(t))]&=&\mathbb{E}_\Sigma\biggl[\bigl(I\otimes c_f(t)\bigr)\bigl\{\tilde{\mathcal{F}}_t(A\otimes X)+\tilde{\mathcal{H}}_t(A\otimes X)\tilde{M}_t(I\otimes I)\bigr\}\nonumber\\
&&+\bigl(I\otimes f(t)c_f(t)\bigr)\bigl\{\tilde{\pi}_t(A\otimes X)\tilde{M}_t(I\otimes I)+\tilde{\mathcal{H}}_t(A\otimes X)\bigr\}\biggr]{\rm d} t.
\nonumber
\end{eqnarray}
Comparing the coefficients of $c_f(t)$ and $f(t)c_f(t)$ respectively, we find the exact forms of $\tilde{\mathcal{F}}_t(A\otimes X)$ and $\tilde{\mathcal{H}}_t(A\otimes X)$. Putting them back into (\ref{eq:aug10_extended}) yields the filter (\ref{conditional}).

We now prove the martingale property $\mathbb{E}_{\Sigma}[\tilde{W}(t)-\tilde{W}(s)|I\otimes \mathscr{Y}(s)]=0$ for all $0\leq s\leq t$. This is equivalent to proving that $\mathbb{E}_\Sigma[(\tilde{W}(t)-\tilde{W}(s))(I\otimes K)]=0$ for all $K\in\mathscr{Y}(s), 0\leq s\leq t$. Obviously,
\begin{eqnarray*}
&&\mathbb{E}_\Sigma[(\tilde{W}(t)-\tilde{W}(s))(I\otimes K)]
\nonumber
\\
&=&
\mathbb{E}_\Sigma[I\otimes (Y(t)-Y(s))(I\otimes K)]-\mathbb{E}_\Sigma\bigl[\int_s^t\tilde{M}_r(I\otimes I){\rm d}r (I\otimes K)\bigr]
\nonumber
\\
&=&
\mathbb{E}_\Sigma[I\otimes (Y(t)-Y(s))(I\otimes K)]\nonumber\\
&&-\mathbb{E}_\Sigma\bigl[ \int_s^t{\tilde{\pi}_r\bigl(\mathcal{K}_{00}(I)\otimes (L+L^*)+\mathcal{K}_{01}(I)\otimes S+\mathcal{K}_{10}(I)\otimes S^*\bigr)}{\rm d}r(I\otimes K)\bigr]
\nonumber
\\
&=&
0.
\end{eqnarray*}
Finally, since ${\rm d}\tilde{W}(t){\rm d}\tilde{W}(t)={\rm d} t$, Levy's Theorem implies that $\tilde{W}(t)$ is a  Wiener process.
\qquad\endproof

{\it Remark }7.
Due to the martingale property of the innovations process $\tilde{W}(t)$, if we take the expected value of (\ref{conditional}), we can recover the master equation (\ref{zhu1}).


\subsection{Two-photon quantum filter} \label{subsec:filter:2photon}

In this subsection, we derive the quantum filter for the original quantum system $G$ driven by the  two-photon state $|\Phi_{11}\r$ defined in equation (\ref{eq:2_photon_state}).

Define implicitly the conditional expectations $\pi_t^{jk;mn}(X)$, $j,k,m,n=0,1$, for the original system $G$ via
\begin{eqnarray}\label{exp_con_2}
(I\otimes \pi_t^{jk;mn}(X))\tilde{\pi}_t(|e_{11}\rangle\langle e_{11}|\otimes I)=\frac{|\alpha_{11}|^2}{\alpha_{jk}^*\alpha_{mn}}\tilde{\pi}_t(|e_{jk}\rangle\langle e_{mn}|\otimes X),
\end{eqnarray}
where $\tilde{\pi}_t(A\otimes X)$ is the conditional expectation for the extended system, as defined in equation (\ref{def_con}). Clearly, $\pi_0^{jk;mn}(X) = \l\eta |X|\eta\r \l \Phi_{jk}|\Phi_{mn} \r$, and
\begin{eqnarray*}
\pi_t^{mn;jk}(X)=(\pi_t^{jk;mn}(X^*))^*, ~~~ \forall~ j,k,m,n=0,1.
\end{eqnarray*}

The equation (\ref{exp_con_2}) is very important in the derivation of the two-photon quantum filter, since it establishes a relationship between
the conditional expectations of the extended system and the original system. And $\pi_t^{11;11}(X)$ defined in this way is exactly the quantum conditional expectation for the two-photon field state $|\Phi_{11}\rangle$ as shown by the following lemma. Then we can get the desired two-photon quantum filter by means of the filtering equations for the extended system, Theorem \ref{thm:2-photon_filter}.
\begin{lemma} \label{lem:equivalence}
For all $K \in \mathscr{Y}(t)$,
\begin{eqnarray}
\mathbb{E}_{11;11}[\pi_t^{jk;mn}(X)K]=\mbb{E}_{jk;mn}[j_t(X)K],  ~~~ \forall~  j,k,m,n=0,1.
\label{eq:pi_j_t_X}
\end{eqnarray}
In particular,
\begin{equation} \label{eq:nov17_1}
\mathbb{E}_{11;11}[\pi_t^{11;11}(X)K]
=
\mathbb{E}_{11;11}[j_t(X)K].\nonumber
\end{equation}
That is, $\pi_t^{11;11}(X)$ is exactly the quantum conditional expectation for the two-photon field state $|\Phi_{11}\rangle$, namely, $\pi_t^{11;11}(X)=\mathbb{E}_{11;11}[j_t(X)|\mathscr{Y}(t)]$.
\end{lemma}

{\em Proof}. Noticing that for all $j,k,m,n=0,1$ and $K \in \mathscr{Y}(t)$, by equation (\ref{exp_con_2}) we have
\begin{eqnarray}
\hspace{-3em}\mathbb{E}_{11;11}[\pi_t^{jk;mn}(X)K]
&=&
\frac{1}{|\alpha_{11}|^2}\mathbb{E}_{\Sigma}[|e_{11}\rangle\langle e_{11}|\otimes (\pi_t^{jk;mn}(X)K)]
\nonumber
\\
&=&\frac{1}{|\alpha_{11}|^2}\mathbb{E}_{\Sigma}[\tilde{\pi}_t(|e_{11}\rangle\langle e_{11}|\otimes I)(I\otimes (\pi_t^{jk;mn}(X))(I\otimes K)]
\nonumber
\\
&=&\frac{1}{\alpha_{jk}^*\alpha_{mn}}\mathbb{E}_{\Sigma}[\tilde{\pi}_t(|e_{jk}\rangle\langle e_{mn}|\otimes X)(I\otimes K)]\nonumber
\\
&=& \frac{1}{\alpha_{jk}^*\alpha_{mn}}\mathbb{E}_{\Sigma}[\tilde{\pi}_t(|e_{jk}\rangle\langle e_{mn}|\otimes j_t(X)K)]
\nonumber
\\
&=&
\mbb{E}_{jk;mn}[j_t(X)K],\nonumber
\end{eqnarray}
which is exactly equation (\ref{eq:pi_j_t_X}).  Setting $j=k=m=n=1$ in (\ref{eq:pi_j_t_X}) gives (\ref{eq:nov17_1}).  Because $K \in \mathscr{Y}(t)$ is arbitrary, $\pi_t^{11;11}(X)$ is exactly the quantum conditional expectation for the two-photon field state $|\Phi_{11}\r$.\qquad \endproof

In what follows we derive the quantum filtering equations for the quantum conditional expectation $\pi_t^{11;11}(X)$.  To present the results clearly, we define the superoperators $M_t^{jk;mn}(X)$ ($j,k,m,n=0,1$) for an arbitrary system operator $X$ as follows:
\begin{eqnarray}\label{eq:aug3_6}
&&\quad\quad M_{t}^{jk;mn}(X)\\
&\triangleq& \pi _{t}^{jk;mn}(XL+L^{\ast }X)+\delta _{m1}\delta _{n0}\xi _{1}(t)\pi
_{t}^{jk;00}(XS)+\delta _{m0}\delta _{n1}\xi _{2}(t)\pi _{t}^{jk;00}(XS)\nonumber\\
&&\hspace{-0.5em}+\delta _{j1}\delta _{k0}\xi
_{1}^{\ast }(t)\pi _{t}^{00;mn}(S^{\ast }X)+\delta _{j0}\delta _{k1}\xi
_{2}^{\ast }(t)\pi _{t}^{00;mn}(S^{\ast }X)+\frac{\delta _{m1}\delta _{n1}}{\sqrt{N_{2}}}\xi _{1}(t)\pi
_{t}^{jk;01}(XS)\nonumber\\
&&\hspace{-0.5em}+\frac{\delta _{m1}\delta _{n1}}{\sqrt{N_{2}}}\xi _{2}(t)\pi
_{t}^{jk;10}(XS)+\frac{\delta _{j1}\delta _{k1}}{\sqrt{N_{2}}}\xi _{1}^{\ast }(t)\pi
_{t}^{01;mn}(S^{\ast }X)+\frac{\delta _{j1}\delta _{k1}}{\sqrt{N_{2}}}\xi
_{2}^{\ast }(t)\pi _{t}^{10;mn}(S^{\ast }X).\nonumber
\end{eqnarray}

\begin{theorem}\label{thm:2-photon_filter}
In the case of homodyne detection, the quantum filter for the quantum system $G$ driven by the $2$-photon state $|\Phi_{11}\rangle$ is given by the following system of Ito differential equations
\begin{eqnarray*}
{\rm d}\pi_t^{11;11}(X)
&=&
\biggl[\pi_t^{11;11}(\mathcal{L}_{00}(X))+\frac{1}{\sqrt{N_2}}\xi_1(t)\pi^{11;01}_t(\mathcal{L}_{01}(X))+\frac{1}{\sqrt{N_2}}\xi_2(t)\pi^{11;10}_t(\mathcal{L}_{01}(X))
\nonumber
\\
&&+\frac{1}{\sqrt{N_2}}\xi_1^*(t)\pi^{01;11}_t(\mathcal{L}_{10}(X))+\frac{1}{\sqrt{N_2}}\xi_2^*(t)\pi^{10;11}_t(\mathcal{L}_{10}(X))
\nonumber
\\
&&+\frac{1}{N_2}|\xi_1(t)|^2\pi^{01;01}_t(\mathcal{L}_{11}(X))+\frac{1}{N_2}|\xi_2(t)|^2\pi^{10;10}_t(\mathcal{L}_{11}(X))\nonumber\\
&&+\frac{1}{N_2}\xi_1(t)\xi_2^*(t)\pi^{10;01}_t(\mathcal{L}_{11}(X))+\frac{1}{N_2}\xi_1^*(t)\xi_2(t)\pi^{01;10}_t(\mathcal{L}_{11}(X))\biggr]{\rm d} t
\nonumber
\\
&&+\biggl[M_t^{11;11}(X)-\pi_t^{11;11}(X)M_t^{11;11}(I)\biggr]{\rm d} W(t),
\label{filter_two}
\end{eqnarray*}
where
\begin{eqnarray*}
{\rm d}\pi_t^{00;00}(X)
&=&
\pi_t^{00;00}(\mathcal{L}_{00}(X)){\rm d} t+\biggl[M_t^{00;00}(X)-\pi_t^{00;00}(X)M_t^{11;11}(I)\biggr]{\rm d} W(t),
\label{filter_two_general_1}
\\
{\rm d}\pi_t^{00;10}(X)
&=&
\biggl[\pi_t^{00;10}(\mathcal{L}_{00}(X))+\xi_1(t)\pi_t^{00;00}(\mathcal{L}_{01}(X))\biggr]{\rm d} t\nonumber\\
&&+\biggl[M_t^{00;10}(X)-\pi_t^{00;10}(X)M_t^{11;11}(I)\biggr]{\rm d} W(t),
\label{filter_two_general_2}
\\
{\rm d}\pi_t^{00;01}(X)
&=&
\biggl[\pi_t^{00;01}(\mathcal{L}_{00}(X))+\xi_2(t)\pi_t^{00;00}(\mathcal{L}_{01}(X))\biggr]{\rm d} t\nonumber\\
&&+\biggl[M_t^{00;01}(X)-\pi_t^{00;01}(X)M_t^{11;11}(I)\biggr]{\rm d} W(t),
\label{filter_two_general_3}
\\
{\rm d}\pi_t^{10;10}(X)
&=&
\biggl[\xi_1(t)\pi_t^{10;00}(\mathcal{L}_{01}(X))+\xi_1^*(t)\pi_t^{00;10}(\mathcal{L}_{10}(X))+|\xi_1(t)|^2\pi_t^{00;00}(\mathcal{L}_{11}(X))\nonumber\\
&&+\pi_t^{10;10}(\mathcal{L}_{00}(X))\biggr]{\rm d} t+\biggl[M_t^{10;10}(X))-\pi_t^{10;10}(X)M_t^{11;11}(I)\biggr]{\rm d} W(t),
\label{filter_two_general_5}
\\
{\rm d}\pi_t^{10;01}(X)
&=&
\biggl[\xi_2(t)\pi_t^{10;00}(\mathcal{L}_{01}(X))+\xi_1^*(t)\pi_t^{00;01}(\mathcal{L}_{10}(X))+\xi_1^*(t)\xi_2(t)\pi_t^{00;00}(\mathcal{L}_{11}(X))\nonumber\\
&&+\pi_t^{10;01}(\mathcal{L}_{00}(X))\biggr]{\rm d} t
+\biggl[M_t^{10;01}(X)-\pi_t^{10;01}(X)M_t^{11;11}(I)\biggr]{\rm d} W(t),
\label{filter_two_general_6}
\\
{\rm d}\pi_t^{01;01}(X)
&=&
\biggl[\xi_2(t)\pi_t^{01;00}(\mathcal{L}_{01}(X))+\xi_2^*(t)\pi_t^{00;01}(\mathcal{L}_{10}(X))+|\xi_2(t)|^2\pi_t^{00;00}(\mathcal{L}_{11}(X))\nonumber\\
&&+\pi_t^{01;01}(\mathcal{L}_{00}(X))\biggr]{\rm d} t
+\biggl[M_t^{01;01}(X)-\pi_t^{01;01}(X)M_t^{11;11}(I)\biggr]{\rm d} W(t),
\label{filter_two_general_7}
\\
{\rm d}\pi_t^{00;11}(X)
&=&
\biggl[\frac{1}{\sqrt{N_2}}\xi_1(t)\pi_t^{00;01}(\mathcal{L}_{01}(X))+\frac{1}{\sqrt{N_2}}\xi_2(t)\pi_t^{00;10}(\mathcal{L}_{01}(X))\nonumber\\
&&+\pi_t^{00;11}(\mathcal{L}_{00}(X))\biggr]{\rm d} t+\biggl[M_t^{00;11}(X)-\pi_t^{00;11}(X)M_t^{11;11}(I)\biggr]{\rm d} W(t),
\label{filter_two_general_4}
\end{eqnarray*}
\begin{eqnarray*}
{\rm d}\pi_t^{10;11}(X)
&=&
\biggl[\frac{1}{\sqrt{N_2}}\xi_1(t)\pi^{10;01}_t(\mathcal{L}_{01}(X))+\frac{1}{\sqrt{N_2}}\xi_2(t)\pi^{10;10}_t(\mathcal{L}_{01}(X))
\nonumber
\\
&&+\xi_1^*(t)\pi^{00;11}_t(\mathcal{L}_{10}(X))+\frac{1}{\sqrt{N_2}}|\xi_1(t)|^2\pi^{00;01}_t(\mathcal{L}_{11}(X))
\nonumber
\\
&&+\frac{1}{\sqrt{N_2}}\xi_1^*(t)\xi_2(t)\pi^{00;10}_t(\mathcal{L}_{11}(X))+\pi_t^{10;11}(\mathcal{L}_{00}(X))\biggr]{\rm d} t\nonumber\\
&&+\biggl[M_t^{10;11}(X)-\pi_t^{10;11}(X)M_t^{11;11}(I)\biggr]{\rm d} W(t),\nonumber\\
\label{filter_two_general_9}
{\rm d}\pi_t^{01;11}(X)
&=&
\biggl[\frac{1}{\sqrt{N_2}}\xi_1(t)\pi^{01;01}_t(\mathcal{L}_{01}(X))+\frac{1}{\sqrt{N_2}}\xi_2(t)\pi^{01;10}_t(\mathcal{L}_{01}(X))
\nonumber
\\
&&+\xi_2^*(t)\pi^{00;11}_t(\mathcal{L}_{10}(X))+\frac{1}{\sqrt{N_2}}\xi_1(t)\xi_2^*(t)\pi^{00;01}_t(\mathcal{L}_{11}(X))\nonumber\\
&&+\frac{1}{\sqrt{N_2}}|\xi_2(t)|^2\pi^{00;10}_t(\mathcal{L}_{11}(X))+\pi_t^{01;11}(\mathcal{L}_{00}(X))\biggr]{\rm d} t\nonumber\\
&&+\biggl[M_t^{01;11}(X)-\pi_t^{01;11}(X)M_t^{11;11}(I)\biggr]{\rm d} W(t),
\label{filter_two_general_8}
\end{eqnarray*}
and
\begin{eqnarray*}
&&\pi_t^{10;00}(X)
=
({\pi_t^{00;10}(X^*)})^*, ~\pi_t^{01;00}(X)
=
({\pi_t^{00;01}(X^*)})^*, ~ \pi_t^{11;00}(X)
=
({\pi_t^{00;11}(X^*)})^*,
\\
&& \pi_t^{01;10}(X)
=
({\pi_t^{10;01}(X^*)})^*, ~ \pi_t^{11;01}(X)
=
({\pi_t^{01;11}(X^*)})^*, ~ \pi_t^{11;10}(X)
=
({\pi_t^{10;11}(X^*)})^*,
\end{eqnarray*}
with the initial conditions $\pi_0^{jk;mn}(X)=\l\eta |X|\eta\r\langle \Phi_{jk}|\Phi_{mn}\rangle$ for all $j,k,m,n=0,1$. Moreover, the innovation process $W(t)$, defined by ${\rm d} W(t) = {\rm d} Y(t)-M_t^{11;11}(I){\rm d} t$,
 is a Wiener process with respect to the two-photon state $|\Phi_{11}\rangle$.
\end{theorem}

The proof is given in Appendix.


In what follows we present the stochastic master equations in Schrodinger picture.  Define conditional density operators $\rho_t^{jk;mn}$ on $\mathsf{H}_{S}\otimes\mathsf{H}_{F}$ in terms of
\begin{equation} \label{eq:nov11_temp1}
\pi_t^{jk;mn}(X)=\mbox{Tr}[(\rho_t^{jk;mn})^* (I\otimes X)], ~~~\forall~ j,k,m,n=0,1.
\end{equation}
Moreover, define superoperators $\mathcal{S}_{t}^{jk;mn}(\rho)$, $j,k,m,n=0,1$, as follows:
\begin{eqnarray*}
\qquad \mathcal{S}_{t}^{jk;mn}(\rho )\triangleq L{\rho _{t}^{jk;mn}+\rho _{t}^{jk;mn}}L^{\ast}+\delta _{m1}\delta _{n0}\xi
_{1}^{\ast }(t){\rho _{t}^{jk;00}}S^{\ast}+\delta _{m0}\delta _{n1}\xi
_{2}^{\ast }(t){\rho _{t}^{jk;00}}S^{\ast}\nonumber\\
+\delta _{j1}\delta _{k0}\xi _{1}(t)S{\rho _{t}^{00;mn}}+\delta_{j0}\delta _{k1}\xi _{2}(t)S{\rho _{t}^{00;mn}}
+\frac{\delta _{m1}\delta _{n1}}{\sqrt{N_{2}}}\xi _{1}^{\ast }(t){\rho
_{t}^{jk;01}}S^{\ast }\nonumber\\
+\frac{\delta _{m1}\delta _{n1}}{\sqrt{N_{2}}}\xi
_{2}^{\ast }(t){\rho _{t}^{jk;10}}S^{\ast }
+\frac{\delta _{j1}\delta _{k1}}{\sqrt{N_{2}}}\xi _{1}(t)S{\rho
_{t}^{01;mn}+}\frac{\delta _{j1}\delta _{k1}}{\sqrt{N_{2}}}\xi _{2}(t)S{\rho
_{t}^{10;mn}}.
\label{eq:S_up}
\end{eqnarray*}

\begin{corollary} \label{cor:filter_up}
In the case of homodyne detection, the stochastic master equations for conditional densities $\rho_t^{jk;mn}$ of the quantum system $G$ driven by the two-photon state $|\Phi_{11}\r$ are
\begin{eqnarray*}
{\rm d}\rho_t^{11;11}
&=&
\biggl[\mathcal{D}_{00}(\rho_t^{11;11})+\frac{1}{\sqrt{N_2}}\xi_1^*(t)\mathcal{D}_{10}(\rho_t^{11;01})+\frac{1}{\sqrt{N_2}}\xi_2^*(t)\mathcal{D}_{10}(\rho_t^{11;10})
\nonumber
\\
&&+\frac{1}{\sqrt{N_2}}\xi_1(t)\mathcal{D}_{01}(\rho_t^{01;11})+\frac{1}{\sqrt{N_2}}\xi_2(t)\mathcal{D}_{01}(\rho_t^{10;11})+\frac{1}{N_2}|\xi_1(t)|^2\mathcal{D}_{11}(\rho_t^{01;01})
\nonumber
\\
&&+\frac{1}{{N_2}}\xi_1(t)\xi_2^*(t)\mathcal{D}_{11}(\rho_t^{01;10})
+\frac{1}{N_2}\xi_1^*(t)\xi_2(t)\mathcal{D}_{11}(\rho_t^{10;01})+\frac{1}{N_2}|\xi_2(t)|^2\mathcal{D}_{11}(\rho_t^{10;10})\biggr]{\rm d} t
\nonumber
\\
&&+\biggl[\mathcal{S}_t^{11; 11}(\rho)-\rho_t^{11;11}\mrm{Tr}[\mathcal{S}_t^{11;11}(\rho)]\biggr]{\rm d}\bar{W}_t,
\end{eqnarray*}
where
\begin{eqnarray*}{\label{filter_2_example}}
{\rm d}\rho_t^{00;00}
&=&
\mathcal{D}_{00}(\rho_t^{00;00}){\rm d} t+\biggl[\mathcal{S}_t^{00;00}(\rho)-\rho_t^{00;00}\mrm{Tr}[\mathcal{S}_t^{11;11}(\rho)]\biggr]{\rm d}\bar{W}_t,
\label{eq:filter_2_0000}
\\
{\rm d}\rho_t^{00;10}
&=&
\biggl[\mathcal{D}_{00}(\rho_t^{00;10})+\xi_1^*(t)\mathcal{D}_{10}(\rho_t^{00;00})\biggr]{\rm d} t+\biggl[\mathcal{S}_t^{00;10}(\rho)-\rho_t^{00;10}\mrm{Tr}[\mathcal{S}_t^{11;11}(\rho)]\biggr]{\rm d}\bar{W}_t,
\label{eq:filter_2_0010}
\\
{\rm d}\rho_t^{00;01}
&=&
\biggl[\mathcal{D}_{00}(\rho_t^{00;01})+\xi_2^*(t)\mathcal{D}_{10}(\rho_t^{00;00})\biggr]{\rm d} t+\biggl[\mathcal{S}_t^{00; 01}(\rho)-\rho_t^{00;01}\mrm{Tr}[\mathcal{S}_t^{11;11}(\rho)]\biggr]{\rm d}\bar{W}_t,
\label{eq:filter_2_0001}
\\
{\rm d}\rho_t^{10;10}
&=&
\biggl[\mathcal{D}_{00}(\rho_t^{10;10})+\xi_1^*(t)\mathcal{D}_{10}(\rho_t^{10;00})+\xi_1(t)\mathcal{D}_{01}(\rho_t^{00;10})+|\xi_1(t)|^2\mathcal{D}_{11}(\rho_t^{00;00})\biggr]{\rm d} t
\nonumber
\\
&&+\biggl[\mathcal{S}_t^{10; 10}(\rho)-\rho_t^{10;10}\mrm{Tr}[\mathcal{S}_t^{11;11}(\rho)]\biggr]{\rm d}\bar{W}_t,
\label{eq:filter_2_1010}
\\
{\rm d}\rho_t^{10;01}
&=&
\biggl[\mathcal{D}_{00}(\rho_t^{10;01})+\xi_2^*(t)\mathcal{D}_{10}(\rho_t^{10;00})+\xi_1(t)\mathcal{D}_{01}(\rho_t^{00;01})+\xi_1(t)\xi_2^*(t)\mathcal{D}_{11}(\rho_t^{00;00})\biggr]{\rm d} t
\nonumber
\\
&&+\biggl[\mathcal{S}_t^{10; 01}(\rho)-\rho_t^{10;01}\mrm{Tr}[\mathcal{S}_t^{11;11}(\rho)]\biggr]{\rm d}\bar{W}_t,
\\
{\rm d}\rho_t^{01;01}
&=&
\biggl[\mathcal{D}_{00}(\rho_t^{01;01})+\xi_2^*(t)\mathcal{D}_{10}(\rho_t^{01;00})+\xi_2(t)\mathcal{D}_{01}(\rho_t^{00;01})+|\xi_2(t)|^2\mathcal{D}_{11}(\rho_t^{00;00})\biggr]{\rm d} t
\nonumber
\\
&&+\biggl[\mathcal{S}_t^{01; 01}(\rho)-\rho_t^{01;01}\mrm{Tr}[\mathcal{S}_t^{11;11}(\rho)]\biggr]{\rm d}\bar{W}_t,
\\
{\rm d}\rho_t^{00;11}
&=&
\biggl[\mathcal{D}_{00}(\rho_t^{00;11})+\frac{1}{\sqrt{N_2}}\xi_1^*(t)\mathcal{D}_{10}(\rho_t^{00;01})+\frac{1}{\sqrt{N_2}}\xi_2^*(t)\mathcal{D}_{10}(\rho_t^{00;10})\biggr]{\rm d} t
\nonumber
\\
&&+\biggl[\mathcal{S}_t^{00; 11}(\rho)-\rho_t^{00;11}\mrm{Tr}[\mathcal{S}_t^{11;11}(\rho)]\biggr]{\rm d}\bar{W}_t,
\\
{\rm d}\rho_t^{10;11}
&=&
\biggl[\mathcal{D}_{00}(\rho_t^{10;11})+\frac{1}{\sqrt{N_2}}\xi_1^*(t)\mathcal{D}_{10}(\rho_t^{10;01})+\frac{1}{\sqrt{N_2}}\xi_2^*(t)\mathcal{D}_{10}(\rho_t^{10;10})
\nonumber
\\
&&+\xi_1(t)\mathcal{D}_{01}(\rho_t^{00;11})+\frac{1}{\sqrt{N_2}}|\xi_1(t)|^2\mathcal{D}_{11}(\rho_t^{00;01})+\frac{1}{\sqrt{N_2}}\xi_1(t)\xi_2^*(t)\mathcal{D}_{11}(\rho_t^{00;10})\biggr]{\rm d} t
\nonumber
\\
&&+\biggl[\mathcal{S}_t^{10; 11}(\rho))-\rho_t^{10;11}\mrm{Tr}[\mathcal{S}_t^{11;11}(\rho)]\biggr]{\rm d}\bar{W}_t,
\nonumber
\end{eqnarray*}
\begin{eqnarray*}
{\rm d}\rho_t^{01;11}
&=&
\biggl[\mathcal{D}_{00}(\rho_t^{01;11})+\frac{1}{\sqrt{N_2}}\xi_1^*(t)\mathcal{D}_{10}(\rho_t^{01;01})+\frac{1}{\sqrt{N_2}}\xi_2^*(t)\mathcal{D}_{10}(\rho_t^{01;10})
\nonumber
\\
&&+\xi_2(t)\mathcal{D}_{01}(\rho_t^{00;11})+\frac{1}{\sqrt{N_2}}\xi_1^*(t)\xi_2(t)\mathcal{D}_{11}(\rho_t^{00;01})+\frac{1}{\sqrt{N_2}}|\xi_2(t)|^2\mathcal{D}_{11}(\rho_t^{00;10})\biggr]{\rm d} t
\nonumber
\\
&&+\biggl[\mathcal{S}_t^{01; 11}(\rho)-\rho_t^{01;11}\mrm{Tr}[\mathcal{S}_t^{11;11}(\rho)]\biggr]{\rm d}\bar{W}_t,
\end{eqnarray*}
and
\begin{eqnarray*}
&&\rho_t^{10;00}=(\rho_t^{00;10})^*,~~\rho_t^{01;00}=(\rho_t^{00;01})^*,~~\rho_t^{01;10}=(\rho_t^{10;01})^*,\\
&&\rho_t^{11;00}=(\rho_t^{00;11})^*,~~\rho_t^{11;10}=(\rho_t^{10;11})^*,~~\rho_t^{11;01}=(\rho_t^{01;11})^*,
\end{eqnarray*}
where the innovation process $\bar{W}_t$ is defined as ${\rm d}\bar{W}_t={\rm d} Y(t)-\mrm{Tr}[\mathcal{S}_t^{11;11}(\rho)]{\rm d} t$.
The initial conditions are $\rho_t^{jk;mn}(0)=\langle\Phi_{mn}|\Phi_{jk}\rangle|\eta\rangle\langle\eta|$ for all $ j,k,m,n=0,1$.
\end{corollary}

{\it Remark }8.
Corollary \ref{cor:filter_up} is an immediate consequence of Theorem \ref{thm:2-photon_filter} and equation (\ref{eq:nov11_temp1}).

{\it Example }3. We continue to study the system in Examples 1 and 2 by looking at its 2-photon filter. Here, we wish to calculate the conditional excitation probability (namely, the conditional excited state population) under homodyne detection, which can be expressed as $\mathbb{P}_e^c(t) \triangleq \mbox{Tr} [\rho_t^{11;11}|e\rangle\langle e|]$
where $\rho_t^{11;11}$ is the solution to the filtering equations in Corollary \ref{cor:filter_up}. Individual trajectories $\mathbb{P}_e^c(t)$ are plotted in figures \ref{fig_ex2_a}--\ref{fig_ex2_c}.   For comparison, we also plotted $\mathbb{P}_e(t)$ for the master equations in red solid lines.  It can be seen clearly that in all the three cases, many quantum trajectories can have maximal excitation probability very close to 1, namely, the unit probability.

\begin{figure}[htb!]
\begin{center}
\includegraphics[width=2.5in]{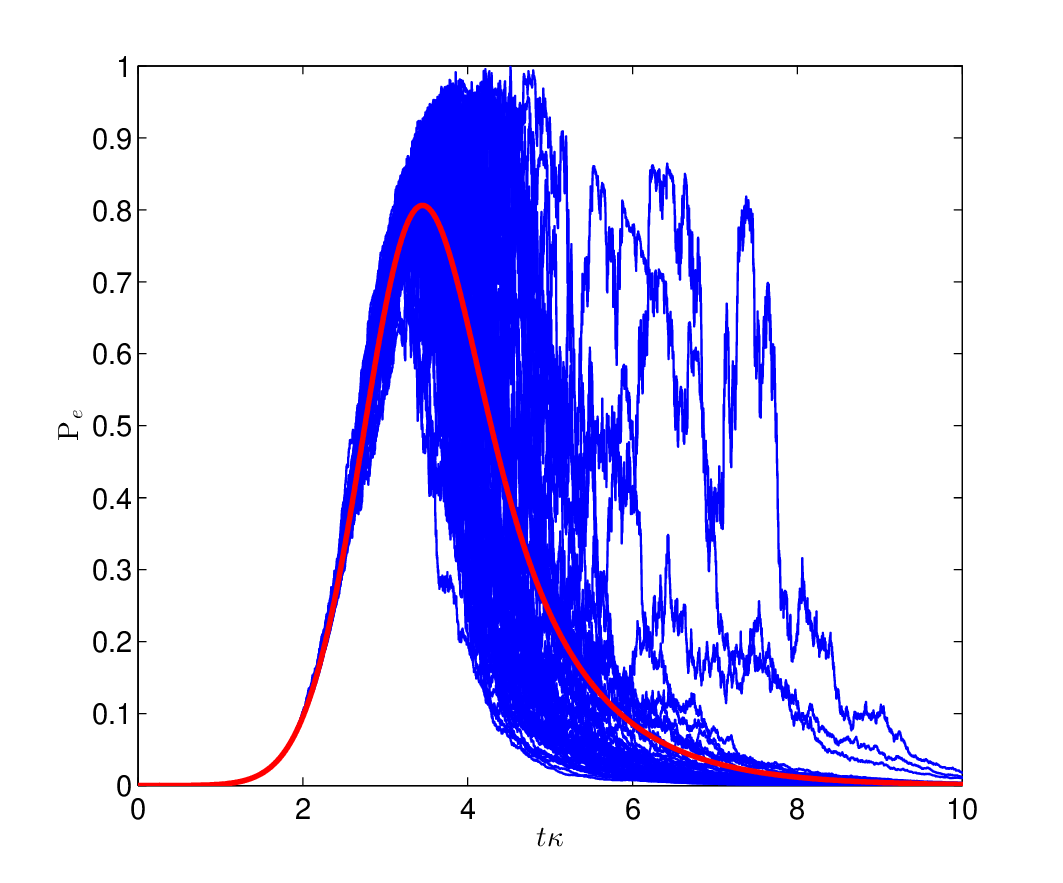}
\caption{The conditional excitation probability of 2-photon filtering for the case when $t_1=t_2=3$ and  $\Omega_1=\Omega_2=1.46\kappa$.  The horizontal axis is time, while the vertical axis is excitation probability.  The red solid line is the unconditional excitation probability  $\mathbb{P}_e(t)$ as calculated by the master equation in Corollary \ref{cor:master_Heisenberg}. The blue lines are individual trajectories of conditional excitation probabilities $\mathbb{P}^c_e(t)$.
}
\label{fig_ex2_a}
\end{center}
\end{figure}

\begin{figure}
\begin{center}
\includegraphics[width=2.5in]{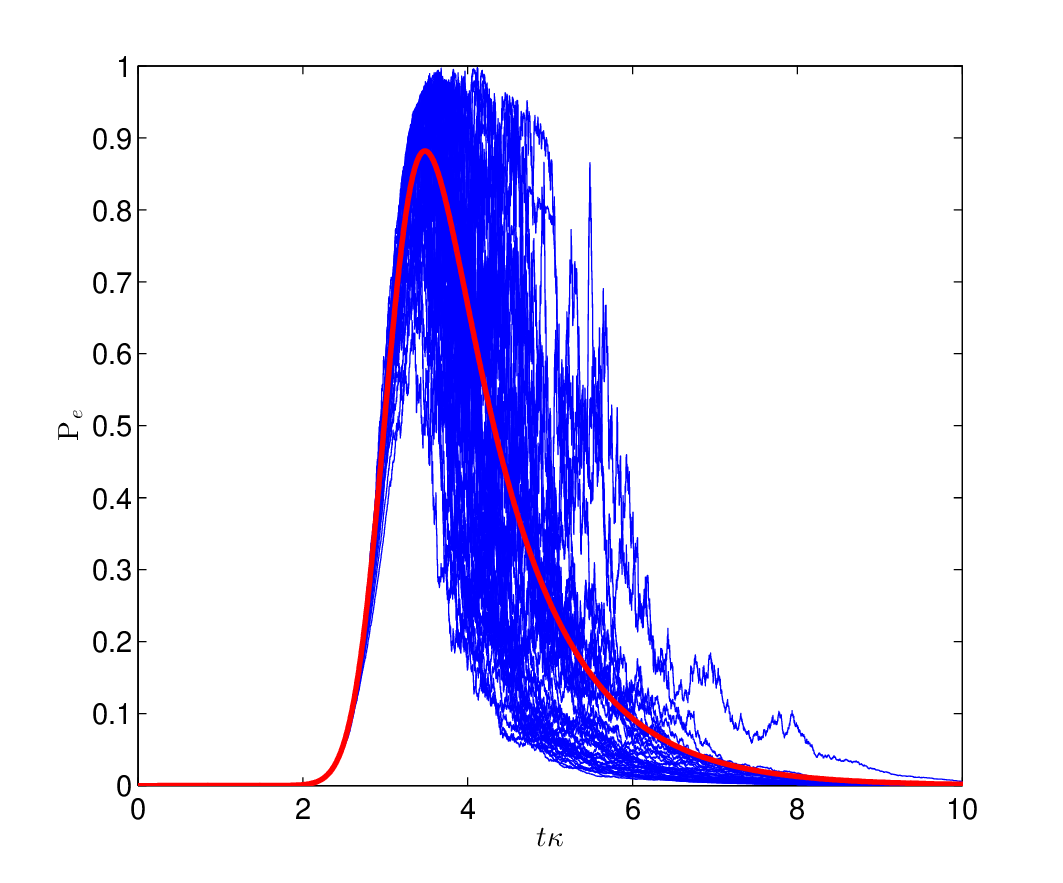}
\caption{The conditional excitation probability of $2$-photon filtering for the case when $t_1=t_2=3$ and $\Omega_1=\Omega_2=2.92\kappa$. The notations of axises and lines are the same to figure \ref{fig_ex2_a}.}
\label{fig_ex2_b}
\end{center}
\end{figure}

\begin{figure}
\begin{center}
\includegraphics[width=2.5in]{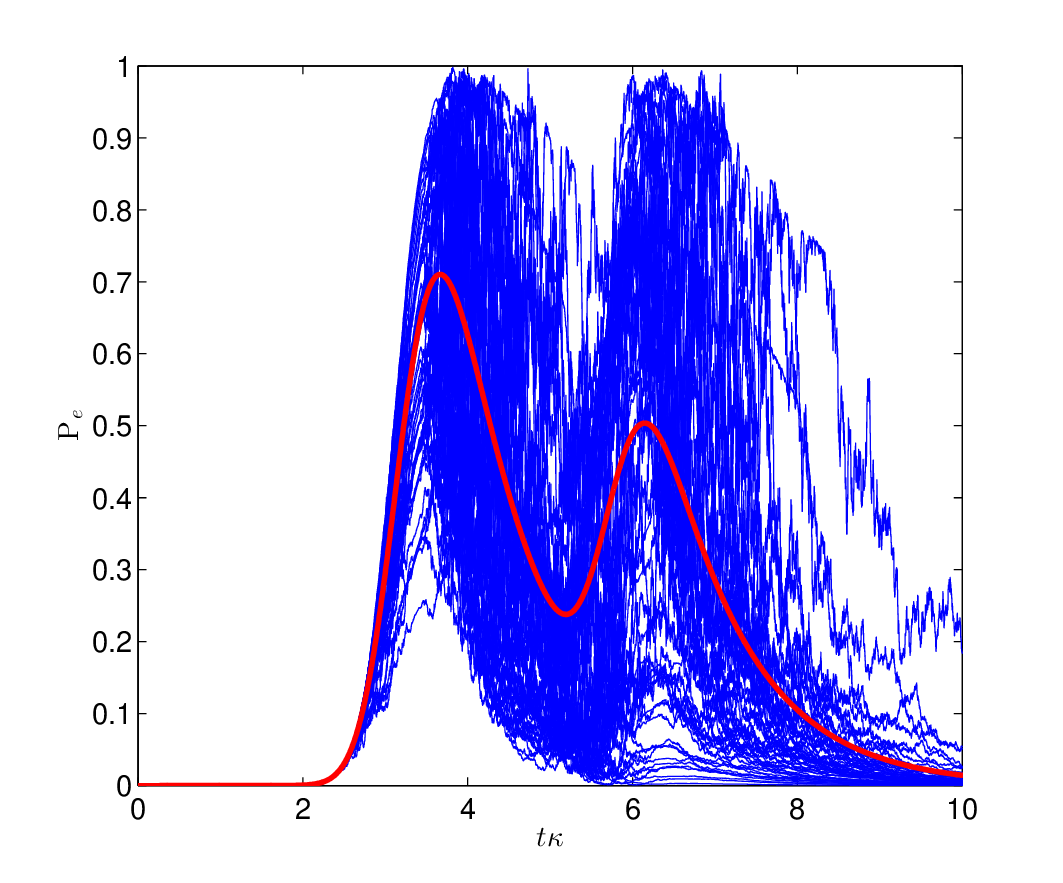}
\caption{The conditional excitation probability of $2$-photon filtering for the case when $t_1=3$, $t_2=5.5$,  and $\Omega_1=\Omega_2=2.92\kappa$.  The notations of axises and lines are the same to figure \ref{fig_ex2_a}.}
\label{fig_ex2_c}
\end{center}
\end{figure}

\section{Multi-photon filtering}\label{sec:multi_photon_filtering}

In this section, we derive filtering equations for an arbitrary quantum system driven by a wavepacket in an $n$-photon state. We follow the logic as carried out in Section \ref{sec:two_photon}. That is, we first derive a filtering equation for an extended system, then present filtering equations for the original system. Multi-photon states are defined in Subsection \ref{subsec:multi_photon_states}. The master equation is presented in Subsection \ref{subsec:multi_photon_master}. Filters in the homodyne detection case are given in Subsection \ref{non_markovian_multi_photon}, and filters in the photon counting case are presented in Subsection \ref{subsec:multi_photon_filter_photoncounting}.

It is worth noting that the notations used in the multi-photon context are slightly different from the 2-photon case, and the following notions turn out very convenient and useful in the derivation of the multi-photon filter. Define a set  $\bar{n} \triangleq \{1,2,\ldots,n\}$. It is implicitly assumed that the elements in each subset of $\bar{n}$ are ordered from the smallest to the largest. Moreover, given a set $R\subset \bar{n}$ and an integer $\mu\in \bar{n}$ but not in $R$, namely, $\mu \in \bar{n}  \setminus R$, define a new ({\it ordered}) subset $R\mu \triangleq R\cup \{\mu\}$ of $\bar{n}$.

Finally, due to page limitation,  all proofs in this section are omitted. 

\subsection{Multi-photon states} \label{subsec:multi_photon_states}
The continuous-mode $n$-photon state is defined as
\begin{eqnarray}\label{def_n_photon}
|\Phi^n\rangle\triangleq\frac{1}{\sqrt{N_n}}\Pi_{j=1}^nB^*(\xi_j)|0\rangle,
\end{eqnarray}
where the superscript $n$ indicates the number of photons, $N_n$ is the normalization coefficient, and $B^\ast(\xi_j)=\int_0^\infty\xi_j(t)b^\ast(t){\rm d} t$ is defined in equation (\ref{eq:B^ast_xi}). This state is completely determined by the set $M_n\triangleq\{\xi_1,\xi_2,\ldots,\xi_n\}$ of $n$ functions in $L_2(\mbb{R}^+,\mbb{C})$. It is worth noting that we distinguish functions in terms of their subscript indices; thus, two (possibly identical) functions with \emph{different} subscript indices are regarded as different functions. For simplicity, we assume all the functions $\xi_k$  are normalized, that is, $\|\xi_k\|=1$ for all $k=1,\ldots,n$. However, these functions are not necessarily orthogonal to each other.   If all the $\xi_i \ (i=1,\ldots,n)$ are equal to $\xi$, the $n$-photon state defined in equation (\ref{def_n_photon}) reduces to  a continuous-mode $n$-photon Fock state:
\begin{eqnarray} \label{eq:Fock_state}
|F^n\rangle \triangleq \frac{1}{\sqrt{n!}}{(B^*(\xi))^n}|0\rangle.
\end{eqnarray}

\subsection{Multi-photon master equation} \label{subsec:multi_photon_master}

In this subsection, we present the master equation of the quantum system $G$ driven by an input field initialized in an $n$-photon state as defined in equation (\ref{def_n_photon}).

For an arbitrary system operator $X$ on the Hilbert space $\mathsf{H}_S$, define its expectation with respect to the  $n$-photon state $|\Phi^n\rangle$ by
\begin{equation}\label{eq:omega^nn}
\omega_t^{n;n}(X)\triangleq\langle\eta\Phi^n|j_t(X)|\eta\Phi^n\rangle  \equiv \mbb{E}_{n;n}[j_t(X)].\nonumber
\end{equation}
Geometrically, $\omega_t^{n;n}(X)$ indicates that, at each time instant $t$, how much information of $j_t(X)$ is contained in the projection space $|\eta\Phi^n\rangle\langle\eta\Phi^n|$.

\begin{figure}
\begin{center}
\includegraphics[width=5.0in]{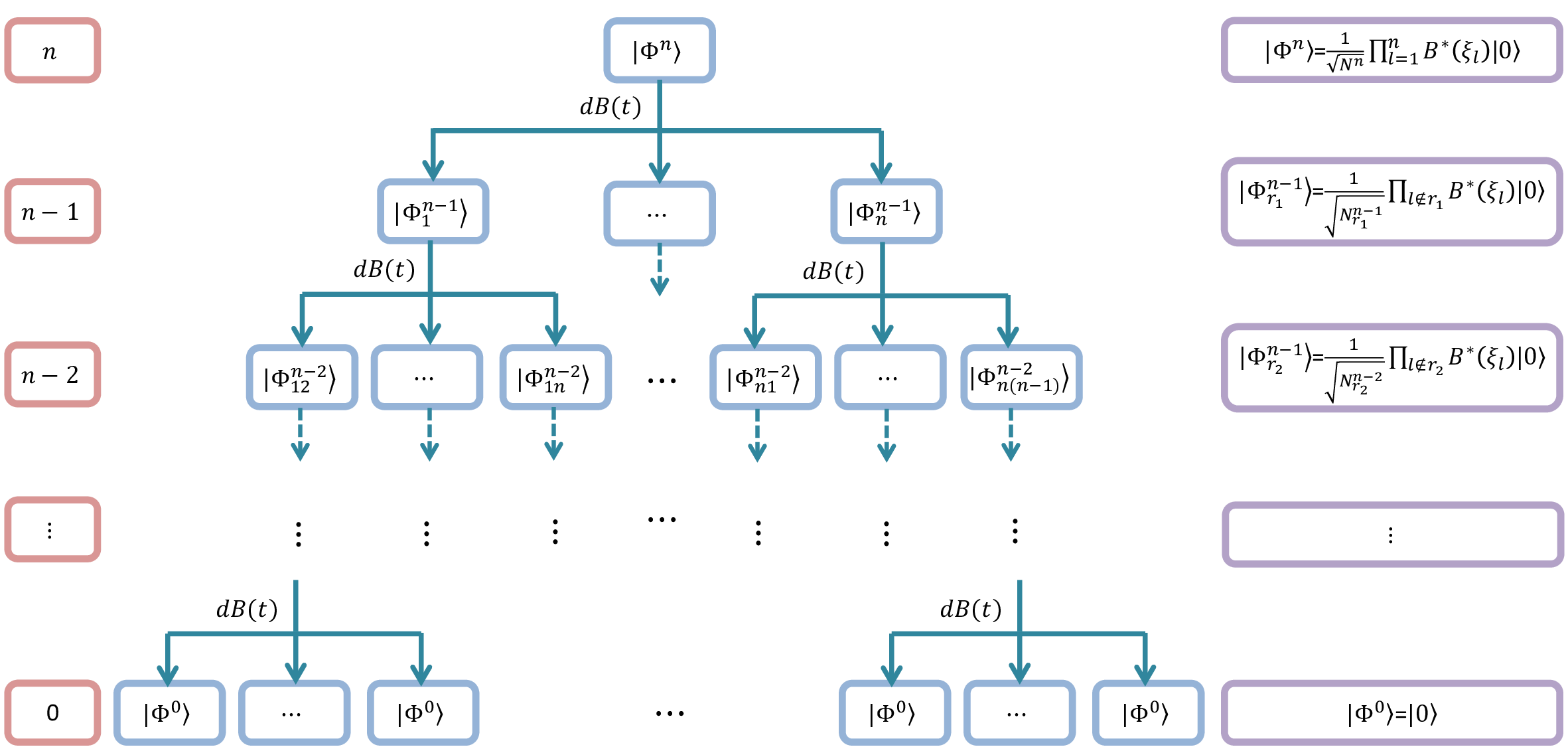}
\caption{The hierarchical structure of the operation of ${\rm d} B(t)$ on photon states.  The left column indicates the number of photons contained in each state on each level. The right column is the general expression of photon states on each level. The column in the middle shows how ${\rm d} B(t)$ acts on various photon states downward, whose ``zoom-in'' version is given in figure \ref{diagram_part}.}
\label{diagram_total}
\end{center}
\end{figure}

In what follows we derive the master equation for $\omega_t^{n;n}(X)$. In analog to (\ref{eq:aug17_temp1}) for the $2$-photon state case, the field operator ${\rm d} B(t)$ acting on the $n$-photon state $|\Phi^n\rangle$ generates $n$ states, each having $n-1$ photons. Similarly, ${\rm d} B(t)$ acting on an $(n-1)$-photon state produces $n-1$ states, each of which has $n-2$ photons; and so on, cf. figure \ref{diagram_total}. As a result, to derive the master equation for $\omega_t^{n;n}(X)$, we have to define the general $(n-k)$-photon states, $k=1,2,\ldots,n$. Moreover, for each $k$, due to the different choices of the functions in the set $M_n$, there are $C_n^k$ different $(n-k)$-photon states. To efficiently distinguish them, we adopt the symbol $|\Phi^{n-k}_{r_k}\rangle$, where the superscript $n-k$ indicates the number of photons, while the subscript $r_k\triangleq \{r^{(1)},r^{(2)},\ldots, r^{(k)}\} \subset \bar{n}$ indicates the set of functions  $M_n \setminus \{\xi_{r^{(1)}},\xi_{r^{(2)}},\ldots,\xi_{r^{(k)}}\}$.  Explicitly, the state $|\Phi^{n-k}_{r_k}\rangle$ is defined as
\begin{eqnarray*} \label{eq:nov11_E_n_k}
|\Phi^{n-k}_{r_k}\rangle \triangleq \frac{1}{\sqrt{N^{n-k}_{r_k}}}\Pi_{l\notin r_k}B^*(\xi_{l})|0\rangle,
\end{eqnarray*}
where $N_{r_k}^{n-k}$ is the corresponding normalization coefficient.  In particular, if $k=n$, then $r_k =\bar{n}$. That is, $|\Phi^{0}_{r_n}\rangle=|0\rangle$ is the vacuum state, cf. the bottom level of the diagram in figure \ref{diagram_total}. When $k=n-1$, $|\Phi^{1}_{r_{n-1}}\r$ is a single-photon state. Because all $\xi_k$ are assumed to be normalized, $N^1_{r_{n-1}}=1$. There are $C_n^{n-1}=n$ such single-photon states which occupy the above-to-bottom level of the diagram in figure \ref{diagram_total}.
 Finally, for notation's convenience, when $k=0$, we denote $r_0  = \emptyset$ (the empty set), and  correspondingly,
$|\Phi^{n}_{r_0}\rangle =  |\Phi^{n}\rangle$, which resides on the top level of the diagram in figure \ref{diagram_total}.

As illustrated in figure \ref{diagram_part}, for the general state $|\Phi^{n-k}_{r_k}\rangle$, we find
\begin{eqnarray} \label{eq:B_Phi_n_k}
{\rm d} B(t)|\Phi^{n-k}_{r_k}\rangle&=&\sum_{m=1}^{n-k}\frac{\sqrt{N^{n-k-1}_{r_kj^{(m)}}}}{\sqrt{N^{n-k}_{r_k}}}\xi_{j^{(m)}}(t)|\Phi^{n-k-1}_{r_k j^{(m)}}\rangle {\rm d} t, \  \ k=0,\ldots, n-1,
\end{eqnarray}
with $\{j^{(1)},j^{(2)},\dots,j^{(n-k)}\}=\bar{n}\setminus r_k$.
Alternatively, the equation (\ref{eq:B_Phi_n_k}) can be re-written as
\begin{eqnarray}\label{eq:B_Phi_n_k_2}
{\rm d} B(t)|\Phi^{n-k}_{r_k}\rangle&=&\sum_{\mu\notin r_k}\frac{\sqrt{N^{n-k-1}_{r_k\mu}}}{\sqrt{N^{n-k}_{r_k}}}\xi_{\mu}(t)|\Phi^{n-k-1}_{r_k \mu}\rangle {\rm d} t, \  \ k=0,\ldots, n-1.
\end{eqnarray}
Finally, when $k=n$,  $r_k =\bar{n}$, $|\Phi^{0}_{r_n}\rangle=|0\rangle$, and thus ${\rm d} B(t)|\Phi^0_{r_n}\rangle = {\rm d} B(t)|0\rangle=0$
which serves as the terminal condition.

As above discussed, to derive the master equation for the quantity $\omega_t^{n;n}(X)$, temporal evolutions of the following quantities
\begin{eqnarray*} \label{eq:nov11_mean__k_l}
\omega_t^{n-j,l_j;n-k,r_k}(X)\triangleq \mathbb{E}_{n-j,l_j;n-k,r_k}[j_t(X)] \equiv  \langle\eta\Phi^{n-j}_{l_j}|j_t(X)|\eta\Phi^{n-k}_{r_k}\rangle, \ \ \forall~l_j, r_k \subset \bar{n}
\end{eqnarray*}
have to be derived.  Once  $j=0$ or $k=0$, the notations $\omega_t^{n-j,l_j;n-k,r_k}$ can be simplified as $\omega_t^{n;n-k,r_k}$ or $\omega_t^{n-j,l_j;n}$, respectively.  Finally, to simplify notation, we make use of $\omega_t^{n-j,l_j;n-k,r_k}(X) \equiv 0$ if either $j>n$ or $k>n$. This notational convention is very handy in our study of multi-photon filtering problem.

\begin{figure}
\begin{center}
\includegraphics[width=5.0in]{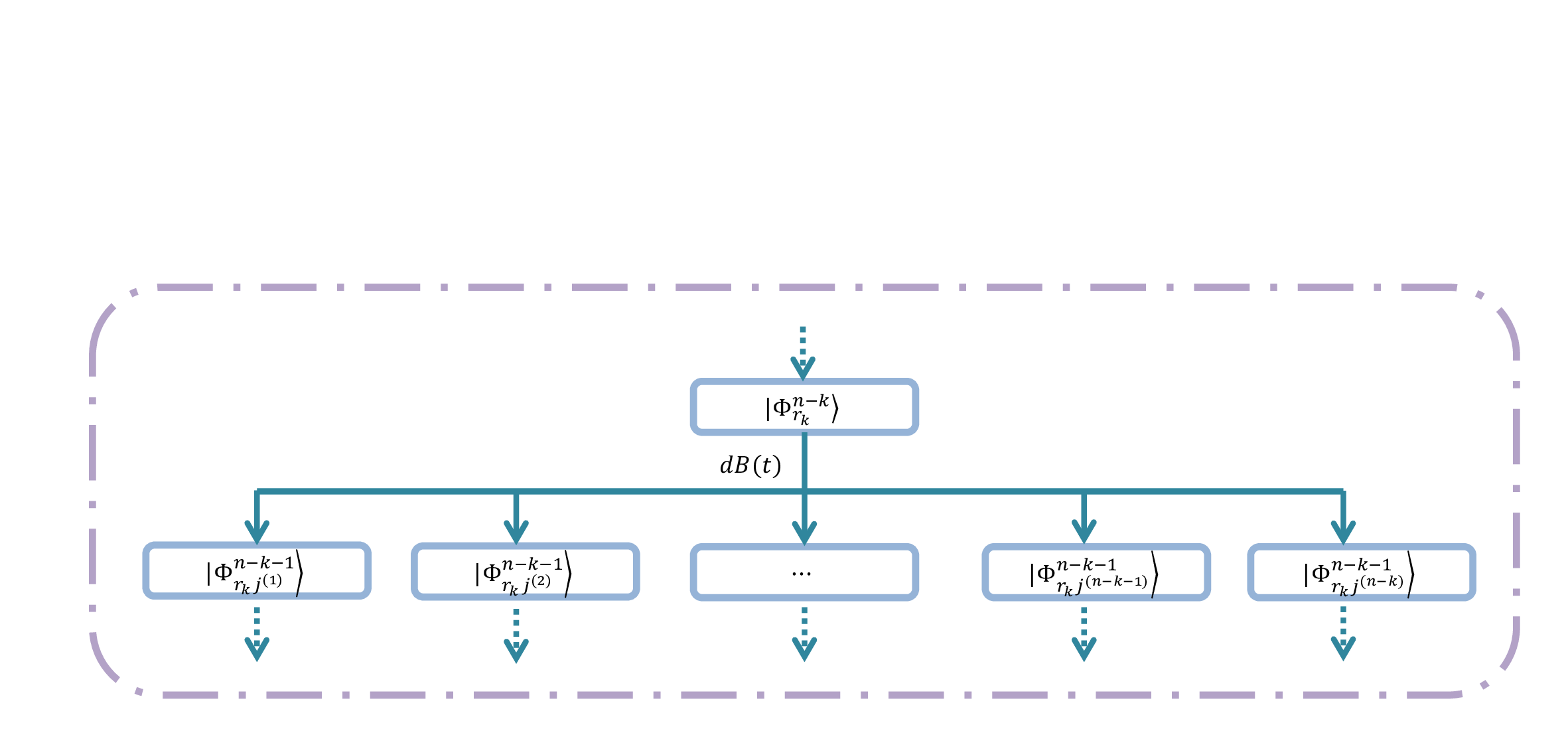}
\caption{The ``zoom-in'' version of figure \ref{diagram_total} .  The action of ${\rm d} B(t)$ on the state $|\Phi_{r_k}^{n-k}\r$ produces $n-k$ states, each of which describes a wavepacket containing $n-k-1$ photons. The subscripts $r_k j^{(i)}$, (i=1,\ldots n-k) are introduced at the end of the second paragraph of Section \ref{sec:multi_photon_filtering}.}
\label{diagram_part}
\end{center}
\end{figure}

From equations (\ref{eq:QDSE_X}) and (\ref{eq:B_Phi_n_k_2}), we can derive the master equation of the system $G$ driven by the $n$-photon state $|\Phi^n\rangle$ as shown by the following theorem, which is the counterpart of Theorem \ref{thm:master_Heisenberg} for the $2$-photon case.


\begin{theorem}\label{master_n}
The master equation in Heisenberg picture for the system $G$ driven by an input field in the $n$-photon state $|\Phi^n\rangle$ is given by the system of differential equations
\begin{eqnarray*}
\dot{\omega}_t^{n;n}(X)
&=&
\sum_{\mu=1}^n\frac{\sqrt{N^{n-1}_{\mu}}}{\sqrt{N_n}}\xi_{\mu}(t)\omega_t^{n;n-1,\mu}(\mathcal{L}_{01}(X))
+\sum_{\nu=1}^n\frac{\sqrt{N^{n-1}_{\nu}}}{\sqrt{N_n}}\xi^*_{\nu}(t)\omega_t^{n-1,\nu;n}(\mathcal{L}_{10}(X))
\nonumber
\\
&&+\sum_{\mu=1}^n\sum_{\nu=1}^n\frac{\sqrt{N^{n-1}_{\mu}}\sqrt{N^{n-1}_{\nu}}}{N_n}\xi_{\mu}(t)\xi^*_{\nu}(t)\omega_t^{n-1,\nu;n-1,\mu}(\mathcal{L}_{11}(X))+\omega_t^{n;n}(\mathcal{L}_{00}(X)),
\end{eqnarray*}
where, for subsets $l_j, r_k \subset \bar{n}$ ~ ($\forall j,k=0,\ldots,n$),
\begin{eqnarray*}
&&\dot{\omega}_t^{n-j,l_j;n-k,r_k}(X)\nonumber\\
&=&\omega_t^{n-j,l_j;n-k,r_k}(\mathcal{L}_{00}(X))+\sum_{\mu\notin r_k}\frac{\sqrt{N^{n-k-1}_{r_k\mu}}}{\sqrt{N^{n-k}_{r_k}}}\xi_{\mu}(t)\omega_t^{n-j,l_j;n-k-1,r_k\mu}(\mathcal{L}_{01}(X))
\nonumber
\\
&&+\sum_{\nu\notin l_j}\frac{\sqrt{N^{n-j-1}_{l_j\nu}}}{\sqrt{N^{n-j}_{l_j}}}\xi^*_{\nu}(t)\omega_t^{n-j-1,l_j\nu;n-k,r_k}(\mathcal{L}_{10}(X))
\nonumber
\\
&&+\sum_{\mu\notin r_k}\sum_{\nu\notin l_j}\frac{\sqrt{N^{n-k-1}_{r_k\mu}}}{\sqrt{N^{n-k}_{r_k}}}\frac{\sqrt{N^{n-j-1}_{l_j\nu}}}{\sqrt{N^{n-j}_{l_j}}}\xi_\mu(t)\xi^*_\nu(t)\omega_t^{n-j-1,l_j\nu;n-k-1,r_k\mu}(\mathcal{L}_{11}(X)),
\end{eqnarray*}
with initial conditions $\omega_0^{n-j,l_j;n-k,r_k}(X)=\langle\eta|X|\eta\rangle\langle\Phi^{n-j}_{l_j}|\Phi^{n-k}_{r_k}\rangle$.
\end{theorem}

{\it Remark }9.
It is clear that the above equations couple downward to the master equation for the vacuum state. This means that for the $n$-photon state case, we should totally consider $2^{2^n}$ equations. Luckily, with the help of the conjugation property
\begin{eqnarray*}
\omega_t^{n-j,l_j;n-k,r_k}(X)=(\omega_t^{n-k,r_k;n-j,l_j}(X^*))^*,
\end{eqnarray*}
 the number of differential equations can be reduced to $\frac{2^n(2^n+1)}{2}$. For example, for the 2-photon case, there are 10 differential equations as shown in Theorem  \ref{thm:master_Heisenberg}.

{\it Remark }10.
Restricted to the Fock state $|F^n\rangle$ defined in equation (\ref{eq:Fock_state}),  the master equation (20) in \cite{baragiola:2012} can be derived from Theorem \ref{master_n}.

\subsection{Multi-photon filter: the homodyne detection case}\label{non_markovian_multi_photon}

In this subsection, we derive the quantum filter for the homodyne detection case.   Following the development in Section \ref{sec:two_photon} for the $2$-photon case, we first derive a filter for the extended system, based on which we derive the filter for the original system.

In analog to Subsection \ref {subsec:extended_system},  we construct an $2^n$-level ancilla for the $n$-photon state case. Specifically,  we choose an orthonormal basis  $\{|e^{n-k}_{r_k}\rangle, ~r_k\subset\bar{n},~k=0,\ldots n\}$ for the vector space $\mbb{C}^{2^n}$ which is defined in the following way. Each $|e^{n-k}_{r_k}\rangle $ has one and only one non-zero entry (which is 1) at the $m$-th location (counted from the top to the bottom). More precisely, if $k=0$, then $m=1$, the vector $|e^{n}_{r_0}\rangle =[1, 0, \ldots, 0]^T$. For $k\geq1$, $m=C_n^0+C_n^1+\cdots+C_n^{k-1}+\Gamma(r_k)$, where $\Gamma(r_k)$ represents the location of the set $r_k \subset \bar{n}$ in the ordered collection of all subsets of $\bar{n}$ having $k$ elements. Here the word ``ordered'' means the lexicographical order \cite{Skiena:1990}. For example, $\{1,2,3\} \prec \{1,2,4\} \prec \{1,3,4\} \prec \{2,3,4\}$.

The extended system is initialized in the superposition state  $|\Sigma^n\rangle \in \mbb{C}^{2^n}\otimes\mathsf{H}_S\otimes\mathsf{F}$:
\begin{eqnarray}\label{sigman}
|\Sigma^n\rangle \triangleq \sum_{k=0}^n\sum_{r_k\subset{\bar{n}}}\alpha^{n-k}_{r_k}|e^{n-k}_{r_k}\eta\Phi^{n-k}_{r_k}\rangle,
\end{eqnarray}
where $\alpha^{n-k}_{r_k}\  (k=0,\ldots,n, r_k\subset \bar{n})$ are arbitrary nonzero numbers that satisfy the normalization condition
$\sum_{k=0}^n\sum_{r_k\subset{\bar{n}}}|\alpha^{n-k}_{r_k}|^2=1$.

For an arbitrary $2^n\times 2^n$ complex matrix $A$ on $\mbb{C}^{2^n}$ and an arbitrary operator $X$ on $H_S$, the expectation with respect to the superposition state $|\Sigma^n\rangle$ is
\begin{equation}
\tilde{\omega}^n_t(A\otimes X) \triangleq \mathbb{E}_{\Sigma^n}[A\otimes j_t(X)].\nonumber
\end{equation}

Define superoperators $\mathcal{K}_{00}^n(A),~\mathcal{K}_{01}^n(A), ~\mathcal{K}_{10}^n(A)$, and $\mathcal{K}_{11}^n(A)$ in the similar way as in  the 2-photon case, cf. (\ref{eq_K_00})--(\ref{eq_K_11}), i.e.,
\begin{eqnarray}\label{K_00_N}
&&\mathcal{K}_{00}^n(A)= A, ~~  \mathcal{K}_{01}^n(A)|\Sigma^n\rangle {\rm d} t=A {\rm d} B(t)|\Sigma^n\rangle,
\\
&&\mathcal{K}_{10}^n(A)=(\mathcal{K}_{01}^n(A^*))^*, ~~ \mathcal{K}_{11}^n(A)=\mathcal{K}^n_{10}(\mathcal{K}^n_{01}(A)).\nonumber
\end{eqnarray}
Then the master equations for $\tilde{\omega}_t^n(A\otimes X)$ are given by the following result.
\begin{theorem}
The expectation $\tilde{\omega}^n_t(A\otimes X)$ for the extended system evolves according to $\dot{\tilde{\omega}}^n_t(A\otimes X)=\tilde{\omega}^n_t(\mathcal{G}^n(A\otimes X))$,
where 
\begin{eqnarray} \label{G^n_AX}
\mathcal{G}^n(A\otimes X)\triangleq\sum_{j,k=0}^1\mathcal{K}_{jk}^n(A)\otimes\mathcal{L}_{jk}(X)
\end{eqnarray}
with $\mathcal{L}_{jk}(X)$ defined in equations (\ref{eq:L_00})--(\ref{eq:L_11}).
\end{theorem}

In the homodyne detection case, define the quantum conditional expectation for the extended system to be $\tilde{\pi}^n_t(A\otimes X)\triangleq\mathbb{E}_{\Sigma^n}[A\otimes j_t(X)|I\otimes \mathscr{Y}(t)]$.
%
The following result presents the quantum filtering equation for the extended system, which is the counterpart of Theorem \ref{thm:filter_extended} for the $2$-photon case.
\begin{theorem}\label{thm_nov17_n_extended}
In the case of homodyne detection, the conditional expectation $\tilde{\pi}_t^n(A\otimes X)$ for the extended system satisfies
\begin{eqnarray*}
{\rm d}\tilde{\pi}_t^n(A\otimes X)=\tilde{\pi}_t^n(\mathcal{G}^n(A\otimes X)){\rm d} t+\tilde{\mathcal{H}}^n_t(A\otimes X){\rm d}\tilde{W}^n(t),
\end{eqnarray*}
where the operator $\mathcal{G}^n(A\otimes X)$ is defined in equation (\ref{G^n_AX}), and
\begin{eqnarray*}
\tilde{\mathcal{H}}_t^n(A\otimes X)=\tilde{M}_t^n(A\otimes X)-\tilde{\pi}_t^n(A \otimes X)\tilde{M}_t^n(I\otimes I)
\end{eqnarray*}
with
\begin{equation}\label{M^n_AX}
\tilde{M}_t^n(A\otimes X)
\triangleq \tilde{\pi}_t^n(\mathcal{K}_{00}^n(A)\otimes (XL+L^*X))+\tilde{\pi}_t^n(\mathcal{K}_{01}^n(A)\otimes XS)+\tilde{\pi}_t^n(\mathcal{K}_{10}^n(A)\otimes S^*X).
\end{equation}
The innovation process $\tilde{W}^n(t)$, defined via ${\rm d}\tilde{W}^n(t)=I\otimes {\rm d} Y(t)-\tilde{M}_t^n(I\otimes I){\rm d} t$, is a Wiener process with respect to the state $|\Sigma^n\rangle$.
\end{theorem}

{\em Proof}.
In analog to the proof of Theorem \ref{thm:filter_extended}, we use the characteristic function method and postulate the form of the filter as
\begin{eqnarray}\label{newe}
{\rm d}\tilde{\pi}_t^n(A\otimes X)=
\tilde{\mathcal{F}}_t^n(A\otimes X)){\rm d} t+\tilde{\mathcal{H}}^n_t(A\otimes X){\rm d}Y(t)
\end{eqnarray}
with the expressions for $\tilde{\mathcal{F}}^n_t(A\otimes X)$ and $\tilde{\mathcal{H}}^n_t(A\otimes X)$ to be determined. Define a stochastic process $c_f(t)={\rm e}^{\int_0^t f(s){\rm d}Y(s)-\frac{1}{2}\int_0^t f^2(s){\rm d}s}$ for an arbitrary function $f\in L_2(\mbb{R}^+, \mbb{C})$. Then it can be verified that this stochastic process $c_f(t)$ satisfies ${\rm d}c_f(t)=f(t)c_f(t){\rm d}Y(t)$ with initial condition $c_f(0)=1$. Since $I\otimes c_f(t)$ is adapted to $I\otimes\mathscr{Y}(t)$, the property of conditional expectations implies that
\begin{eqnarray*}
\mathbb{E}_{\Sigma^n}[(A\otimes j_t(X)) (I\otimes c_f(t))]=\mathbb{E}_{\Sigma^n}[\tilde{\pi}_t^n(A\otimes X)(I\otimes c_f(t))].
\end{eqnarray*}

By differentiating both sides of the above equation and using the properties of conditional expectations, we obtain
\begin{eqnarray*}
&&{\rm d}\mathbb{E}_{\Sigma^n}[A\otimes (j_t(X)c_f(t))]\nonumber\\
&=&\mathbb{E}_{\Sigma^n}[I\otimes c_f(t)A\otimes {\rm d}j_t(X)]+\mathbb{E}_{\Sigma^n}[I\otimes f(t)c_f(t)(A\otimes j_t(X)+A\otimes {\rm d}j_t(X)){\rm d}Y(t)]\nonumber\\
&=&\mathbb{E}_{\Sigma^n}[I\otimes c_f(t)\tilde{\pi}^n_t(\mathcal{G}^n(A\otimes X))]{\rm d}t+\mathbb{E}_{\Sigma^n}[I\otimes f(t)c_f(t)\tilde{M}_t^n(A\otimes X)]{\rm d}t,
\end{eqnarray*}
and
\begin{eqnarray*}
&&{\rm d}\mathbb{E}_{\Sigma^n}[\tilde{\pi}_t^n(A\otimes X)(I\otimes c_f(t))]\nonumber\\
&=&\mathbb{E}_{\Sigma^n}[I\otimes c_f(t){\rm d}\tilde{\pi}_t^n(A\otimes X)]+\mathbb{E}_{\Sigma^n}[I\otimes f(t)c_f(t)(\tilde{\pi}_t^n(A\otimes X)+{\rm d}\tilde{\pi}_t^n(A\otimes X)){\rm d}Y(t)]\nonumber\\
&=&\mathbb{E}_{\Sigma^n}[I\otimes c_f(t)(\tilde{F}^n_t(A\otimes X)+\tilde{H}_t^n(A\otimes X)\tilde{M}_t^n(I\otimes I))]{\rm d}t\nonumber\\
&&+\mathbb{E}_{\Sigma^n}[I\otimes f(t)c_f(t)(\tilde{\pi}^n_t(A\otimes X)\tilde{M}^n_t(I\otimes I)+\tilde{H}_t^n(A\otimes X))]{\rm d}t.
\end{eqnarray*}
Comparing the coefficients of $c_f(t)$ and $f(t)c_f(t)$ respectively, and using the arbitrariness of the function $f(t)$, we find the exact forms of $\tilde{\mathcal{F}}^n_t(A\otimes X)$ and $\tilde{\mathcal{H}}^n_t(A\otimes X)$. Putting them back into equation (\ref{newe}) yields the filter equation in Theorem \ref{thm_nov17_n_extended}.

Now we prove that the innovation process $\tilde{W}^n(t)$ is a Wiener process with respect to the state $|\Sigma^n\rangle$. Firstly, its martingale property is verified. For all $K\in\mathscr{Y}(s),~0\leq s\leq t$,
\begin{eqnarray*}
&&\mathbb{E}_{\Sigma^n}[(\tilde{W}^n(t)-\tilde{W}^n(s))(I\otimes K)]\nonumber\\
&=&\mathbb{E}_{\Sigma^n}[I\otimes (Y(t)-Y(s))(I\otimes K)]-\mathbb{E}_{\Sigma^n}[\int_s^t\tilde{M}^n_r(I\otimes I)(I\otimes K){\rm d} r]\nonumber\\
&=&\mathbb{E}_{\Sigma^n}[I\otimes (Y(t)-Y(s))(I\otimes K)]\nonumber\\
&&-\mathbb{E}_{\Sigma^n}[\int_s^t\tilde{\pi}^n_r(\mathcal{K}_{00}^n(I)\otimes (L+L^*)+\mathcal{K}_{01}^n(I)\otimes S+\mathcal{K}_{10}^n(I)\otimes S^*){\rm d}r(I\otimes K)]\nonumber\\
&=&0.
\end{eqnarray*}
Secondly, it can be readily proved that ${\rm d}\tilde{W}^n(t){\rm d}\tilde{W}^n(t)={\rm d}t$. Then Levy's Theorem indicates that $\tilde{W}^n(t)$ is a Wiener process.
\endproof

In analog to equation (\ref{exp_con_2}), for all $l_j,r_k \subset \bar{n}$, define the conditional quantities $\pi_t^{n-j,l_j;n-k,r_k}(X)$  by means of
\begin{eqnarray}\label{neww2}
\qquad\qquad (I\otimes \pi^{n-j,l_j;n-k,r_k}(X))\tilde{\pi}^n_t(|e^n\rangle\langle e^n|\otimes I)\triangleq\frac{|\alpha^{n}|^2}{(\alpha^{n-j}_{l_j})^\ast\alpha^{n-k}_{r_k}}\tilde{\pi}_t^n(|e^{n-j}_{l_j}\rangle\langle e^{n-k}_{r_k}|\otimes X).
\end{eqnarray}
Similar to (\ref{eq:pi_j_t_X}), it can be shown that
\begin{eqnarray*} \label{eq:nov11_E}
\mathbb{E}_{n;n}[\pi_t^{n-j,l_j;n-k,r_k}(X)K]=\mathbb{E}_{n-j,l_j;n-k,r_k}[j_t(X)K], ~~~~\forall ~K\in\mathscr{Y}(t).
\end{eqnarray*}
In particular, when $j=k=0$,  the above equation reduces to
\begin{equation}
\mathbb{E}_{n;n}[\pi_t^{n;n}(X)K]=\mathbb{E}_{n;n}[j_t(X)K], ~~~~~~~\forall ~K\in\mathscr{Y}(t).\nonumber
\end{equation}
Since $K$ is arbitrary, $\pi_t^{n;n}(X)$ is exactly the conditional system operator of the original system $G$ driven by the $n$-photon state $|\Phi^n\r$, namely,  equation (\ref{eq:def}).


\begin{theorem}\label{homodyne_n_fiter_general}
In the case of homodyne detection, the quantum filter for the conditional expectation $\pi^{n;n}_t(X)$ is given by the following Ito differential equation
\begin{eqnarray*}
{\rm d}\pi_t^{n;n}(X)=L_t^{n;n}(X){\rm d} t+\biggl[M_t^{n;n}(X)-\pi_t^{n;n}(X)M_t^{n;n}(I)\biggr]{\rm d} W^n(t).
\nonumber
\end{eqnarray*}
And more generally for subsets $l_j, r_k \subset \bar{n}$ ~ ($\forall j,l=0,\ldots,n$),
\begin{eqnarray*}
{\rm d}\pi_t^{n-j,l_j;n-k,r_k}(X)&=&L_t^{n-j,l_j;n-k,r_k}(X){\rm d}t
\nonumber
\\
&&+\biggl[M_t^{n-j,l_j;n-k,r_k}(X)-\pi_t^{n-j,l_j;n-k,r_k}(X)M_t^{n;n}(I)\biggr]{\rm d} W^n(t).\nonumber
\end{eqnarray*}
Here,
\begin{eqnarray*}
&&L_t^{n-j,l_j;n-k,r_k}(X)\nonumber\\
&\triangleq&\pi_t^{n-j,l_j;n-k,r_k}(\mathcal{L}_{00}(X))+\sum_{\mu\notin r_k}\frac{\sqrt{N_{r_k\mu}^{n-k-1}}}{\sqrt{N_{r_k}^{n-k}}} \xi_{\mu}(t)\pi_t^{n-j,l_j;n-k-1,r_k\mu}(\mathcal{L}_{01}(X))\nonumber\\
&&+\sum_{\nu\notin l_j}\frac{\sqrt{N_{l_j\nu}^{n-j-1}}}{\sqrt{N_{l_j}^{n-j}}} \xi_{\nu}^*(t)\pi_t^{n-j-1,l_j\nu;n-k,r_k}(\mathcal{L}_{10}(X))
\nonumber
\\
&&+\sum_{\mu\notin r_k}\sum_{\nu\notin l_j}\frac{\sqrt{N_{r_k\mu}^{n-k-1}}}{\sqrt{N_{r_k}^{n-k}}}\frac{\sqrt{N_{l_j\nu}^{n-j-1}}}{\sqrt{N_{l_j}^{n-j}}} \xi_{\mu}(t)\xi_{\nu}^*(t)\pi_t^{n-j-1,l_j\nu;n-k-1,r_k\mu}(\mathcal{L}_{11}(X))\biggr],
\label{eq:aug26_temp_1}
\end{eqnarray*}
and
\begin{eqnarray*}
&&M_t^{n-j,l_j;n-k,r_k}(X)\nonumber\\
&\triangleq&
 \pi_t^{n-j,l_j;n-k,r_k}(XL+L^*X)+\sum_{\mu\notin r_k}\frac{\sqrt{N_{r_k\mu}^{n-k-1}}}{\sqrt{N_{r_k}^{n-k}}} \xi_{\mu}(t)\pi_t^{n-j,l_j;n-k-1,r_k\mu}(XS)
 \nonumber
 \\
&&+\sum_{\nu\notin l_j}\frac{\sqrt{N_{l_j\nu}^{n-j-1}}}{\sqrt{N_{l_j}^{n-j}}} \xi_{\nu}^*(t)\pi_t^{n-j-1,l_j\nu;n-k,r_k}(S^*X).
\end{eqnarray*}
The innovation process $W^n(t)$ defined by ${\rm d}W^n(t)={\rm d}Y(t)-M_t^{n;n}(I){\rm d}t$, is a $\mathscr{Y}(t)$ Wiener process with respect to the $n$-photon state $|\Phi^n\rangle$.
The initial conditions are $\pi_0^{n-j,l_j;n-k,r_k}(X)=\langle\Phi^{n-j}_{l_j}|\Phi^{n-k}_{r_k}\rangle\langle\eta|X|\eta\rangle$.
\end{theorem}

{\em Proof}.
We postulate that the filter for $\pi^{n-j,l_j;n-k,r_k}(X)$ has the following form
\begin{eqnarray}\label{form_cond}
{\rm d}\pi_t^{n-j,l_j;n-k,r_k}(X)=F_t^{n-j,l_j;n-k,r_k}(X){\rm d}t+H_t^{n-j,l_j;n-k,r_k}(X){\rm d}Y(t),
\end{eqnarray}
where the expressions for $F_t^{n-j,l_j;n-k,r_k}(X)$ and $H_t^{n-j,l_j;n-k,r_k}(X)$ are to be determined.

Differentiating both sides of equation (\ref{neww2}), we can get that
\begin{eqnarray}
&&\qquad \qquad \frac{|\alpha^{n}|^2}{(\alpha^{n-j}_{l_j})^\ast\alpha^{n-k}_{r_k}}{\rm d}\tilde{\pi}_t^n(|e^{n-j}_{l_j}\rangle\langle e^{n-k}_{r_k}|\otimes X)
\label{eq:nov21_1}
\\
& =&(I\otimes {\rm d}\pi_t^{n-j,l_j;n-k,r_k}(X))\tilde{\pi}_t^n(|e^n\rangle\langle e^n|\otimes I)+(I\otimes \pi_t^{n-j,l_j;n-k,r_k}(X)){\rm d}\tilde{\pi}_t^n(|e^n\rangle\langle e^n|\otimes I)
\nonumber
\\
& &+(I\otimes {\rm d}\pi_t^{n-j,l_j;n-k,r_k}(X)){\rm d}\tilde{\pi}_t^n(|e^n\rangle\langle e^n|\otimes I). \nonumber
\end{eqnarray}

Based on the definitions of superoperators $\mathcal{K}_{ij}^n(A),~i,j=0,1$ in equation (\ref{K_00_N}) and conditional quantities $\pi_t^{n-j,l_j;n-k,r_k}(X)$ in equation (\ref{neww2}), the following equations can be derived for arbitrary functions $f_i(X),i=1,2,3,4$ of a system operator $X$:
\begin{eqnarray*}
&&\tilde{\pi}_t^n(\mathcal{K}_{00}^n(|e_{l_j}^{n-j}\rangle\langle e_{r_k}^{n-k}|\otimes f_1(X)))
\nonumber
\\
&=&\tilde{\pi}_t^n(|e_{l_j}^{n-j}\rangle\langle e_{r_k}^{n-k}|\otimes f_1(X))
\nonumber
\\
&=&\frac{(\alpha^{n-j}_{l_j})^\ast\alpha^{n-k}_{r_k}}{|\alpha^{n}|^2}(I\otimes \pi_t^{n-j,l_j;n-k,r_k}(f_1(X)))\tilde{\pi}_t^n(|e^n\rangle\langle e^n|\otimes I),
\nonumber
\end{eqnarray*}
\begin{eqnarray*}
&&\tilde{\pi}_t^n(\mathcal{K}_{01}^n(|e_{l_j}^{n-j}\rangle\langle e_{r_k}^{n-k}|\otimes f_2(X)))\nonumber\\
&=&\sum_{\mu\notin r_k}\frac{\alpha_{r_k}^{n-k}}{\alpha_{r_k\mu}^{n-k-1}}\frac{\sqrt{N_{r_k\mu}^{n-k-1}}}{\sqrt{N_{r_k}^{n-k}}} \xi_{\mu}(t)\tilde{\pi}_t^n(|e_{l_j}^{n-j}\rangle\langle e_{r_k\mu}^{n-k-1}|\otimes f_2(X))
\nonumber
\\
&=&\sum_{\mu\notin r_k}\frac{(\alpha^{n-j}_{l_j})^\ast\alpha^{n-k}_{r_k}}{|\alpha^{n}|^2}\frac{\sqrt{N_{r_k\mu}^{n-k-1}}}{\sqrt{N_{r_k}^{n-k}}} \xi_{\mu}(t)(I\otimes\pi_t^{n-j,l_j;n-k-1,r_k\mu}(f_2(X)))\tilde{\pi}_t^n(|e^{n}\rangle\langle e^{n}|\otimes I),
\nonumber
\end{eqnarray*}
\begin{eqnarray*}
&&\tilde{\pi}_t^n(\mathcal{K}_{10}^n(|e_{l_j}^{n-j}\rangle\langle e_{r_k}^{n-k}|\otimes f_3(X)))\nonumber\\
&=&\sum_{\nu\notin l_j}\frac{(\alpha_{l_j}^{n-j})^*}{(\alpha_{l_j\nu}^{n-j-1})^*}\frac{\sqrt{N_{l_j\nu}^{n-j-1}}}{\sqrt{N_{l_j}^{n-j}}} \xi_{\nu}^*(t)\tilde{\pi}_t^n(|e_{l_j\nu}^{n-j-1}\rangle\langle e_{r_k}^{n-k}|\otimes f_3(X))\nonumber\\
&=&\sum_{\nu\notin l_j}\frac{(\alpha^{n-j}_{l_j})^\ast\alpha^{n-k}_{r_k}}{|\alpha^{n}|^2}\frac{\sqrt{N_{l_j\nu}^{n-j-1}}}{\sqrt{N_{l_j}^{n-j}}} \xi_{\nu}^*(t)(I\otimes\pi_t^{n-j-1,l_j\nu;n-k,r_k}(f_3(X)))\tilde{\pi}_t^n(|e^{n}\rangle\langle e^{n}|\otimes I),
\nonumber
\end{eqnarray*}
\begin{eqnarray*}
&&\tilde{\pi}_t^n(\mathcal{K}_{11}^n(|e_{l_j}^{n-j}\rangle\langle e_{r_k}^{n-k}|\otimes f_4(X)))\nonumber\\
&=&\sum_{\nu\notin l_j\atop\mu\notin r_k}\frac{(\alpha_{l_j}^{n-j})^*}{(\alpha_{l_j\nu}^{n-j-1})^*}\frac{\alpha_{r_k}^{n-k}}{\alpha_{r_k\mu}^{n-k-1}}\frac{\sqrt{N_{r_k\mu}^{n-k-1}}}{\sqrt{N_{r_k}^{n-k}}} \frac{\sqrt{N_{l_j\nu}^{n-j-1}}}{\sqrt{N_{l_j}^{n-j}}} \xi_{\nu}^*(t)\xi_{\mu}(t)\tilde{\pi}_t^n(|e_{l_j\nu}^{n-j-1}\rangle\langle e_{r_k\mu}^{n-k-1}|\otimes f_4(X))\nonumber\\
&=&\sum_{\nu\notin l_j \atop \mu\notin r_k}\frac{(\alpha^{n-j}_{l_j})^\ast\alpha^{n-k}_{r_k}}{|\alpha^{n}|^2}\frac{\sqrt{N_{l_j\nu}^{n-j-1}}}{\sqrt{N_{l_j}^{n-j}}}\frac{\sqrt{N_{r_k\mu}^{m-k-1}}}{\sqrt{N_{r_k}^{n-k}}}  \xi_{\nu}^*(t)\xi_{\mu}(t)(I\otimes\pi_t^{n-j-1,l_j\nu;n-k-1,r_k\mu}(f_4(X)))\tilde{\pi}_t^n(|e^{n}\rangle\langle e^{n}|\otimes I) .\nonumber
\end{eqnarray*}
Substituting the above four equations into equations (\ref{G^n_AX}) and (\ref{M^n_AX}) we get
\begin{eqnarray*}
\tilde{\pi}^n_t(\mathcal{G}^n(|e_{l_j}^{n-j}\rangle\langle e_{r_k}^{n-k}|\otimes X))=\frac{(\alpha^{n-j}_{l_j})^\ast\alpha^{n-k}_{r_k}}{|\alpha^{n}|^2}(I\otimes L_t^{n-j,l_j;n-k,r_k}(X))\tilde{\pi}_t^n(|e^n\rangle\langle e^n|\otimes I).
\end{eqnarray*}
and
\begin{eqnarray*}
\tilde{M}_t^n(|e_{l_j}^{n-j}\rangle\langle e_{r_k}^{n-k}|\otimes X)=\frac{(\alpha^{n-j}_{l_j})^\ast\alpha^{n-k}_{r_k}}{|\alpha^{n}|^2}(I\otimes M_t^{n-j,l_j;n-k,r_k}(X))\tilde{\pi}_t^n(|e^n\rangle\langle e^n|\otimes I).
\end{eqnarray*}
By Theorem \ref{thm_nov17_n_extended},
\begin{eqnarray*}
&&\frac{|\alpha^{n}|^2}{(\alpha^{n-j}_{l_j})^\ast\alpha^{n-k}_{r_k}}{\rm d}\tilde{\pi}_t^n(|e^{n-j}_{l_j}\rangle\langle e^{n-k}_{r_k}|\otimes X)\nonumber\\
&=&I\otimes L_t^{n-j,l_j;n-k,r_k}(X)\tilde{\pi}_t^n(|e^n\rangle\langle e^n|\otimes I){\rm d}t+ \bigg [I\otimes M_t^{n-j,l_j;n-k,r_k}(X)\nonumber\\
&&-I\otimes
\pi_t^{n-j,l_j;n-k,r_k}(X)\tilde{M}^n_t(I\otimes I)\bigg ]\tilde{\pi}_t^n
(|e^n\rangle\langle e^n|\otimes I){\rm d}\tilde{W}^n(t),\nonumber
\end{eqnarray*}
and correspondingly,
\begin{eqnarray*}
{\rm d}\tilde{\pi}_t^n(|e^n\rangle\langle e^n|\otimes I)=[I\otimes M_t^{n;n}(I)-\tilde{M}^n_t(I\otimes I)]\tilde{\pi}_t^n(|e^n\rangle\langle e^n|\otimes I){\rm d}\tilde{W}^n(t),
\end{eqnarray*}
where we have used $\mathcal{L}_{ij}(I)=0$ for $i,j=0,1$.

Similarly, for the right-hand side of equation \ref{eq:nov21_1}, we have
\begin{eqnarray*}
&&(I\otimes {\rm d}\pi_t^{n-j,l_j;n-k,r_k}(X))\tilde{\pi}_t^n(|e^n\rangle\langle e^n|\otimes I)\nonumber\\
&=&(I\otimes F_t^{n-j,l_j;n-k,r_k}(X))\tilde{\pi}_t^n(|e^n\rangle\langle e^n|\otimes I){\rm d}t\nonumber\\
&&+(I\otimes H_t^{n-j,l_j;n-k,r_k}(X))\tilde{\pi}_t^n(|e^n\rangle\langle e^n|\otimes I){\rm d}Y(t),\nonumber
\end{eqnarray*}
\begin{eqnarray*}
&&(I\otimes \pi_t^{n-j,l_j;n-k,r_k}(X)){\rm d}\tilde{\pi}^t_n(|e^n\rangle\langle e^n|\otimes I)\nonumber\\
&=&(I\otimes \pi_t^{n-j,l_j;n-k,r_k}(X))[I\otimes M_t^{n;n}(I)-\tilde{M}^n_t(I\otimes I)]\tilde{\pi}_t^n(|e^n\rangle\langle e^n|\otimes I){\rm d}\tilde{W}^n(t),\nonumber
\end{eqnarray*}
\begin{eqnarray*}
&&(I\otimes {\rm d}\pi_t^{n-j,l_j;n-k,r_k}(X)){\rm d}\tilde{\pi}_t^n(|e^n\rangle\langle e^n|\otimes I)\nonumber\\
&=&(I\otimes H_t^{n-j,l_j;n-k,r_k}(X))[I\otimes M_t^{n;n}(I)-\tilde{M}^n_t(I\otimes I)]\tilde{\pi}_t^n(|e^n\rangle\langle e^n|\otimes I){\rm d}t,
\end{eqnarray*}
and ${\rm d}\tilde{W}^n(t)={\rm d}t-\tilde{M}^n_t(I\otimes I){\rm d}Y(t)$. By comparing the coefficients of ${\rm d}t$ and ${\rm d}Y(t)$, we can get that
\begin{eqnarray*}
&&I\otimes F_t^{n-j,l_j;n-k,r_k}(X)+(I\otimes \pi_t^{n-j,l_j;n-k,r_k}(X))[I\otimes M_t^{n;n}(I)-\tilde{M}_t^n(I\otimes I)]\tilde{M}_t^n(I\otimes I)\nonumber\\
&&+I\otimes H_t^{n-j,l_j;n-k,r_k}(X)[I\otimes M_t^{n;n}(I)-I\otimes
\pi_t^{n;n}(I)\tilde{M}^n_t(I\otimes I)]\nonumber\\
&=&-[I\otimes M_t^{n-j,l_j;n-k,r_k}(X)-I\otimes
\pi_t^{n-j,l_j;n-k,r_k}(X)\tilde{M}^n_t(I\otimes I)]\tilde{M}^n_t(I\otimes I)\nonumber\\
&&+I\otimes L_t^{n-j,l_j;n-k,r_k}(X),\nonumber\\
\end{eqnarray*}
and
\begin{eqnarray*}
&&I\otimes M_t^{n-j,l_j;n-k,r_k}(X)-I\otimes
\pi_t^{n-j,l_j;n-k,r_k}(X)\tilde{M}^n_t(I\otimes I)\nonumber\\
&=&I\otimes H_t^{n-j,l_j;n-k,r_k}(X)+I\otimes \pi_t^{n-j,l_j;n-k,r_k}(X)[I\otimes M_t^{n;n}(I)-\tilde{M}^n_t(I\otimes I)].
\end{eqnarray*}
By solving the above two equations, the final expressions of $F^{n-j,l_j;n-k,r_k}(X)$ and $H^{n-j,l_j;n-k,r_k}(X)$ can be obtained as follows
\begin{eqnarray*}
&&H_t^{n-j,l_j;n-k,r_k}(X)=M_t^{n-j,l_j;n-k,r_k}(X)-\pi_t^{n-j,l_j;n-k,r_k}(X)M^{n;n}(I),\nonumber\\
&&F_t^{n-j,l_j;n-k,r_k}(X)=L_t^{n-j,l_j;n-k,r_k}(X)-H_t^{n-j,l_j;n-k,r_k}(X)M^{n;n}(I).
\end{eqnarray*}
Putting them back into the equation (\ref{form_cond}), we can get the quantum filter.

Now we prove the martingale property of the innovation process $W^n(t)$, that is equivalent to prove $\mathbb{E}_{n;n}[(W(t)-W(s))(I\otimes K)]=0$ for all $K\in\mathscr{Y}(s), 0\leq s\leq t$. Obviously,
\begin{eqnarray*}
&&\mathbb{E}_{n;n}[({W}^n(t)-{W}^n(s))(I\otimes K)]\nonumber\\
&=&\mathbb{E}_{n;n}[I\otimes (Y(t)-Y(s))(I\otimes K)]-\mathbb{E}_{n;n}[\int_s^tM^{n;n}_r(I\otimes I)(I\otimes K){\rm d} r]\nonumber\\
&=&\mathbb{E}_{n;n}[I\otimes (Y(t)-Y(s))(I\otimes K)]\nonumber\\
&&-\mathbb{E}_{n;n}\biggl[\int_s^t{\pi}^{n;n}_r(L+L^*)+\sum_{\mu\notin r_k}\frac{\sqrt{N_{r_k\mu}^{n-k-1}}}{\sqrt{N_{r_k}^{n-k}}} \xi_{\mu}(r)\pi_r^{n-j,l_j;n-k-1,r_k\mu}(S)
 \nonumber
 \\
&&+\sum_{\nu\notin l_j}\frac{\sqrt{N_{l_j\nu}^{n-j-1}}}{\sqrt{N_{l_j}^{n-j}}} \xi_{\nu}^*(r)\pi_r^{n-j-1,l_j\nu;n-k,r_k}(S^*){\rm d}r(I\otimes K)\biggr]\nonumber\\
&=&0.
\end{eqnarray*}
Since ${\rm d}{W}^n(t){\rm d}{W}^n(t)={\rm d}t$, Levy's Theorem implies that ${W}^n(t)$ is a Wiener process with respect to the $n$-photon state.
\endproof

{\it Remark }11.
It can be verified that Theorem \ref{homodyne_n_fiter_general} reduces to Theorem \ref{thm:2-photon_filter} when $n=2$, namely the 2-photon case.

\subsection{Multi-photon filter: the photon-counting case}\label{subsec:multi_photon_filter_photoncounting}

In this subsection, we present the multi-photon filter for the photon-counting case by deriving the evolution of the quantum conditional expectation $\hat{\pi}_t^{n;n}(X) \triangleq \mathbb{E}_{n;n}[j_t(X)|\mathscr{Y}^\Lambda(t)]$.

Similar to the development in Subsection \ref{non_markovian_multi_photon}, we first need extend the system and initialize the state as $|\Sigma^n\rangle$ in equation (\ref{sigman}), then we calculate the filtering equation for the extended system, i.e., the evolution of $\tilde{\pi}_t^\Lambda(A\otimes X)\triangleq \mathbb{E}_{\Sigma^n}[A\otimes j_t(X)|I\otimes \mathscr{Y}^\Lambda(t)]$. Finally, with the relationship between $\hat{\pi}_t^{n;n}(X)$ and $\tilde{\pi}_t^\Lambda(A\otimes X)$, we derive the quantum filter for $\hat{\pi}_t^{n;n}(X)$.


The following theorem is the filtering equation for the extended system, which is the counterpart of Theorem \ref{thm_nov17_n_extended}.

\begin{theorem}
In the case of photon-counting monitoring, the conditional expectation $\tilde{\pi}_t^\Lambda(A\otimes X)$ for the extended system satisfies
\begin{eqnarray*}
{\rm d}\tilde{\pi}_t^\Lambda (A\otimes X)=\tilde{\pi}_t^\Lambda(\mathcal{G}^n(A\otimes X)){\rm d} t+\tilde{H}_t^\Lambda(A\otimes X){\rm d} \tilde{N}^\Lambda(t),
\end{eqnarray*}
where $(\mathcal{G}^n(A\otimes X)$ is defined in equation (\ref{G^n_AX}) and
\begin{eqnarray*}
\tilde{H}_t^\Lambda(A\otimes X)=(\tilde{Q}_t^\Lambda(I\otimes I))^{-1}\tilde{Q}_t^\Lambda(A\otimes X)-\tilde{\pi}_t^\Lambda(A\otimes X),
\end{eqnarray*}
with
\begin{eqnarray*}
\tilde{Q}_t^\Lambda(A\otimes X)&\triangleq&\tilde{\pi}_t^\Lambda(\mathcal{K}_{00}^n(A)\otimes L^*XL)+\tilde{\pi}_t^\Lambda(\mathcal{K}_{01}^n(A)\otimes L^*XS)+\tilde{\pi}_t^\Lambda(\mathcal{K}_{10}^n(A)\otimes S^*XL)\nonumber\\
&&+\tilde{\pi}_t^\Lambda(\mathcal{K}_{11}^n(A)\otimes S^*XS).
\end{eqnarray*}
The innovation process is given as ${\rm d}\tilde{N}^\Lambda(t)={\rm d} Y^\Lambda(t)-\tilde{Q}_t^\Lambda(I\otimes I){\rm d} t$.
\end{theorem}

{\em Proof}.
We postulate the form of the filter as
\begin{eqnarray} \label{eq:nov21_2}
{\rm d}\tilde{\pi}_t^\Lambda(A\otimes X)=\tilde{F}^\Lambda_t(A\otimes X){\rm d}t+\tilde{H}^\Lambda_t(A\otimes X){\rm d}Y^\Lambda(t),
\end{eqnarray}
where the expressions of $\tilde{F}^\Lambda_t(A\otimes X)$ and $\tilde{H}^\Lambda_t(A\otimes X)$ are to be determined.

For any positive function $\hat{f}(t)$ such that $\ln[\hat{f}(t)+1]$ is integrable, define a stochastic process $\hat{c}_f(t)$ as $\hat{c}_f(t)={\rm e}^{h(t)}\triangleq {\rm e}^{\int_0^t\ln[\hat{f}(s)+1]{\rm d}Y^\Lambda(s)}$. In what follows we verify that $\hat{c}_f(t)$  satisfies ${\rm d}\hat{c}_f(t)=\hat{f}(t)\hat{c}_f(t){\rm d}Y^\Lambda(t)$ with the initial condition $\hat{c}_f(0)=1$. Specifically, notice that
\begin{eqnarray*}
{\rm d}\hat{c}_f(t)&=&\hat{c}_f(t+{\rm d}t)-\hat{c}_f(t)\nonumber\\
&=&{\rm e}^{h(t+{\rm d}t)}-{\rm e}^{h(t)}={\rm e}^{h(t)+{\rm d}h(t)}-{\rm e}^{h(t)}={\rm e}^{h(t)}[{\rm e}^{{\rm d}h(t)}-1]\nonumber\\
&=&{\rm e}^{h(t)}[1+{\rm d}h(t)+\frac{1}{2!}({\rm d}h(t))^2+\cdots+\frac{1}{k!}({\rm d}h(t))^k+\cdots-1].
\end{eqnarray*}
Since ${\rm d}h(t)=\ln[\hat{f}(t)+1]{\rm d}Y^\Lambda(t)$ and  $({\rm d}h(t))^k=\ln^k[\hat{f}(t)+1]{\rm d}Y^\Lambda(t)$, we can get
\begin{eqnarray*}
{\rm d}\hat{c}_f(t)&=&{\rm e}^{h(t)}[{\rm e}^{\ln[\hat{f}(t)+1]}-1]{\rm d}Y^\Lambda(t)=\hat{c}_f(t)\hat{f}(t){\rm d}Y^\Lambda(t).
\end{eqnarray*}

Using the property of conditional expectations, it can be easily verified that
\begin{eqnarray*}
\mathbb{E}_{\Sigma^n}[A\otimes j_t(X)\hat{c}_f(t)]=\mathbb{E}_{\Sigma^n}[\tilde{\pi}_t^\Lambda(A\otimes X)(I\otimes \hat{c}_f(t))].
\end{eqnarray*}
When we differentiate both sides of the above equation, take expectations and conditional expectations, we can obtain
\begin{eqnarray*}
&&{\rm d}\mathbb{E}_{\Sigma^n}[A\otimes j_t(X)\hat{c}_f(t)]\nonumber\\
&=&\mathbb{E}_{\Sigma^n}[I\otimes \hat{c}_f(t) A\otimes {\rm d}j_t(X)]+\mathbb{E}_{\Sigma^n}[I\otimes \hat{f}(t)\hat{c}_f(t) (A\otimes j_t(X)+ A\otimes {\rm d}j_t(X)){\rm d}Y^\Lambda(t)]\nonumber\\
&=&\mathbb{E}_{\Sigma^n}[I\otimes \hat{c}_f(t) \tilde{\pi}_t^\Lambda(\mathcal{G}^n(A\otimes X))]{\rm d}t+\mathbb{E}_{\Sigma^n}[(I\otimes \hat{f}(t)\hat{c}_f(t)) \tilde{Q}_t^\Lambda(A\otimes X)]{\rm d}t,
\end{eqnarray*}
and
\begin{eqnarray*}
&&{\rm d}\mathbb{E}_{\Sigma^n}[\tilde{\pi}_t^\Lambda(A\otimes X)I\otimes\hat{c}_f(t)]\nonumber\\
&=&\mathbb{E}_{\Sigma^n}[I\otimes \hat{c}_f(t)(\tilde{F}_t^\Lambda(A\otimes X){\rm d}t+\tilde{H}_t^\Lambda(A\otimes X){\rm d}Y^\Lambda(t))]\nonumber\\
&&+\mathbb{E}_{\Sigma^n}[I\otimes \hat{f}(t)\hat{c}_f(t)(\tilde{\mathcal{\pi}}_t^\Lambda(A\otimes X){\rm d}Y^\Lambda(t)+\tilde{H}_t^\Lambda(A\otimes X){\rm d}Y^\Lambda(t))]\nonumber\\
&=&\mathbb{E}_{\Sigma^n}[I\otimes \hat{c}_f(t)(\tilde{F}_t^\Lambda(A\otimes X)+\tilde{H}_t^\Lambda(A\otimes X)\tilde{Q}_t^\Lambda(I\otimes I))]{\rm d}t\nonumber\\
&&+\mathbb{E}_{\Sigma^n}[I\otimes \hat{f}(t)\hat{c}_f(t)(\tilde{\mathcal{\pi}}_t^\Lambda(A\otimes X)\tilde{Q}_t^\Lambda(I\otimes I)+\tilde{H}_t^\Lambda(A\otimes X)\tilde{Q}_t^\Lambda(I\otimes I))]{\rm d}t.
\end{eqnarray*}
Comparing the coefficients of $\hat{c}_f(t)$ and $\hat{f}(t)\hat{c}_f(t)$, and noting the arbitrariness of the function $\hat{f}(t)$, we can get the final expressions of $\tilde{F}_t^\Lambda(A\otimes X)$ and $\tilde{H}^\Lambda_t(A\otimes X)$ as follows,
\begin{eqnarray*}
&&\tilde{H}^\Lambda_t(A\otimes X)=(\tilde{Q}_t^\Lambda(I\otimes I))^{-1}\tilde{Q}_t^\Lambda(A\otimes X)-\tilde{\pi}_t^\Lambda(A\otimes X),\nonumber\\
&&\tilde{F}^\Lambda_t(A\otimes X)=\tilde{\pi}_t^\Lambda(\mathcal{G}^n(A\otimes X))-\tilde{H}_t^\Lambda(A\otimes X)\tilde{Q}_t^\Lambda(I\otimes I).
\end{eqnarray*}
Putting them back to equation \ref{eq:nov21_2} we see that the conditional expectation $\tilde{\pi}^\Lambda_t(A\otimes X)$ satisfies
\begin{eqnarray*}
{\rm d}\tilde{\pi}_t^\Lambda=\tilde{\pi}_t^\Lambda(\mathcal{G}^n(A\otimes X)){\rm d} t+\tilde{H}_t^\Lambda(A\otimes X){\rm d} \tilde{N}^\Lambda(t),
\end{eqnarray*}
with the innovation process ${\rm d}\tilde{N}^\Lambda(t)={\rm d}t-\tilde{Q}_t^\Lambda(I\otimes I){\rm d}Y^\Lambda(t)$.      \endproof


Defining
\begin{eqnarray}\label{conditional_photon_n}
\qquad\qquad (I\otimes \hat{\pi}_t^{n-j,l_j;n-k,r_k}(X))\tilde{\pi}_t^\Lambda(|e^n\rangle\langle e^n|\otimes I)=\frac{|\alpha^n|^2}{(\alpha_{l_j}^{n-j})^\ast\alpha_{r_k}^{n-k}}\tilde{\pi}_t^\Lambda(|e^{n-j}_{l_j}\rangle\langle e_{r_k}^{n-k}|\otimes X),
\end{eqnarray}
we can directly verify the following equation
\begin{eqnarray*}
\mathbb{E}_{n;n}[\hat{\pi}_t^{n-j,l_j;n-k,r_k}(X)K]\!=\!\mathbb{E}_{n-j,l_j;n-k,r_k}[j_t(X)K],~~~\forall K\in\mathscr{Y}^\Lambda(t).
\end{eqnarray*}
If we set $j=k=0$, and note that $K\in\mathscr{Y}^\Lambda(t)$ is arbitrary, we can deduce that $\hat{\pi}_t^{n;n}(X)$ defined in equation (\ref{conditional_photon_n}) is exactly the desired conditional expectation with respect to the $n$-photon field state $|\Phi^n\r$ for the photon detection, namely equation (\ref{eq:poisson}). 

The following result presents the quantum filter for photodetection, the counterpart of Theorem  \ref{homodyne_n_fiter_general}.

\begin{theorem}\label{thm:n_filter_PD}
In the case of photon-counting measurement, the quantum filter for the
conditional expectation $\hat{\pi}_t^{n;n}(X)$ is given by the following Ito differential equation
\begin{eqnarray*}
{\rm d}\hat{\pi}_{t}^{n;n}(X)&=&\hat{P}_t^{n;n}(X){\rm d} t+\biggl[\left( \Delta _{t}^{n;n}(I)\right) ^{-1}\Delta _{t}^{n;n}(X)-\hat{\pi}_{t}^{n;n}(X)\biggr]{\rm d} N_{t}.
\end{eqnarray*}
And more generally for subsets $l_j,r_k\subset \bar{n}$ ~ ($\forall j,k=0,\ldots,n$),
\begin{eqnarray*}
&&{\rm d}\hat{\pi}_{t}^{n-j,l_{j};n-k,r_{k}}(X)\nonumber\\
&=&\hat{P}_t^{n-j,l_j;n-k,r_k}(X){\rm d} t
+\biggl[\left( \Delta _{t}^{n;n}(I)\right) ^{-1}\Delta
_{t}^{n-j,l_{j};n-k,r_{k}}(X)-\hat{\pi}_{t}^{n-j,l_{j};n-k,r_{k}}(X)\biggr]{\rm d} N_{t},
\end{eqnarray*}
where
\begin{eqnarray*}
&&\hat{P}_t^{n-j,l_j;n-k,r_k}(X)\nonumber\\
&\triangleq&\hat{\pi}
_{t}^{n-j,l_{j};n-k,r_{k}}(\mathcal{L}_{00}(X))+\sum_{\mu \notin r_{k}}\frac{\sqrt{N_{r_{k}\mu }^{n-k-1}}}{\sqrt{N_{r_{k}}^{n-k}}}\xi _{\mu }(t)\hat{\pi}_{t}^{n-j,l_{j};n-k-1,r_{k}\mu }(\mathcal{L}_{01}(X))
\nonumber
\\
&&+\sum_{\nu \notin l_{j}}\frac{\sqrt{N^{n-j-1}_{l_j\nu}}}{\sqrt{N_{l_{j}}^{n-j}}}\xi _{\nu }^{\ast }(t)\hat{\pi}_{t}^{n-j-1,l_{j}\nu
;n-k,r_{k}}(\mathcal{L}_{10}(X))
\nonumber
\\
&&+\sum_{\mu \notin r_{k}\atop\nu \notin l_{j}}\frac{\sqrt{N_{r_{k}\mu}^{n-k-1}}}{\sqrt{N_{r_{k}}^{n-k}}}\frac{\sqrt{N_{l_{j}\nu }^{n-j-1}}}{\sqrt{N_{l_{j}}^{n-j}}}\xi _{\nu }^{\ast }(t)\xi _{\mu }(t)\hat{\pi}_{t}^{n-j-1,l_{j}\nu ;n-k-1,r_{k}\mu }(\mathcal{L}_{11}(X)),
\end{eqnarray*}
and
\begin{eqnarray*}
&&\Delta _{t}^{n-j,l_{j};n-k,r_{k}}(X)\nonumber\\
&\triangleq& \hat{\pi}_{t}^{n-j,l_{j};n-k,r_{k}}(L^{\ast }XL)+\sum_{\mu
\notin r_{k}}\frac{\sqrt{N_{r_{k}\mu }^{n-k-1}}}{\sqrt{N_{r_{k}}^{n-k}}}\xi
_{\mu }(t)\hat{\pi}_{t}^{n-j,l_{j};n-k-1,r_{k}\mu }(L^{\ast }XS)\nonumber\\
&&+\sum_{\nu
\notin l_{j}}\frac{\sqrt{N_{l_{j}\nu }^{n-j-1}}}{\sqrt{N_{l_{j}}^{n-j}}}\xi
_{\nu }^{\ast }(t)\hat{\pi}_{t}^{n-j-1,l_{j}\nu ;n-k,r_{k}}(S^{\ast }XL)
\nonumber
\\
&&+\sum_{\mu \notin r_{k}\atop\nu \notin l_{j}}\xi _{\nu }^{\ast }(t)\xi
_{\mu }(t)\frac{\sqrt{N_{r_{k}\mu }^{n-k-1}}}{\sqrt{N_{r_{k}}^{n-k}}}\frac{\sqrt{N_{l_{j}\nu }^{n-j-1}}}{\sqrt{N_{l_{j}}^{n-j}}}\hat{\pi}_{t}^{n-j-1,l_{j}\nu ;n-k-1,r_{k}\mu }(S^{\ast }XS),
\end{eqnarray*}
The innovation process $N_{t}$ is defined by ${\rm d} N_{t}={\rm d} Y^{\Lambda
}(t)-\Delta _{t}^{n;n}(I){\rm d} t$, and the initial conditions are $\hat{\pi}_{0}^{n-j,l_{j};n-k,r_{k}}(X)=\langle \Phi _{l_{j}}^{n-j}|\Phi
_{r_{k}}^{n-k}\rangle \langle \eta |X|\eta \rangle$.
\end{theorem}

{\em Proof}.
Firstly, we postulate the form of the filter as
\begin{eqnarray}\label{form_conu_mul}
{\rm d}\pi_t^{n-j,l_j;n-k,r_k}(X)=\hat{F}_t^{n-j,l_j;n-k,r_k}(X){\rm d}t+\hat{H}_t^{n-j,l_j;n-k,r_k}(X){\rm d}Y^\Lambda(t),
\end{eqnarray}
with the exact expressions of $\hat{F}_t^{n-j,l_j;n-k,r_k}(X)$ and $\hat{H}_t^{n-j,l_j;n-k,r_k}(X)$ to be determined.

Then we differentiate both sides of the equation (\ref{conditional_photon_n}), we can get that
\begin{eqnarray*}
&&\frac{|\alpha^n|^2}{(\alpha_{l_j}^{n-j})^\ast\alpha_{r_k}^{n-k}}{\rm d}\tilde{\pi}_t^\Lambda(|e^{n-j}_{l_j}\rangle\langle e_{r_k}^{n-k}|\otimes X)\nonumber\\
&=&(I\otimes {\rm d}\hat{\pi}_t^{n-j,l_j;n-k,r_k}(X))\tilde{\pi}_t^\Lambda(|e^n \rangle\langle e^n|\otimes I)+(I\otimes \hat{\pi}_t^{n-j,l_j;n-k,r_k}(X)){\rm d}\tilde{\pi}_t^\Lambda(|e^n \rangle\langle e^n|\otimes I)\nonumber\\
&&+(I\otimes {\rm d}\hat{\pi}_t^{n-j,l_j;n-k,r_k}(X)){\rm d}\tilde{\pi}_t^\Lambda(|e^n \rangle\langle e^n|\otimes I).
\end{eqnarray*}
Since for arbitrary functions $f_i(X),~i=1,2,3,4$ of the observable $X$, the following equations hold,
\begin{eqnarray*}
&&\tilde{\pi}_t^\Lambda(|e_{l_j}^{n-j}\rangle\langle e_{r_k}^{n-k}|\otimes f_1(X))\nonumber\\
&=&\frac{(\alpha_{l_j}^{n-j})^*\alpha_{r_k}^{n-k}}{|\alpha^n|^2}I\otimes \hat{\pi}_t^{n-j,l_j;n-k,r_k}(f_1(X))\tilde{\pi}_t^\Lambda(|e^n\rangle\langle e^n|\otimes I),\nonumber
\end{eqnarray*}
\begin{eqnarray*}
&&\tilde{\pi}_t^\Lambda(\mathcal{K}_{01}^n(|e_{l_j}^{n-j}\rangle\langle e_{r_k}^{n-k})|\otimes f_2(X))\nonumber\\
&=&\sum_{\mu\notin r_k}\frac{(\alpha^{n-j}_{l_j})^\ast\alpha^{n-k}_{r_k}}{|\alpha^{n}|^2}\frac{\sqrt{N_{r_k\mu}^{n-k-1}}}{\sqrt{N_{r_k}^{n-k}}} \xi_{\mu}(t)I\otimes{\hat{\pi}}_t^{n-j,l_j;n-k-1,r_k\mu}(f_2(X))\tilde{\pi}_t^\Lambda(|e^{n}\rangle\langle e^{n}|\otimes I),\nonumber
\end{eqnarray*}
\begin{eqnarray*}
&&\tilde{\pi}_t^\Lambda(\mathcal{K}_{10}^n(|e_{l_j}^{n-j}\rangle\langle e_{r_k}^{n-k}|\otimes f_3(X)))\nonumber\\
&=&\sum_{\nu\notin l_j}\frac{(\alpha^{n-j}_{l_j})^\ast\alpha^{n-k}_{r_k}}{|\alpha^{n}|^2}
\frac{\sqrt{N_{l_j\nu}^{n-j-1}}}{\sqrt{N_{l_j}^{n-j}}} \xi_{\nu}^*(t)I\otimes\hat{\pi}_t^{n-j-1,l_j\nu;n-k,r_k}(f_3(X))\tilde{\pi}_t^\Lambda(|e^{n}\rangle\langle e^{n}|\otimes I),\nonumber
\end{eqnarray*}
\begin{eqnarray*}
&&\tilde{\pi}_t^\Lambda(\mathcal{K}_{11}^n(|e_{l_j}^{n-j}\rangle\langle e_{r_k}^{n-k}|\otimes f_4(X)))(\tilde{\pi}_t^\Lambda(|e^{n}\rangle\langle e^{n}|\otimes I))^{-1}\nonumber\\
&=&\sum_{\nu\notin l_j}\sum_{\mu\notin r_k}\frac{(\alpha^{n-j}_{l_j})^\ast\alpha^{n-k}_{r_k}}
{|\alpha^{n}|^2}\frac{\sqrt{N_{l_j\nu}^{n-j-1}}}{\sqrt{N_{l_j}^{n-j}}}
\frac{\sqrt{N_{r_k\mu}^{m-k-1}}}{\sqrt{N_{r_k}^{n-k}}}  \xi_{\nu}^*(t)\xi_{\mu}(t)I\otimes\hat{\pi}_t^{n-j-1,l_j\nu;n-k-1,r_k\mu}
(f_4(X)),\nonumber
\end{eqnarray*}
we have
\begin{eqnarray*}
&&\tilde{\pi}_t^\Lambda(\mathcal{G}^n(|e_{l_j}^{n-j}\rangle\langle e_{r_k}^{n-k}|\otimes X))=\frac{(\alpha^{n-j}_{l_j})^\ast\alpha^{n-k}_{r_k}}
{|\alpha^{n}|^2}I\otimes \hat{P}_t^{n-j,l_j;n-k,r_k}(X)\tilde{\pi}_t^\Lambda(|e^n\rangle\langle e^n|\otimes I),\nonumber
\end{eqnarray*}
and
\begin{eqnarray*}
&&\tilde{Q}_t^\Lambda(|e_{l_j}^{n-j}\rangle\langle e_{r_k}^{n-k}|\otimes X)=\frac{(\alpha^{n-j}_{l_j})^\ast\alpha^{n-k}_{r_k}}
{|\alpha^{n}|^2}I\otimes \Delta_t^{n-j,l_j;n-k,r_k}(X)\tilde{\pi}_t^\Lambda(|e^n\rangle\langle e^n|\otimes I).\nonumber
\end{eqnarray*}
As a result,
\begin{eqnarray*}
&&\frac{|\alpha^{n}|^2}{(\alpha^{n-j}_{l_j})^\ast\alpha^{n-k}_{r_k}}
{\rm d}\tilde{\pi}_t^\Lambda(|e_{l_j}^{n-j}\rangle\langle e_{r_k}^{n-k}|\otimes X)\nonumber\\
&=&(I\otimes \hat{P}_t^{n-j,l_j;n-k,r_k}(X))\tilde{\pi}_t^\Lambda(|e^n\rangle\langle e^n|\otimes I){\rm d}t+[I\otimes \Delta_t^{n-j,l_j;n-k,r_k}(X)(\tilde{Q}_t^\Lambda(I\otimes I))^{-1}\nonumber\\
&&(-I\otimes \pi_t^{n-j,l_j;n-k,r_k}(X)])\tilde{\pi}_t^\Lambda(|e^n\rangle\langle e^n|\otimes I){\rm d}\tilde{N}^\Lambda(t),\nonumber
\end{eqnarray*}
\begin{eqnarray*}
&&(I\otimes \hat{\pi}_t^{n-j,l_j;n-k,r_k}(X)){\rm d}\tilde{\pi}_t^\Lambda(|e^n\rangle\langle e^n|\otimes I)\nonumber\\
&=&(I\otimes \hat{\pi}_t^{n-j,l_j;n-k,r_k}(X))[I\otimes \Delta_t^{n;n}(I)(\tilde{Q}_t^\Lambda(I\otimes I))^{-1}-I]\tilde{\pi}_t^\Lambda(|e^n\rangle\langle e^n|\otimes I){\rm d}\tilde{N}^\Lambda(t),\nonumber
\end{eqnarray*}
\begin{eqnarray*}
&&(I\otimes {\rm d}\hat{\pi}_t^{n-j,l_j;n-k,r_k}(X)){\rm d}\tilde{\pi}_t^\Lambda(|e^n\rangle\langle e^n|\otimes I)\nonumber\\
&=&(I\otimes \hat{H}_t^{n-j,l_j;n-k,r_k}(X))[I\otimes \Delta_t^{n;n}(I)(\tilde{Q}_t^\Lambda(I\otimes I))^{-1}-I]\tilde{\pi}_t^\Lambda(|e^n\rangle\langle e^n|\otimes I){\rm d}Y^\Lambda(t).\nonumber
\end{eqnarray*}
Together with ${\rm d}\tilde{N}^\Lambda(t)={\rm d}Y^\Lambda(t)-\tilde{Q}_t^{\Lambda}(I\otimes I){\rm d}t$, we can get two equations about the coefficients of ${\rm d}t$ and ${\rm d}Y^\Lambda(t)$.

Finally, the expressions of $\hat{F}_t^{n-j,l_j;n-k,r_k}(X)$ and $\hat{H}_t^{n-j,l_j;n-k,r_k}(X)$ can be separately calculated as 
\begin{eqnarray*}
&&\hat{H}_t^{n-j,l_j;n-k,r_k}(X)=(\Delta^{n;n}(I))^{-1}\Delta^{n-j,l_j;n-k,r_k}(X)-\hat{\pi}_t^{n-j,l_j;n-k,r_k}(X),\nonumber\\
&&\hat{F}_t^{n-j,l_j;n-k,r_k}(X)=\hat{P}^{n-j,l_j;n-k,r_k}(X)-\hat{H}_t^{n-j,l_j;n-k,r_k}(X)\Delta^{n;n}(I),
\end{eqnarray*}
by solving the equations. Putting them back into equation (\ref{form_conu_mul}), we can get the quantum filter.
\endproof

{\it Remark }12. The single-photon filter for photodetection studied in \cite[Section 3-F]{gough:2012b} is a special case of  the general $n$-photon filter presented in Theorem \ref{thm:n_filter_PD} when $n=1$.

\vspace{-0.1cm}

\section{Conclusion}\label{sec:conclusion}
In this paper we have investigated the filtering problem for an arbitrary open quantum system driven by an incident wavepacket prepared in a continuous-mode multi-photon state.  The two-photon case has been studied in detail. Moreover, for the general multi-photon case, the filtering equations for both homodyne detection and photodetection have been proposed. A model of a two-level system driven by a two-photon wavepacket has been used to demonstrate some of the results in the paper. This example reveals physical features of the two-photon case distinct from the single-photon case and the Fock state case. Future research includes measurement-based feedback control by means of the multi-photon filtering framework presented here.

%

{\bf Acknowledgment.} The authors are grateful to Matthew James and Hendra Nurdin for their very helpful discussions.


\Appendix\section*{} The proof of Theorem \ref{thm:2-photon_filter} proceeds along the following three steps.
\begin{description}
\item[Step 1.] Express the filtering equations in Theorem  \ref{thm:2-photon_filter} in a unified manner in terms of the conditional expectations  $\pi _{t}^{jk;mn}(X)$  defined in equation (\ref{exp_con_2}).
\item[Step 2.] Postulate the general form of the filtering equation of $\pi _{t}^{jk;mn}(X)$.
\item[Step 3.] Derive the exact expression of the quantum filter postulated in Step 2.
\end{description}

In what follows we work out the detail for each step.
%
%

{\it Step 1.} For given $j,k,m,n=0,1$, define superoperators ${\mathcal{T}_{t}^{jk;mn}(X)}$ to be
\begin{eqnarray*}
&&{\mathcal{T}_{t}^{jk;mn}(X)}
\\
&\triangleq&\pi _{t}^{jk;mn}(\mathcal{L}_{00}(X))
\nonumber
\\
 &&
+[\delta _{j0}\delta _{k1} ~\xi _{2}^{\ast }(t)
  +\delta _{j1}\delta _{k0} ~\xi _{1}^{\ast }(t)] \pi _{t}^{00;mn}(\mathcal{L}_{10}(X))
\nonumber
\\
&&
+[\delta _{m0} \delta _{n1}~\xi _{2}(t) +\delta _{m1}\delta _{n0} ~\xi _{1}(t)]\pi _{t}^{jk;00}(\mathcal{L}_{01}(X))
 \nonumber
 \\
&&
+\delta _{j1}\delta _{k1}\frac{1}{\sqrt{N_2}}[\xi _{1}^{\ast }(t)\pi
_{t}^{01;mn}(\mathcal{L}_{10}(X))+\xi
_{2}^{\ast }(t)\pi _{t}^{10;mn}(\mathcal{L}_{10}(X))]
\nonumber
\\
 &&
+\delta _{m1}\delta _{n1}\frac{1}{\sqrt{N_2}}[\xi _{1}(t)\pi _{t}^{jk;01}(\mathcal{L}_{01}(X))+\xi _{2}(t)\pi _{t}^{jk;10}(\mathcal{L}_{01}(X))]  \nonumber
\\
&&+\delta _{j1}\delta _{k1}\delta _{m1}\delta _{n1}\frac{1}{N_2}[\vert \xi
_1(t)\vert^2 \pi _{t}^{01;01}\left( \mathcal{L}_{11}(X)\right) +\xi _{1}(t)\xi
_{2}^{\ast }(t) \pi _{t}^{10;01}\left( \mathcal{L}_{11}(X)\right)]
   \nonumber
\\
&&
+\delta _{j1}\delta _{k1}\delta _{m1}\delta _{n1}\frac{1}{N_2}[\xi _{1}^{\ast}(t)\xi _{2}(t)\pi _{t}^{01;10}\left( \mathcal{L}_{11}(X)\right)
+\vert \xi_{2}(t)\vert^2\pi _{t}^{10;10}\left( \mathcal{L}_{11}(X)\right) ]  \nonumber
\\
&&
+\delta _{j0}\delta _{k1}\delta _{m1}\delta _{n1}\frac{1}{\sqrt{N_2}} [\xi_1(t)\xi
_2^{\ast}(t) \pi_t^{00;01} (\mathcal{L}_{11}(X)) +\vert \xi_2(t)\vert^2 \pi _{t}^{00;10}(\mathcal{L}_{11}(X))]
 \nonumber
 \\
&&
+\delta _{j1}\delta _{k0}\delta _{m1}\delta _{n1}\frac{1}{\sqrt{N2}}[ \vert \xi
_{1}(t)\vert^2 \pi _{t}^{00;01}\left( \mathcal{L}_{11}(X)\right) ) +\xi_1^{\ast
}(t)\xi_2(t)\pi _{t}^{00;10}\left( \mathcal{L}_{11}(X)\right)]
  \nonumber
\\
&&
+\delta _{j1}\delta _{k1}\delta _{m0}\delta _{n1}\frac{1}{\sqrt{N_2}} [\xi_1^{\ast}(t)\xi_2(t) \pi _{t}^{01;00}\left( \mathcal{L}_{11}(X)\right) + \vert \xi_2(t)\vert ^2 \pi _{t}^{10;00}\left(\mathcal{L}_{11}(X)\right)]
\nonumber
\\
&&
+ \delta _{j1}\delta _{k1}\delta _{m1}\delta _{n0}\frac{1}{\sqrt{N_2}}
[\vert \xi_1(t)\vert^2 \pi _{t}^{01;00}\left( \mathcal{L}_{11}(X)\right)  + \xi_1(t)\xi
_2^{\ast }(t)\pi _{t}^{10;00}\left( \mathcal{L}_{11}(X)\right) ]
 \nonumber
 \\
 &&
+\delta _{j0}\delta
_{k1} [\delta _{m0}\delta _{n1}\left\vert \xi _{2}(t)\right\vert ^{2}+\delta _{m1}\delta _{n0}\xi _{1}(t)\xi _{2}^{\ast }(t)]\pi_{t}^{00;00}\left( \mathcal{L}_{11}(X)\right)
\nonumber
\\
&&
+\delta _{j1}\delta _{k0}[\delta _{m0}\delta _{n1}\xi _{1}^{\ast }(t)\xi
_{2}(t)+\delta _{m1}\delta _{n0}|\xi _{1}(t)|^{2}]\pi
_{t}^{00;00}\left( \mathcal{L}_{11}(X)\right).
\nonumber
\label{eq:aug10_1}
\end{eqnarray*}
Then it can be verified that for all $j,k,m,n=0,1$ the equations in Theorem \ref{thm:2-photon_filter} can be re-written in a unified way as
\begin{equation}
{\rm d}\pi _{t}^{jk;mn}(X)={\mathcal{T}_{t}^{jk;mn}(X){\rm d} t+}\left\{
M_{t}^{jk;mn}(X)-\pi _{t}^{jk;mn}(X){M_{t}^{11;11}(I)}\right\}
{\rm d} W(t), \label{eq:aug10_2}
\end{equation}
where the superoperators $M_{t}^{jk;mn}(X)$ are given in equation (\ref{eq:aug3_6}). As a result, it suffices to establish (\ref{eq:aug10_2}) to prove the Theorem \ref{thm:2-photon_filter}.

{\it Step 2.} We postulate the filtering equation of $\pi _{t}^{jk;mn}(X)$ to be
\begin{equation}
{\rm d}\pi _{t}^{jk;mn}(X)=F_{t}^{jk;mn}(X){\rm d} t+H_{t}^{jk;mn}(X){\rm d} Y(t), ~~\forall ~ j,k,m,n=0,1.
\label{eq:Aug2_0}
\end{equation}
The expressions for $F_{t}^{jk;mn}(X)$ and $H_{t}^{jk;mn}(X)$ in equation (\ref{eq:Aug2_0}) are to be determined in the next step.

{\it Step 3.} For the sake of clarity, we re-write equation (\ref{exp_con_2}) as below
\begin{equation}
 \frac{|\alpha _{11}|^{2}}{\alpha _{jk}^{\ast }\alpha _{mn}}\tilde{\pi}_{t}(|e_{jk}\rangle \langle e_{mn}|\otimes X)=(I\otimes \pi _{t}^{jk;mn}(X))
\tilde{\pi}_{t}(|e_{11}\rangle \langle e_{11}|\otimes I).
\label{exp_con_3_temp}
\end{equation}
 Differentiating the both sides of equation (\ref{exp_con_3_temp}) and comparing corresponding terms we have
\begin{eqnarray}
&&\frac{|\alpha _{11}|^{2}}{\alpha _{jk}^{\ast }\alpha _{mn}}
\tilde{M}_{t}(|e_{jk}\rangle \langle e_{mn}|\otimes X)\label{aug02_7}\\
&=& I\otimes H_{t}^{jk;mn}(X)\tilde{\pi}_{t}(|e_{11}\rangle \langle e_{11}|\otimes
I)+(I\otimes \pi _{t}^{jk;mn}(X))\tilde{M}_{t}(|e_{11}\rangle \langle
e_{11}|\otimes I), \nonumber
\end{eqnarray}
and
\begin{eqnarray}
&&~~~~~\frac{|\alpha _{11}|^{2}}{\alpha _{jk}^{\ast }\alpha _{mn}}\tilde{\pi}
_{t}(\mathcal{G}(|e_{jk}\rangle \langle e_{mn}|\otimes X))-\frac{|\alpha
_{11}|^{2}}{\alpha _{jk}^{\ast }\alpha _{mn}}\tilde{M}_{t}(|e_{jk}\rangle
\langle e_{mn}|\otimes X)\tilde{M}_{t}(I\otimes I) \label{aug02_6}\\
&=&I\otimes F_{t}^{jk;mn}(X)\tilde{\pi}_{t}(|e_{11}\rangle \langle e_{11}|\otimes I)-(I\otimes \pi _{t}^{jk;mn}(X))\tilde{M}_{t}(|e_{11}\rangle \langle
e_{11}|\otimes I)\tilde{M}_{t}(I\otimes I)
\nonumber\\
&&+I\otimes H_{t}^{jk;mn}(X)\left[ \tilde{M}
_{t}(|e_{11}\rangle \langle e_{11}|\otimes I)-\tilde{\pi}_{t}(|e_{11}\rangle
\langle e_{11}|\otimes I)\tilde{M}_{t}(I\otimes I)\right]. \nonumber
\end{eqnarray}
By means of (\ref{exp_con_3_temp}), the definitions of ${\tilde{M}
_{t}(A\otimes X)}$ in (\ref{eq:tilde_M}) and ${M_{t}^{jk;mn}(X)}$ in (\ref{eq:aug3_6}), we are able to derive
\begin{equation}
\frac{|\alpha _{11}|^{2}}{\alpha _{jk}^{\ast }\alpha _{mn}}\tilde{M}
_{t}(|e_{jk}\rangle \langle e_{mn}|\otimes X)=\left( I\otimes {
M_{t}^{jk;mn}(X)}\right) \tilde{\pi}_{t}(|e_{11}\rangle \langle
e_{11}|\otimes I).  \label{eq:aug18_1}
\end{equation}
Putting (\ref{eq:aug18_1}) back into (\ref{aug02_7}) yields
\begin{equation}
H_{t}^{jk;mn}(X)={M_{t}^{jk;mn}(X)-}\pi _{t}^{jk;mn}(X){M_{t}^{11;11}(I)}.
\label{eq:aug18_2}
\end{equation}
That is, we have derived the expression for $H_{t}^{jk;mn}(X)$. Next we derive the expression for $F_{t}^{jk;mn}(X)$.
Substituting (\ref{eq:aug18_1})--(\ref{eq:aug18_2}) into (\ref{aug02_6}) yields
\begin{eqnarray}
&&~~~~ I\otimes F_{t}^{jk;mn}(X)\tilde{\pi}_{t}(|e_{11}\rangle \langle e_{11}|\otimes I)\label{aug02_9}\\
&=&\frac{|\alpha _{11}|^{2}}{\alpha _{jk}^{\ast }\alpha _{mn}}\tilde{\pi}
_{t}(\mathcal{G}(|e_{jk}\rangle \langle e_{mn}|\otimes X))-\left( I\otimes
H_{t}^{jk;mn}(X){M_{t}^{11;11}(I)}\right) \tilde{\pi}_{t}(|e_{11}\rangle
\langle e_{11}|\otimes I). \nonumber
\end{eqnarray}
Thus we have to find expression for $\tilde{\pi}_{t}(\mathcal{G}
(|e_{jk}\rangle \langle e_{mn}|\otimes X))$. Observe that by (\ref{eq:G}),
\begin{eqnarray}
&&~~~~~~~ \frac{|\alpha _{11}|^{2}}{\alpha _{jk}^{\ast }\alpha _{mn}}\tilde{\pi}_{t}(
\mathcal{G}(|e_{jk}\rangle \langle e_{mn}|\otimes X))  \label{eq:aug3_3} \\
&=&\frac{|\alpha _{11}|^{2}}{\alpha _{jk}^{\ast }\alpha _{mn}}\tilde{\pi}
_{t}(|e_{jk}\rangle \langle e_{mn}|\otimes \mathcal{L}_{00}(X))  +\frac{|\alpha _{11}|^{2}}{\alpha _{jk}^{\ast }\alpha _{mn}}\tilde{\pi}
_{t}(\mathcal{K}_{01}(|e_{jk}\rangle \langle e_{mn}|)\otimes \mathcal{L}
_{01}(X))  \nonumber \\
&&+\frac{|\alpha _{11}|^{2}}{\alpha _{jk}^{\ast }\alpha _{mn}}\tilde{\pi}
_{t}(\mathcal{K}_{10}(|e_{jk}\rangle \langle e_{mn}|)\otimes \mathcal{L}
_{10}(X))  +\frac{|\alpha _{11}|^{2}}{\alpha _{jk}^{\ast }\alpha _{mn}}\tilde{\pi}
_{t}(\mathcal{K}_{11}(|e_{jk}\rangle \langle e_{mn}|)\otimes \mathcal{L}
_{11}(X)), \nonumber
\end{eqnarray}
we need to calculate each term on the right-hand side of equation (\ref{eq:aug3_3}).
Firstly,
\begin{eqnarray}
\quad\quad\frac{|\alpha _{11}|^{2}}{\alpha _{jk}^{\ast }\alpha _{mn}}\tilde{\pi}
_{t}(|e_{jk}\rangle \langle e_{mn}|\otimes \mathcal{L}_{00}(X))\tilde{\pi}_{t}(|e_{11}\rangle \langle e_{11}|\otimes I)^{-1}=
I\otimes \pi _{t}^{jk;mn}(\mathcal{L}_{00}(X)).
\label{eq:aug10:5}
\end{eqnarray}
Secondly, based on the definition of the superoperator $\mathcal{K}_{01}(A)$
in (\ref{eq_K_01}) we have
\begin{eqnarray}
&&\frac{|\alpha _{11}|^{2}}{\alpha _{jk}^{\ast }\alpha _{mn}}\tilde{\pi}_{t}(
\mathcal{K}_{01}(|e_{jk}\rangle \langle e_{mn}|)\otimes \mathcal{L}_{01}(X))\tilde{\pi}_{t}(|e_{11}\rangle \langle e_{11}|\otimes I)^{-1}
\label{eq:aug10:6} \\
&=&\delta _{m1}\delta _{n1}\frac{\xi _{1}(t)}{\sqrt{N_{2}}}(I\otimes \pi
_{t}^{jk;01}(\mathcal{L}_{01}(X))+\delta _{m1}\delta _{n1}\frac{\xi _{2}(t)}{\sqrt{N_{2}}}
(I\otimes \pi _{t}^{jk;10}(\mathcal{L}_{01}(X)) \nonumber \\
&&+\delta _{m1}\delta _{n0}\xi _{1}(t)(I\otimes \pi _{t}^{jk;00}(\mathcal{L}
_{01}(X))+\delta
_{m0}\delta _{n1} \xi _{2}(t)(I\otimes \pi _{t}^{jk;00}(\mathcal{L}
_{01}(X)).  \nonumber
\end{eqnarray}
Thirdly, according to the definition of the superoperator $\mathcal{K}_{10}(A)$ in
(\ref{eq_K_10}), we have
\begin{eqnarray}
&&\frac{|\alpha _{11}|^{2}}{\alpha _{jk}^{\ast }\alpha _{mn}}\tilde{\pi}_{t}(
\mathcal{K}_{10}(|e_{jk}\rangle \langle e_{mn}|)\otimes \mathcal{L}_{10}(X))
\tilde{\pi}_{t}(|e_{11}\rangle \langle e_{11}|\otimes I)^{-1}\label{eq:aug10:7}
 \\
&=&\delta _{j1}\delta _{k1}\frac{\xi _{1}^{\ast }(t)}{\sqrt{N_{2}}}(I\otimes
\pi _{t}^{01;mn}(\mathcal{L}_{10}(X))+\delta _{j1}\delta _{k1}\frac{\xi
_{2}^{\ast }(t)}{\sqrt{N_{2}}}(I\otimes \pi _{t}^{10;mn}(\mathcal{L}_{10}(X))
\nonumber \\
&&+\delta _{j1}\delta _{k0}\xi _{1}^{\ast }(t)(I\otimes \pi _{t}^{00;mn}(
\mathcal{L}_{10}(X))+\delta _{j0}\delta _{k1}\xi _{2}^{\ast }(t)(I\otimes
\pi _{t}^{00;mn}(\mathcal{L}_{10}(X)).  \nonumber
\end{eqnarray}
Fourthly, by the definition of the superoperator $\mathcal{K}_{11}(A)$ in (\ref{eq_K_11}), it can be shown that
\begin{eqnarray}
&&\frac{|\alpha _{11}|^{2}}{\alpha _{jk}^{\ast }\alpha _{mn}} \tilde{
\pi}_{t}(\mathcal{K}_{11}(|e_{jk}\rangle \langle e_{mn}|)\otimes \mathcal{L}
_{11}(X))\tilde{\pi}_{t}(|e_{11}\rangle \langle e_{11}|\otimes I\mathcal{)}
^{-1}\nonumber\\
&=&\delta _{j1}\delta _{k1}\delta _{m1}\delta _{n1}\frac{\left\vert \xi
_{1}(t)\right\vert ^{2}}{N_{2}}I\otimes \pi _{t}^{01;01}\left( \mathcal{L}
_{11}(X)\right) +\delta _{j1}\delta _{k1}\delta _{m1}\delta _{n1}\frac{\xi_{1}(t)\xi
_{2}^{\ast }(t)}{N_{2}}I\otimes \pi _{t}^{10;01}\left( \mathcal{L}
_{11}(X)\right)   \nonumber \\
&&+\delta _{j1}\delta _{k0}\delta _{m1}\delta _{n1}\frac{\left\vert \xi
_{1}(t)\right\vert ^{2}}{\sqrt{N_{2}}}I\otimes \pi _{t}^{00;01}\left(
\mathcal{L}_{11}(X)\right) +\delta _{j0}\delta _{k1}\delta _{m1}\delta _{n1}
\frac{\xi _{1}(t)\xi _{2}^{\ast }(t)}{\sqrt{N_{2}}}I\otimes \pi
_{t}^{00;01}\left( \mathcal{L}_{11}(X)\right)
\nonumber \\
&&+\delta _{j1}\delta _{k1}\delta _{m1}\delta _{n1}\frac{\xi _{1}^{\ast
}(t)\xi _{2}(t)}{N_{2}}I\otimes \pi _{t}^{01;10}\left( \mathcal{L}
_{11}(X)\right) +\delta _{j1}\delta _{k1}\delta _{m1}\delta _{n1}\frac{
\left\vert \xi _{2}(t)\right\vert ^{2}}{N_{2}}I\otimes \pi
_{t}^{10;10}\left( \mathcal{L}_{11}(X)\right)
\nonumber \\
&&+\delta _{j1}\delta _{k0}\delta _{m1}\delta _{n1}\frac{\xi _{1}^{\ast
}(t)\xi _{2}(t)}{\sqrt{N_{2}}}I\otimes \pi _{t}^{00;10}\left( \mathcal{L}
_{11}(X)\right) +\delta _{j0}\delta _{k1}\delta _{m1}\delta _{n1}\frac{
\left\vert \xi _{2}(t)\right\vert ^{2}}{\sqrt{N_{2}}}I\otimes \pi
_{t}^{00;10}\left( \mathcal{L}_{11}(X)\right)
 \nonumber \\
&&+\delta _{j1}\delta _{k1}\delta _{m1}\delta _{n0}\frac{\left\vert \xi
_{1}(t)\right\vert ^{2}}{\sqrt{N_{2}}}I\otimes \pi _{t}^{01;00}\left(
\mathcal{L}_{11}(X)\right) +\delta _{j1}\delta _{k1}\delta _{m1}\delta _{n0}
\frac{\xi _{1}(t)\xi _{2}^{\ast }(t)}{\sqrt{N_{2}}}I\otimes \pi
_{t}^{10;00}\left( \mathcal{L}
_{11}(X)\right)
\nonumber \\
&&+\delta _{j1}\delta _{k0}\delta _{m1}\delta _{n0}|\xi _{1}(t)|^{2}I\otimes
\pi _{t}^{00;00}\left( \mathcal{L}_{11}(X)\right) +\delta _{j0}\delta
_{k1}\delta _{m1}\delta _{n0}\xi _{1}(t)\xi _{2}^{\ast }(t)I\otimes \pi
_{t}^{00;00}\left( \mathcal{L}_{11}(X)\right)   \nonumber \\
&&+\delta _{j1}\delta _{k1}\delta _{m0}\delta _{n1}\frac{\xi _{1}^{\ast
}(t)\xi _{2}(t)}{\sqrt{N_{2}}}I\otimes \pi _{t}^{01;00}\left( \mathcal{L}
_{11}(X)\right) +\delta _{j1}\delta _{k1}\delta _{m0}\delta _{n1}\frac{
\left\vert \xi _{2}(t)\right\vert ^{2}}{\sqrt{N_{2}}}I\otimes \pi
_{t}^{10;00}\left( \mathcal{L}_{11}(X)\right)   \nonumber \\
&&+\delta _{j1}\delta _{k0}\delta _{m0}\delta _{n1}\xi _{1}^{\ast }(t)\xi
_{2}(t)I\otimes \pi _{t}^{00;00}\left( \mathcal{L}_{11}(X)\right) +\delta
_{j0}\delta _{k1}\delta _{m0}\delta _{n1}\left\vert \xi _{2}(t)\right\vert
^{2}I\otimes \pi _{t}^{00;00}\left(\mathcal{L}_{11}(X)\right).\nonumber\\
\label{eq:aug10:9}
\end{eqnarray}
Finally, on substitution of (\ref{eq:aug10:5})--(\ref
{eq:aug10:9}) into equation (\ref{eq:aug3_3}), we have
\begin{equation}
\frac{|\alpha _{11}|^{2}}{\alpha _{jk}^{\ast }\alpha _{mn}}\tilde{\pi}_{t}(
\mathcal{G}(|e_{jk}\rangle \langle e_{mn}|\otimes X))=(I\otimes \mathcal{T}_t^{jk;mn}(X))\tilde{\pi}_t(|e_{11}\rangle\langle e_{11}|\otimes I),  \label{eq:aug10:10}
\end{equation}
where $\mathcal{T}_t^{jk;mn}(X)$ is that defined in (\ref{eq:aug10_1}). Substituting (\ref{eq:aug10:10}) into
(\ref{aug02_9}) gives
\begin{equation}
F_{t}^{jk;mn}(X)=\mathcal{T}_{t}^{jk;mn}(X)-H_{t}^{jk;mn}(I){
M_{t}^{11;11}(I)}.  \label{eq:aug18_3}
\end{equation}
That is, we have derived the expression for $F_{t}^{jk;mn}(X)$.
Putting $H_t^{jk;mn}(X)$ in (\ref{eq:aug18_2}) and $F_{t}^{jk;mn}(X)$ in (\ref{eq:aug18_3}) back into equation (\ref{eq:Aug2_0}) yields, for all $j,k,m,n=0,1$,
\begin{eqnarray*}
{\rm d}\pi _{t}^{jk;mn}(X)=\mathcal{T}_{t}^{jk;mn}(X){\rm d} t+\left[ M_{t}^{jk;mn}(X)-\pi _{t}^{jk;mn}(X){
M_{t}^{11;11}(I)}\right] {\rm d} W(t), \nonumber\label{eq:aug3_10}
\end{eqnarray*}
which is exactly (\ref{eq:aug10_2}). The proof is completed.

{\it Remark }13.
It is worth noting that the coefficients $\alpha_{jk}$\ ($j,k=0,1$) in the superposition
state $|\Sigma\r$ in equation (\ref{eq:sigma_2_photon}) for the extended system studied in Subsection \ref{subsec:extended_system}  do not appear in the filtering equations in Theorem \ref{thm:2-photon_filter}. This can be seen clearly from the above proof. Specifically, the unifying filtering equation (\ref{eq:Aug2_0}) depends on two operators $F_{t}^{jk;mn}(X)$ and $H_{t}^{jk;mn}(X)$, which satisfy two coupled algebraic equations (\ref{aug02_7})--(\ref{aug02_6}). Both these equations contain the coefficients $\alpha_{jk}$. However, by (\ref{eq:aug18_1}), $\alpha_{jk}$ on the left-hand side of equation (\ref{aug02_7}) disappear, so $H_{t}^{jk;mn}(X)$ does not depend on $\alpha_{jk}$, cf. (\ref{eq:aug18_2}). Similarly, by (\ref{eq:aug10:10}) as well as (\ref{eq:aug18_1}), $\alpha_{jk}$ on the left-hand side of equation (\ref{aug02_6}) disappear too. As a result, $F_{t}^{jk;mn}(X)$  does not depend on $\alpha_{jk}$ either, cf. (\ref{eq:aug18_3}). Therefore, the coefficients $\alpha_{jk}$ do not appear in the unifying filtering equation (\ref{eq:Aug2_0}), or equivalently the filtering equations in Theorem \ref{thm:2-photon_filter}.


\end{document}